\documentclass[a4paper,12pt]{article}
\pdfoutput=1 % if your are submitting a pdflatex (i.e.\ if you have
\PassOptionsToPackage{linecolor=blue,backgroundcolor=blue!25,bordercolor=blue,textsize=scriptsize}{todonotes}
\PassOptionsToPackage{dvipsnames}{xcolor}

\usepackage{jheppub} % for details on the use of the package, please
\usepackage{float, array, xspace, amscd, amsthm}
\usepackage{fancyhdr}
\usepackage{longtable}

\usepackage{amsmath}
\usepackage{slashed}
\usepackage{microtype}

\allowdisplaybreaks

\makeatletter
\g@addto@macro\bfseries{\boldmath}
\makeatother

\let\originalleft\left
\let\originalright\right
\renewcommand{\left}{\mathopen{}\mathclose\bgroup\originalleft}
\renewcommand{\right}{\aftergroup\egroup\originalright}

\usepackage{amsthm,amssymb,euscript,bbold}
\usepackage{amsfonts}
\usepackage{mathrsfs}
\usepackage{graphicx}
\usepackage{mathtools}
\usepackage{slashed}
\usepackage{amsthm}
\usepackage{amscd}
\usepackage{tensor}

\usepackage[shortlabels]{enumitem}
\usepackage{epsfig}                    % ps figures
\usepackage[matrix,arrow]{xy}          % commutative diagrams
\usepackage{xspace}                    % bbfont and space in abbreviations
\usepackage{stmaryrd}                  % extra symbol font (dbl bracket)
\usepackage{slashed}
\usepackage{appendix,graphicx}
\usepackage{enumitem}
\usepackage{hyperref}

\usepackage{units} % nice diagonal fractions
\usepackage{tikz}
\usepackage{tikz-cd}
\usepackage{color}

\usepackage{todonotes}

\usepackage{tikz-cd}

\DeclareSymbolFont{bbold}{U}{bbold}{m}{n}
\DeclareSymbolFontAlphabet{\mathbbold}{bbold}

\newtheorem{lemma}{Lemma}

\newcommand{\dd}{{\rm d}}

\newcommand{\End}{{\rm End}}
\newcommand{\tr}{{\rm tr}}

\newcommand{\cN}{\mathcal{N}}
\newcommand{\cA}{\mathcal{A}}

\def\bea#1\eea{\begin{align}#1\end{align}}

\usepackage{color}

\theoremstyle{definition}

\global\long\def\ext{\Lambda}%

\newcommand{\ii}{\mathrm{i}}
\newcommand{\ee}{\mathrm{e}}

\newcommand{\der}{\partial}
\newcommand{\del}{\partial}
\newcommand{\delb}{\bar{\partial}}

\newcommand{\bbR}{\mathbb{R}}
\newcommand{\bbC}{\mathbb{C}}

\DeclareMathOperator{\SU}{\mathit{SU}}

\DeclareMathOperator{\SO}{\mathit{SO}}

\DeclareMathOperator{\GL}{\mathit{GL}}

\DeclareMathOperator{\Spin}{\mathit{Spin}}
\DeclareMathOperator{\G}{\mathit{G}}

\DeclareMathOperator{\Cliff}{Cliff}

\newcommand{\rep}[1]{\mathbf{#1}}

\newcommand{\repp}[2]{(\rep{#1}, \rep{#2})}
\newcommand{\id}{\mathbbold{1}}

\DeclareMathOperator{\im}{Im}

\DeclareMathOperator{\adj}{ad}

\newcommand{\LC}{\nabla}

\newcommand{\dc}{\check{\dd}} 
\newcommand{\dcDagger}{\dc^\dagger}

\newcommand{\proj}[1]{\times_{#1}}

\newcommand{\inn}{\mathbin{\lrcorner}}

\newcommand{\on}{\mathbin{\otimes_{N}}}

\newcommand{\ba}{\bar{a}}
\newcommand{\bb}{\bar{b}}
\newcommand{\bc}{\bar{c}}

\newcommand{\be}{\bar{e}}

\newcommand{\he}{\hat{e}}
\newcommand{\hE}{\hat{E}}

\newcommand{\hs}[1]{\hspace{#1}}

\newcommand{\ra}{\rightarrow}
\newcommand{\Ra}{\Rightarrow}
\newcommand{\lra}{\leftrightarrow}

\newcommand{\tm}{{\tilde{m}}}
\newcommand{\tn}{{\tilde{n}}}

\newcommand{\dil}{\phi}

\newcommand{\complex}{BPS complex\hs{5pt}}
\newcommand{\complexes}{BPS complexes\hs{5pt}}
\newcommand{\complexNoSpace}{BPS complex}
\newcommand{\complexesNoSpace}{BPS complexes}

\newcommand{\mc}{\mathcal} 
\newcommand{\mf}{\mathfrak} 
\newcommand{\mb}{\mathbb} 
\renewcommand{\on}{\operatorname}

\newcommand{\slot}{\;\cdot\;} 

\newcommand{\curv}{R}
\newcommand{\cx}{\mc A} %'generalised dolbeaut CompleX'
\newcommand{\bx}{A} %'Bundle for the generalised dolbeaut compleX'
\newcommand{\cy}{\mc B} %'ComplementarY subcomplex'
\newcommand{\by}{B} %'Bundle underlying the complementarY subcomplex'
\newcommand{\gc}{\mc F} %the full \Gamma(\Lambda^\bullet E)
\newcommand{\gb}{F} %the full \Lambda^\bullet E

\newcommand{\Riem}{R} 
\newcommand{\Dirac}{\mc{D}}

\newcommand{\gTwoForm}{\phi}

\newcommand{\diracOperator}{\slashed{\LC}}
\newcommand{\gTwoEvenSpinor}[1]{\psi(#1_{0}, #1_{2})}
\newcommand{\gTwoOddSpinor}[1]{\chi(#1_{1}, #1_3)}
\newcommand{\gTwoComplexSpinor}{\Psi}
\newcommand{\diracBar}[1]{\overline{#1}}
\newcommand{\gTwoInvariantSpinor}{\epsilon}
\newcommand{\formInnerProduct}[2]{\langle #1, #2\rangle}
\newcommand{\gTwoProjector}[2]{\mathcal{P}^{#1}_{\textbf{#2}}}
\newcommand{\covariantExteriorDerivative}{\dd_\LC}
\newcommand{\contraction}{\lrcorner}

\title{BPS complexes and Chern--Simons theories from $G$-structures in gauge theory and gravity}
\author[a]{Julian Kupka,}
\emailAdd{j.kupka@herts.ac.uk}
\author[a]{Charles Strickland-Constable,}
\emailAdd{c.strickland-constable@herts.ac.uk}
\author[b]{Eirik Eik Svanes,}
\emailAdd{eirik.e.svanes@uis.no}
\author[c]{\\David Tennyson,}
\emailAdd{dtennyson@tamu.edu}
\author[a]{and Fridrich Valach}
\emailAdd{f.valach@herts.ac.uk}

\affiliation[a]{Department of Physics, Astronomy and Mathematics,
University of Hertfordshire, College Lane, Hatfield, AL10 9AB, United Kingdom}

\affiliation[b]{Department of Mathematics and Physics, Faculty of Science and Technology,
University of Stavanger, N-4036, Stavanger, Norway}

\affiliation[c]{Mitchell Institute for Fundamental Physics and Astronomy, Texas A\&M University, College Station, TX, 77843, USA}

\abstract{
We consider a variety of physical systems in which one has states that can be thought of as generalised instantons.
These include Yang--Mills theories on manifolds with a torsion-free $G$-structure, analogous gravitational instantons and certain supersymmetric solutions of ten-dimensional supergravity, using their formulation as generalised $G$-structures on Courant algebroids. 
We provide a universal algebraic construction of a complex, which we call the BPS complex, that computes the infinitesimal moduli space of the instanton as one of its cohomologies. 
We call a class of these spinor type complexes, which are closely connected to supersymmetric systems, and show how their Laplacians have nice properties. 
In the supergravity context, the BPS complex becomes a double complex, in a way that corresponds to the left- and right-moving sectors of the string, and becomes much like the double complex of $(p,q)$-forms on a K\"ahler manifold. 
If the BPS complex has a symplectic inner product, one can write down an associated linearised BV Chern--Simons theory, which reproduces several classic examples in gauge theory. 
We discuss applications to (quasi-)topological string theories and heterotic superpotential functionals, whose quadratic parts can also be constructed naturally from the BPS complex.}

\subheader{\hfill\textrm{MI-HET-834}}

\begin{document}

\maketitle
% \flushbottom

\section{Introduction}

When studying a quantum field theory or string theory, it is often convenient to study a simpler `twisted' theory first which, in nice cases, teaches us something about the full theory.
In the past, this strategy has been used to great effect and has led to interesting connections with other topics in geometry and topology.
For example, topological quantum field theories~\cite{Witten:1990} are simple models of quantum field theory in that they have no propagating degrees of freedom and they were famously related to topological invariants of the underlying manifolds in a variety of situations~\cite{Witten:1990,Witten:1988hf,Witten:1988xj}. 
The topological versions of string $\sigma$-models~\cite{Witten:1988xj,Vafa:1990mu} also gave rise to topological string theories~\cite{Witten:1991zz,Bershadsky:1993cx}, which are subsectors of the full string theory where exact computations can be performed. 
In fact, these computations were shown to give results that remain valid in the full theory~\cite{Bershadsky:1993cx,Antoniadis:1993ze}. 
A further interesting class of simplified QFTs are holomorphic field theories, as studied in~\cite{nekrassov1996four,Popov:1999cq,Costello:2011np,Williams:2018ows}, which are simpler than full smooth QFTs but have considerably more structure than topological field theories. 

Many of these theories arise from taking a supersymmetric theory and applying a twist~\cite{Witten:1990}. 
This is a two-step process in which one first modifies the Lorentz symmetry of the theory, and correspondingly the energy-momentum tensor, so that one of the spinor supercharges becomes a scalar that can be defined globally regardless of the particular metric on the manifold. 
One then adds this supercharge to the BRST operator of the theory, which significantly modifies the spectrum of physical states in general. 
In the case that the energy-momentum tensor becomes BRST exact, the correlation functions become metric independent and one has a topological field theory of ``Witten type". 
For example, Donaldson-Witten theory can be constructed as a twist of 4d $\mathcal{N}=2$ Yang-Mills theory, in which case the spectrum of the twisted BRST operator becomes the $\SU(2)$ instantons. Other examples are the A- and B-models which are twists of 2d $\mathcal{N}=(2,2)$ $\sigma$-models. Subsequently, such twists have been studied extensively~\cite{Baulieu:1988xs,Vafa:1994tf,Blau:1996bx,Baulieu:1997jx,Acharya:1997jn,Baulieu:1997nj,Kapustin:2006pk,Elliott:2020ecf}.

For metrics which admit a parallel spinor field, one can single out a global scalar supercharge for the theory without the need to modify the Lorentz symmetry or energy-momentum tensor~\cite{Acharya:1997gp,Blau:1997pp}, 
though one still needs to add the scalar supercharge to the BRST operator of the theory (for example, in order that the Yang-Mills action becomes a topological term up to BRST exact pieces). 
The parallel spinor field is equivalent to a torsion-free $G$-structure on the tangent bundle, with a structure group that preserves the spinor. 
In other cases, one could have a $G$-structure which admits a line of singlet spinors which are charged under some $U(1)$ subgroup, and then one must perform a twist of the theory to obtain a scalar supercharge. 
Either way, such $G$-structures play an important role in the study of twisted theories. 
In gauge theories, the physical states which are kept in the twisted theories are then the various types of instanton configurations defined with respect to these $G$-structures,
and the quantum path integrals localise onto these states. 
These instantons solve first-order equations which imply the second order Yang-Mills field equations, as explored and classified in~\cite{Corrigan:1982th,Ward:1983zm}. 
There have been many previous studies of the geometry of these instanton configurations and their moduli spaces~\cite{ATIYAH1978185,Baulieu:1997jx,Acharya:1996tw,Dorey:2002ik}, their coupling to gravity in the context of heterotic supergravity~\cite{delaOssa:2016ivz,Bunk:2023ojb,Silva:2024fvl}, and their relation to D-branes in type II strings \cite{Bianchi:1998nk,Douglas:1995bn}. 

Recently, it was proposed how to twist supergravity theories~\cite{Costello:2016mgj} using this latter strategy, enabling also the study of holographic duality to the twisted gauge theories~\cite{Costello:2018zrm} as a toy model of the AdS/CFT correspondence~\cite{Maldacena:1997re}. In the supergravity context, the twist is performed by working in the BV formalism and giving a vacuum expectation value to the bosonic ghost field corresponding to one of the supersymmetries. For supergravity theories coupled to super Yang-Mills multiplets, this procedure leads to the twist of the Yang-Mills sector in agreement with the previous constructions. 
Much work has been done to construct twisted supergravity theories~\cite{Eager:2021ufo,Saberi:2021weg,Costello:2021kiv,Raghavendran:2021qbh,Hahner:2023kts}, much of it using the pure-spinor superfield formalism (see e.g.~\cite{Cederwall:2013vba} and references therein) and recent developments of it~\cite{Eager:2021wpi,Cederwall:2023wxc}. 

Ubiquitous in these constructions is a differential complex, properties of which provide information about the instanton states as well as 1-loop corrections.
In this paper, we show how one can construct these complexes associated to instantons and their gravitational analogues via a seemingly universal algebraic procedure whose main ingredient is a (generalised) $G$-structure on the manifold. 
We shall refer to these complexes as the \complexNoSpace.
One of the cohomologies of the \complex computes the infinitesimal moduli of the instantonic configuration. 
We go on to show how \complexes can be used in various ways to construct Chern--Simons-type actions associated to instantonic states and explore how their relation with spinors and supersymmetry provides additional relations between the various differential operators acting on them. 

A precursor of central importance for the present work is the complex defined by Carri\'on~\cite{ReyesCarrion:1998si}. It is associated to a $G$-structure (which must satisfy relevant constraints on its intrinsic torsion) such that the first cohomology group gives the infinitesimal moduli space of the relevant instanton configurations in gauge theories. In this article, we show how to construct such complexes not just for gauge theory instantons, but for instanton states associated to $G$-structures quite generally. The gravitational versions of such instantons are the torsion-free $G$-structures themselves while in supergravity they are general supersymmetric Minkowksi backgrounds (including fluxes). 
In the cases of heterotic geometries with flux, our construction reproduces the complexes found in earlier systematic constructions of the infinitesimal moduli~\cite{delaOssa:2014cia,delaOssa:2017pqy} after redefining fields (see also~\cite{Anderson:2014xha,Garcia-Fernandez:2015hja}). Indeed, this shows very concretely how the bundles denoted $Q$ in those references, 
and which also underlie the algebroids discussed in~\cite{Anderson:2014xha,Garcia-Fernandez:2018ypt,Garcia-Fernandez:2020awc}, arise naturally in generalised geometry.

Generalised geometry~\cite{Hitchin_2003,Gualtieri-thesis} and the resulting formulation of supergravity~\cite{CSW1} and supersymmetric backgrounds~\cite{Grana:2004bg,Grana:2005sn,CSW4,Coimbra:2016ydd} is a natural language in which to consider general supersymmetric solutions, and it will be a central part of our construction here. In this picture, the conditions for supersymmetry become equivalent to the existence of a torsion-free generalised $G$-structure~\cite{CSW4,Coimbra:2016ydd} (or singlet torsion in the AdS case~\cite{Coimbra:2015nha,Ashmore:2016qvs,Coimbra:2017fqv}) and this formalism has been used, for example, to make general statements about the moduli spaces of such backgrounds with no assumptions about the nature of the fluxes~\cite{Ashmore:2019qii,Ashmore:2019rkx,Tennyson:2021qwl,Smith:2022baw}, as well as many other developments. Here, this reformulation enables us to define the \complex associated to a supersymmetric flux geometry as the direct analogue of the complex we define for torsion-free $G$-structures in ordinary Riemannian geometry. 
In this work, for simplicity we restrict attention to $O(d,d)\times\bbR^+$ generalised geometry for the NS-NS sector of type II theories as described in~\cite{CSW1} and the $O(d,d+n)\times\bbR^+$ generalised geometry for heterotic supergravity~\cite{Garcia-Fernandez:2013gja,Coimbra:2014qaa}.\footnote{In this work we will take the connection on the tangent bundle which appears in the Riemann tensor squared term in the Bianchi identity to be an additional spurious degree of freedom, as in e.g.~\cite{delaOssa:2014cia,Ashmore:2019rkx}. See~\cite{McOrist:2021dnd,McOrist:2024zdz} for an interesting approach to resolving this issue.}
In these geometries, the decomposition of the generalised tangent bundle into positive and negative sub-bundles under the generalised metric enables one to associate target space geometric features with the left- and right-moving sectors of the string worldsheets. 
This in fact gives our complex the structure of a double complex and we will see that it behaves analogously to the double complex of $(p,q)$-forms on a K\"ahler manifold. 
The $\SU(3)\times\SU(3)$, $G_2 \times G_2$ and $\Spin(7) \times \Spin(7)$ cases were previously described in detail in~\cite{Ashmore:2021pdm}.

For particular $G$-structures, the \complex obtains extra structure, and we call these \emph{spinor-type} complexes. 
These appear already in examples of the gauge theory \complex of Carri\'on. It can happen that the sum of terms appearing in the even and odd parts of the \complex form the decomposition of spinor representations of the orthogonal group under the structure group $G$. 
The \complex thus becomes isomorphic to pair of spinors, which together form a representation of the Clifford algebra. 
The key point is that under this isomorphism, the Dirac operator is related to a combination of the differential and its adjoint. 
In these cases, one finds a general proof that the Laplacian for the \complex is proportional to the de Rham Laplacian acting on the forms. 
This provides a general structure to case-by-case computations which have appeared in other sources~\cite{2003math......5124B,Ashmore:2021pdm}. 

The notion of spinor type complex also appears in $O(d,d)$ generalised geometry, where the weighted spinor bundle associated to the generalised tangent space is isomorphic to the polyforms on the manifold~\cite{Hitchin_2003,Gualtieri-thesis,CSW1}, with the $O(d,d)$ generalised Dirac operator becoming the exterior derivative twisted by the three-form field strength $H$. Previous works~\cite{Gualtieri-thesis,Cavalcanti:2012fr,goto2005deformations} have studied the decomposition of this complex which occurs when one has generalised $G$-structures of various types and its Hodge theory. While in the generalised K\"ahler case we reproduce the complex of \cite{Cavalcanti:2012fr}, we emphasise that our construction is different to what is done in those works, allowing for complexes which do not have a generalised spinor interpretation. 
Further, even in the generalised K\"ahler case we use a different grading on the double complex according to the exterior powers of the sub-bundles $C_+$ and $C_-$ of the generalised tangent space in which the terms occur. This is crucial for matching the elements of the complex against worldsheet features of string theory. Rather than viewing the $O(d,d)$ Dirac operator $(\dd + H \wedge)$ as a differential, in our construction we view it as a sum of left- and right-moving Dirac operators. That its square is zero corresponds to the equality of the left- and right-moving Laplacians $\Delta_+$ and $\Delta_-$, which is related to level-matching on the string. Coupled with the K\"ahler identities mentioned above, this equality also ensures that we have an analogue of the $\der\bar\der$-lemma.  

Note that in the context of supersymmetric backgrounds, the isomorphisms of the \complexesNoSpace, which provide the deformation complexes of the generalised $G$-structure associated to the NSNS fields, with polyforms that could be thought of as RR fields, is a manifestation of the spectral flow on the worldsheet from a spacetime perspective. This identification of NSNS and RR degrees of freedom is special to $\cN=2$ backgrounds and explains why one can naturally describe the supersymmetry of NSNS backgrounds in terms of generalised pure spinors~\cite{Grana:2004bg,Grana:2005sn,Claus_Jeschek_2005}, which one would more naturally associate with RR degrees of freedom. 
In this sense, the description of the \complex that we provide here, thought of as a quotient of $\Lambda^\bullet(E)$, puts the discussion firmly back into the sphere of NSNS objects, and also applies to general structure groups, rather than restricting to the spinor type cases. 

As well as providing the infinitesimal moduli spaces of solutions, a principle application of our formalism lies in the systematic construction of field theories in which the classical states are instantons. In many interesting cases, the \complexes we define have a symplectic inner product, such that one can directly write down a quadratic BV action for a field theory from them. For certain examples of Carri\'on's gauge theory \complexesNoSpace, this reproduces the quadratic part of the actions for Chern--Simons theory and its various higher dimensional generalisations~\cite{thomas1997gauge,Donaldson:1996kp,Earp:thesis,Costello:2013zra} which have been of great interest to physicists and geometers alike. 
For the complexes arising from supergravity backgrounds, we are able to reconstruct the $\SU(3)$ and $G_2$ heterotic superpotential functionals to quadratic order. Indeed, via this construction, there will be a heterotic analogue of each of the aforementioned gauge theories. It was also shown in previous work~\cite{Ashmore:2021pdm} how the $\SU(3)\times\SU(3)$, $G_2 \times G_2$ and $\Spin(7) \times \Spin(7)$ complexes for type II geometry gave target space descriptions of the worldsheet BRST complexes of the corresponding (quasi-)topological string theories~\cite{Witten:1988xj,deBoer:2005pt}. This enables one to write target space quadratic actions (where the relevant symplectic pairing exists) whose quantisation gives the correct one-loop terms~\cite{Pestun:2005rp,Ashmore:2021pdm}.

This paper is structured as follows. In section~\ref{sec:Carrion}, we review Carri\'on's construction and go on to introduce spinor type complexes, prove general results about their Hodge theory and Laplacians, discuss the definition of instantons and the relation to BPS states and comment on product manifolds. 
Section~\ref{sec:BV} contains a series of examples of quadratic BV actions which result from the \complexes of section~\ref{sec:Carrion}. 
We generalise the construction to gravitational instantons (i.e.\ torsion-free $G$-structures) in section~\ref{sec:grav}. 
Section~\ref{sec:gen-Dolbeault} provides the further generalisation to generalised $G$-structures, which in relevant cases correspond to generic supersymmetric Minkowksi backgrounds of type II and heterotic supergravity with NSNS and gauge field fluxes. We also derive the K\"ahler identities for $G_+ \times G_-$ structures, and introduce spinor type complexes in the type II cases, for which we prove the equality of the Laplacians $\Delta_\pm$. 
Section~\ref{sec:applications} contains some examples of applications of the formalism: infinitesimal moduli spaces of flux backgrounds, topological string theories and heterotic superpotential functionals. 

\section{Carri\'on's instanton complex}
\label{sec:Carrion}

In this section, we review Carri\'on's construction of a complex associated to gauge theory instantons on a manifold with a torsion-free $G$-structure.\footnote{One can in fact relax the torsion-free requirement to allow for some intrinsic torsion in some cases, but for simplicity we will restrict to considering the case of torsion-free structures here.} These are the \complexes for gauge theories in our terminology. We go on to define what we call ``spinor-type" examples of these complexes, for which we are able to provide a general proof that their Laplacians are proportional to the usual de Rham Laplacians acting on forms. We also demonstrate that on product manifolds with product metrics, the \complex becomes the tensor product of the \complexes on the factors, such that the product of two spinor type examples remains spinor type. Finally, in cases where the \complex has a symplectic inner product of the right degree, one can see it as the BV complex of a field theory and we discuss several examples of theories which arise in this way in section~\ref{sec:BV}. 

\subsection{Colour-stripped \complexNoSpace}\label{sec:basic_carrion}

Let us recall how Carri\'on~\cite{ReyesCarrion:1998si} defines a complex $(\cx^\bullet, \check\dd)$ on a real manifold $M$ of dimension $d$ with a $G$-structure, for $G \subset \SO(d)$. In fact, to begin with, let us examine a ``colour-stripped" version, in which one considers only differential forms, rather than the endomorphism valued forms considered in~\cite{ReyesCarrion:1998si} which we will address later in section~\ref{sec:instanton}. By a slight abuse of notation, we will write $\mf g$ both for the Lie algebra of $G$ and the induced subbundle of $\Lambda^2T^*:=\Lambda^2T^*M\cong \mf{so}(TM)$. At first we consider the case with trivial gauge bundle for simplicity. One defines the bundle $\bx^k$ to be the quotient of the bundle of $k$-forms
\begin{equation}\Lambda^k T^*\end{equation} 
by the subbundle 
\begin{equation}\by^k:=\mf g\wedge \Lambda^{k-2} T^*.\end{equation}
Sections of $\bx^k$ and $\by^k$ will be denoted $\cx^k$ and $\cy^k$, respectively. The differential on the \complex $\cx^\bullet$ is then defined to be $\check{\dd} = \mathcal{P} \circ \dd$ where $\mathcal{P}$ is the projector onto the representation appearing in $\cx^\bullet$.

For a general $G$-structure, one would not find that $\check \dd^2=0$. However, it is easy to see that a sufficient condition for this to hold is that the image of $\dd$ restricted to $\cy^2=\Gamma(\mathfrak{g}) \subset \Gamma(\Lambda^2 T^*)=\Omega^2$ lies inside $\cy^3$. To see this, note that any element $\beta\in\cy$ can be written as a sum of terms
\begin{equation}
\begin{aligned}
	\beta = \sum_i \beta_i \wedge \lambda_i
\end{aligned}
\end{equation}
with $\beta_i \in \Gamma(\mathfrak{g})$, and so if $\dd(\cy^2)\subset\cy^3$ then
\begin{equation}
\begin{aligned}
	\dd \beta = \sum_i (\dd\beta_i) \wedge \lambda_i +  \sum_i \beta_i \wedge (\dd \lambda_i)
\end{aligned}
\end{equation}
also lies in $\cy$. This means that $\dd$ preserves $\cy$, which is then a differential ideal. To see that then $\check{\dd}^2 = 0$, note that given $\alpha \in \cx^k$, one has 
\begin{equation}
\begin{aligned}
	\dd \alpha = \check{\dd} \alpha + \eta
\end{aligned}
\end{equation}
with $\eta \in \cy$. Thus
\begin{equation}
\label{eq:check-dd-squared}
\begin{aligned}
	0 = \dd^2 \alpha = (\check{\dd}^2 \alpha + \eta') + \dd\eta
\end{aligned}
\end{equation}
where $\eta'$ and $\dd\eta$ are in $\cy$. This implies that both $\check{\dd}^2 \alpha= 0$ and $\eta' + \dd\eta = 0$ as $\Lambda^kT^*$ is the direct sum of $\bx^k$ and $\by^k$. We have essentially just shown that $(\cy, \dd)$ is a sub-complex and taken the quotient of $(\Omega^\bullet, \dd)$ by it to construct $(\cx,\check{d})$.

The condition that $\dd$ sends $\cy^2$ to $\cy^3$ can be viewed as an intrinsic torsion condition as follows. Let $\hat{\LC} = \LC + \Sigma$ be a $G$-compatible connection, which is not necessarily torsion-free. Here $\LC$ is the Levi-Civita connection and $\Sigma \in \Gamma( T^* \otimes (T \otimes T^*))$ is a generic tensor. As $\hat{\LC}$ is $G$-compatible, it preserves $\mathfrak{g}$ representations, and thus for any $\beta \in \cy^2 \subset \Omega^2$ one has
\begin{equation}
\begin{aligned}
	\hat{\LC} \beta \in \Gamma(T^* \otimes \mathfrak{g})
	\qquad \text{and so}
	\qquad \dd_{\hat{\LC}} \beta \in \Gamma(\mathfrak{g} \wedge T^*)=\cy^3.
\end{aligned}
\end{equation}
Therefore, we have
\begin{equation}
\begin{aligned}
	\dd \beta = \dd_\LC \beta = \dd_{\hat{\LC}} \beta - \dd_\Sigma \beta 
\end{aligned}
\end{equation}
Thus the condition that $\dd\beta \in \cy^3$ is equivalent to $\dd_\Sigma \beta \in \cy^3$. This is simply a condition on which $G$-irreducible parts of $\Sigma$ are allowed. However, any $G$-compatible part of $\Sigma$ can be absorbed into $\hat{\LC}$ and $\dd_\Sigma$ depends only on the torsion of $\Sigma$ (i.e.\ $\Sigma_{[m}{}^n{}_{p]}$) so in fact this is a constraint on the intrinsic torsion of the structure.

%%%%%%%%%%%%%%%%%%%%%%%%%%%

\subsection{Examples}

We start with the two trivial examples. Firstly, one can consider the case $G=\SO(d)$, in which case all two-forms are in the span of the Lie algebra at each point and are thus projected out, leaving
\begin{equation}
\label{eq:trivial-ord-complex}
\begin{aligned}
	0 \ra \Omega^{0} \stackrel{\dd}{\ra} \Omega^{1} 
	{\ra}  0
\end{aligned}
\end{equation}

The other trivial case is that of an identity structure $G=\{ \id \}$, in which case the quotient does nothing and the \complex is simply the de Rham complex.

As we will elaborate on further in section~\ref{sec:instanton}, one of the initial motivations for the construction was the observation that taking structure group $\SU(2)_L \subset \SU(2)_L \times \SU(2)_R \simeq \Spin(4)$ in four dimensions, one finds the complex 
\begin{equation}
\label{eq:4d-Carrion}
\begin{aligned}
	0 \ra \Omega^{0} \stackrel{\dd}{\ra} \Omega^{1} \stackrel{\mc{P}_+\dd}{\ra} \Omega^{2,+} 
	{\ra}  \,0
\end{aligned}
\end{equation}
where $\Omega^{2,+}$ is the self-dual two-forms and $\mc{P}_+$ is the projector $\Omega^{2} \ra \Omega^{2,+}$. A one-form $\alpha$ is thus $\check{d}$-closed if $\dd \alpha$ is anti-self-dual. 

Next, in any even dimension $d=2N$ one can consider $\mathfrak{u}(N) \subset \mathfrak{so}(d)$. We have a preserved complex structure $J$ with respect to which $\Gamma(\mathfrak{u}(N)) \simeq \Omega^{1,1}$, so that the quotient procedure gives
\begin{equation}
\label{eq:Dolbeault-complex-full}
\begin{aligned}
	0 \ra \Omega^{0,0} \stackrel{\der + \bar\der}{\ra} \Omega^{1,0} \oplus \Omega^{0,1}  
	\stackrel{\der + \bar\der}{\ra}  \Omega^{2,0} \oplus \Omega^{0,2}  
	\stackrel{\der + \bar\der}{\ra} \Omega^{3,0} \oplus \Omega^{0,3}   \stackrel{}{\ra} \dots
\end{aligned}
\end{equation}
One can also consider the subtly different complex which comes from taking $\mathfrak{su}(N) \subset \mathfrak{so}(d)$, which is the direct higher-dimensional analogue of~\eqref{eq:4d-Carrion}. This is the same as~\eqref{eq:Dolbeault-complex-full} aside from an extra piece in degree two, corresponding to forms proportional to the K\"ahler form $\omega$:
\begin{equation}
\label{eq:SUN-complex}
\begin{aligned}
	0 \ra \Omega^{0,0} \stackrel{\der + \bar\der}{\ra} \Omega^{1,0} \oplus \Omega^{0,1}  
	\stackrel{\der + \bar\der}{\ra}  \Omega^{2,0} \oplus \langle \omega \rangle \oplus \Omega^{0,2} 
	\stackrel{\der + \bar\der}{\ra} \Omega^{3,0} \oplus \Omega^{0,3}   \stackrel{}{\ra} \dots
\end{aligned}
\end{equation}
In the definition of an instanton that we describe in section~\ref{sec:instanton} below, this extra piece gives an extra constraint on the curvature of our gauge connection which imposes the full hermitian Yang-Mills equations rather than merely the condition that the gauge bundle is holomorphic.

Aside from the degree zero piece,~\eqref{eq:Dolbeault-complex-full} is clearly the induced real space of a sum of two complexes of complex forms: the usual Dolbeault complexes $(\Omega^{\bullet,0}, \der)$ and $(\Omega^{0,\bullet}, \bar\der)$. In what follows, we will complexify the scalar and tend to focus on the anti-holomorphic half
\begin{equation}
\label{eq:Dolbeault-complex-anti}
\begin{aligned}
	%\mathcal{A} : \hs{30pt} 
	0 \ra \Omega^{0,0} \stackrel{\bar\der}{\ra} \Omega^{0,1}  
	\stackrel{\bar\der}{\ra} \Omega^{0,2}  
	\stackrel{\bar\der}{\ra} \Omega^{0,3}    
		\stackrel{}{\ra} \dots \stackrel{\bar\der}{\ra} \Omega^{0,N} \ra 0
\end{aligned}
\end{equation}
The complexification of the scalar is important as it is the generator of gauge transformations in what follows, and in applications of this complex to moduli spaces, results in geometric invariant theory suggest that the moduli space including the D-term supersymmetry conditions should be the solutions of the F-terms modulo complexified gauge transformations (see e.g.~\cite{Baulieu:1997jx,Ashmore:2019qii} for discussions of this phenomenon). In the applications that we consider, the $U(N)$ complex~\eqref{eq:Dolbeault-complex-full} will be living on a Calabi-Yau manifold where there is in fact a further reduction of the structure group to $\SU(N)$, and this $\SU(N)$ structure is also torsion-free. 
The additional degree two piece of the complex that is present in~\eqref{eq:SUN-complex} versus~\eqref{eq:Dolbeault-complex-full} would impose the hermitian Yang-Mills condition on our instanton field strength, which can be thought of as a D-term condition. By working with the $U(N)$ complex and complexifying the scalar, we thus account for these additional conditions in a simple fashion when considering moduli. 
Thus in the remainder of this article, we will take~\eqref{eq:Dolbeault-complex-anti} (or very occasionally its complex conjugate) when considering $U(N) \subset \SO(2N)$ structures.

For other cases, we use the notation that $\Omega^k_{\rep{r}}$ denotes $k$-forms which are sections of the sub-bundle of $\Lambda^k T^*M$ transforming in the representation $\rep{r}$ of the structure group $G$.

In seven dimensions, one can consider torsion-free $G_2$ structures (i.e.\ metrics with Riemannian holonomy $G_2$), with $\mf g_2\subset \mf{so}(7)$. This gives a \complexNoSpace:
\begin{equation}
\label{eq:g2-ord-complex}
\begin{aligned}
	0 \ra \Omega^{0} \stackrel{\dd}{\ra} \Omega^{1} 
	\stackrel{\check\dd}{\ra}  \Omega^{2}_{\rep{7}}
	\stackrel{\check\dd}{\ra} \Omega^{3}_{\rep{1}}   \stackrel{}{\ra} 0
\end{aligned}
\end{equation}
Similarly, for $\Spin(7)$ in eight dimensions the \complex is:
\begin{equation}
\label{eq:spin7-ord-complex}
\begin{aligned}
	0 \ra \Omega^{0} \stackrel{\dd}{\ra} \Omega^{1} 
	\stackrel{\check\dd}{\ra}  \Omega^{2}_{\rep{7}}
	\stackrel{}{\ra} 0
\end{aligned}
\end{equation}
%

%%%%%%%%%%%%%%%%%%%%%%%%%%%

\subsection{Spinor type complexes and Hodge theory}
\label{sec:spinor-type}

In the construction above, we assumed that the structure group is a subgroup of $\SO(d)$ and thus in particular induces a Riemannian metric on $M$, and consequently also an inner product on $\Lambda^k T^*$. We can then identify $\bx$ with the orthogonal complement of $\by$, i.e.\ we have an orthogonal decomposition 
\begin{equation}
\begin{aligned}
	\Lambda^k T^*  = \bx^k \oplus \by^k  = \bx^k \oplus (\mf{g} \wedge \Lambda^{k-2}T^*) \ .
\end{aligned}
\end{equation}
Using the induced positive inner product on $\bx^\bullet$ one can define the adjoint $\check\dd^\dagger$ of the differential $\check\dd$. Armed with this, one can then define the Laplacian operator 
\begin{equation}
\begin{aligned}
	\check\Delta = \{ \check\dd , \check\dd^\dagger \}
\end{aligned}
\end{equation}
and find a Hodge decomposition of each space $\cx^k$ in the usual way\footnote{We ignore finer analytic details concerning the completeness of the inner product in making this statement here.} 
\begin{equation}
\begin{aligned}
	\cx^k = \im \check\dd \oplus \im \check\dd^\dagger \oplus H^k_{\check\dd} \ .
\end{aligned}
\end{equation}

In fact, it is often convenient to introduce conventional numerical factors into the definitions of the inner product and the adjoint so that these operations remain compatible with isomorphisms between the different spaces $\cx^k$ in the complex. We do not wish to go into details of this here, but refer the reader to~\cite{Ashmore:2021pdm} for explicit details in the cases of $G_2$ in seven dimensions and $\Spin(7)$ in eight dimensions. There it is seen that the convention choices result in, for example, Laplacian operators that depend only on the representation of the structure group and not the degree of form in which the representation appears.

One can observe that, in some interesting cases, the vector bundle underlying the \complex is isomorphic to a pair of spinor bundles $\bx^\bullet \simeq S_+ \oplus S_-$. We will refer to these cases as \complexes of ``spinor type". Here, in even dimensions the summands $S^{\pm}$ are the spinor bundles of positive and negative chirality, while in odd dimensions they are simply two isomorphic copies of the spinor bundle. 
Overall, they form a representation of the Clifford algebra, and thus there is an action of the Dirac operator on them, which is crucial for what follows below. 
One can also see that they correspond to the odd and even degree forms in the \complex respectively and thus transform oppositely under parity so that they form a pinor overall. 
For spinor type complexes, the inner products used to define the adjoint operator and Laplacian are naturally the relevant spinor inner products, written in terms of forms, and it is the representation theory of spinors that leads to the seemingly strange numerical factor choice arising in the treatment of e.g.~\cite{Ashmore:2021pdm}. Further, one has that a Dirac-type operator 
\begin{equation}
\label{eq:Dirac-complex}
\Dirac = a\check\dd + b\check\dd^\dagger
\end{equation}
corresponds to the usual spinor Dirac operator $\slashed{\LC}$ acting on the pair of spinors.\footnote{This is also noted in~\cite{Baulieu:1997jx}.} The coefficients $a$ and $b$ depend on the particular details of the case in question, but the square of the operator is always proportional to $\check\Delta$
\begin{equation}
\label{eq:Dirac-complex_square}
\Dirac^2 = ab\{ \check\dd ,\check\dd^\dagger \} = ab \check\Delta
\end{equation}

For example, for the $U(3)$ complex in six-dimensions, we have 
\begin{equation}
	0 \ra \Omega_\rep{1} \ra \Omega_\rep{3} \ra \Omega_\rep{\bar3} \ra \Omega_\rep{1} \ra 0
\end{equation}
where we have in fact decomposed into $\SU(3)$ representations as we assume that the metric has a torsion-free $\SU(3)$ structure as discussed above. The even parts are thus $\rep{1} + \rep{\bar3} \simeq \rep{\bar4}$ and odd parts are $\rep{1} + \rep{3} \simeq \rep{4}$. This is the usual isomorphism between $\Omega^{0,\text{(even/odd)}}$ and $S^{\pm}$ on a Calabi-Yau manifold given by 
\begin{equation}
	\omega_{\ba_1 \dots \ba_k} 
	\lra \omega_{\ba_1 \dots \ba_k} \gamma^{\ba_1 \dots \ba_k} \epsilon
\end{equation}
where $\epsilon$ is the parallel singlet spinor. It is then a simple calculation to see that the Dirac operator $\slashed{\LC}$ acts as $\bar\der + 2\bar\der^\dagger$ on the corresponding $(0,k)$-forms. This then squares to $2\Delta_{\bar{\der}} = \Delta_{\rm{dR}}$.

Another example of a spinor type complex comes from the $G_2$ case in seven-dimensions where we have the \complex representations 
\begin{equation}
	0 \ra \Omega_\rep{1} \ra \Omega_\rep{7} \ra \Omega_\rep{7} \ra \Omega_\rep{1} \ra 0
\end{equation}
so that again the even and odd parts become spinors:
\begin{equation}
\begin{aligned}
	   \gTwoEvenSpinor{\omega} &= \omega_0 \gTwoInvariantSpinor + \tfrac12 \omega_{mn} \gamma^{mn} \gTwoInvariantSpinor \\
	   \gTwoOddSpinor{\omega} &= \omega_m \gamma^m \gTwoInvariantSpinor + \tfrac1{3!} \tilde\omega_3 \gTwoForm_{mnp} \gamma^{mnp} \gTwoInvariantSpinor
\end{aligned}
\end{equation}
where $\gTwoInvariantSpinor$ is the parallel singlet spinor and the singlet 3-form is $\omega_{mnp} = \tilde\omega_3 \gTwoForm_{mnp}$ for a scalar $\omega_3$ and $G_2$ three-form $\gTwoForm$.

Following the conventions of \cite{Coimbra:2016ydd}, we first note that we can choose $\gTwoInvariantSpinor$ Majorana, and that the $\gamma$-matrices are purely imaginary. Thus, it is natural to form a complex spinor with real part $\psi$ and imaginary part $\chi$ as
\begin{align}
    \gTwoComplexSpinor(\omega) \coloneqq \gTwoEvenSpinor{\omega} + \gTwoOddSpinor{\omega}.
\end{align}
We can then express the spinor Dirac inner product via form inner products, and find:
\begin{align} \label{eq:spinor_inner_product}
\begin{split}
    \diracBar{\gTwoComplexSpinor(\omega)} \gTwoComplexSpinor(\lambda) = &\formInnerProduct{\omega_0}{\lambda_0} +  \formInnerProduct{\omega_1}{\lambda_1} + 3 \cdot \formInnerProduct{\omega_2}{\lambda_2} + 7 \cdot \formInnerProduct{\omega_3}{\lambda_3} \\ 
    &+ i \left(\formInnerProduct{\omega_0 \gTwoForm}{ \lambda_3} - \formInnerProduct{\omega_3}{\lambda_0 \gTwoForm} + \formInnerProduct{\omega_1 \contraction \gTwoForm}{ \lambda_2} - \formInnerProduct{\omega_2}{ \lambda_1 \contraction \gTwoForm}\right)
\end{split}
\end{align}    
where we used the standard metric induced inner product on forms $\formInnerProduct{\omega_p}{\lambda_p} = \frac{1}{p!} \omega_{m_1 \dots m_p} \lambda^{m_1 \dots m_p}$. We note that the factors in front of the two- and three-form inner product arise naturally through $\gamma$-matrix contractions, matching the ones imposed in \cite{Coimbra:2016ydd}, and thus providing an explanation for their appearance in the $G_2$ complex. We thus take the, positive definite, real part of the inner product eq. \eqref{eq:spinor_inner_product} to define the adjoint operators of $\dc$.

Since the Dirac operator $\diracOperator$ maps even to odd spinors and vice versa, we make the ansatz
\begin{align}
    \diracOperator \gTwoEvenSpinor{\omega} &= \gTwoOddSpinor{(\Dirac\omega)}, \\
    \diracOperator \gTwoOddSpinor{\omega} &= \gTwoEvenSpinor{(\Dirac\omega)},
\end{align}
for some differential operator $\Dirac$ acting on the forms.

Using the inner product formula \eqref{eq:spinor_inner_product}, we deduce the action of $\Dirac$ on the even spinor to be 
\begin{align}
    (\Dirac\omega)_m &= \LC_m \omega_0 + 3 \LC^{p} \omega_{pm} , \label{eq:D_omega_zero}\\
    (\Dirac\omega)_{mnp} &=  ( \gTwoProjector{3}{1} \covariantExteriorDerivative 
    	\omega_{2})_{mnp},
\end{align}
for $\gTwoProjector{3}{1}$ the projector onto the singlet part of the three-form, while the action on odd spinors is given by 
\begin{align}
    (\Dirac \omega)_0 &= \LC^m \omega_m, \\
    (\Dirac \omega)_{mn} &= (\gTwoProjector{2}{7} \covariantExteriorDerivative \omega_1)_{mn} + \frac{7}{3} \LC^p\omega_{p mn} \label{eq:D_omega_three} \ ,
\end{align}
for $\gTwoProjector{2}{7}$ the projector onto the $\rep{7}$ part of the two-form. Note, that not only do we find the $\dc \equiv \gTwoProjector{}{} \circ \covariantExteriorDerivative$ action naturally, but also the correct factors for the adjoint, as found in \cite{Coimbra:2016ydd}. Thus we can identify
\begin{align}
    (\Dirac\omega)_1 &= \dc \omega_0  -\dcDagger \omega_2, \\
    (\Dirac\omega)_3 &= \dc \omega_2, \\
    (\Dirac\omega)_0 &= -\dcDagger \omega_1, \\
    (\Dirac\omega)_2 &= \dc \omega_1 -\dcDagger \omega_3. 
\end{align}
Consequently, we have shown the $G_2$ complex to be of spinor type and that the Dirac operator corresponds to $\Dirac = \check{\dd} - \check{\dd}^\dagger$ so that
\begin{equation}
	 \diracOperator \gTwoComplexSpinor(\omega)
	= \gTwoComplexSpinor \big( (\check{\dd} - \check{\dd}^\dagger) \omega \big) 
\end{equation}

We note that where the complex becomes spinor type, there must be a singlet spinor $\epsilon$ corresponding to the scalar $\Omega^0$ and (at least in the $\SU(N)$, $G_2$ and $\Spin(7)$ cases) the vector representation does not feature a singlet of the structure group, so that we must have $\bar\epsilon \gamma^m \epsilon = 0$ and thus $\epsilon$ is a pure spinor. 
This suggests possible connections to the pure spinor superfield formalism~\cite{Nilsson:1985cm,Howe:1991bx,Berkovits:2000fe,Berkovits:2001rb,Cederwall:2001dx,Cederwall:2010tn}, which has been a topic of particular interest recently~\cite{Saberi:2021weg,Eager:2021wpi,Cederwall:2023wxc}. 

Using the spinor presentation, one can see that the Laplacian $\check\Delta$ will always be proportional to the de Rham Laplacian acting on the corresponding form if the metric is Ricci flat as follows.
The Laplacian $\check\Delta = (\check{\dd} + \check{\dd}^\dagger)^2$ is proportional to the square of the Dirac operator in the spinor formulation which a standard calculation expands as
\begin{equation}
	\slashed{\LC}^2 \psi = (\LC^m \LC_m - \tfrac14 R) \psi \ ,
\end{equation}
where $R$ is the Ricci scalar. For a spinor $\psi = \slashed{\omega} \epsilon$ with $\LC_m \epsilon = 0$, the operator on the RHS acts as a scalar operator and as $\LC \epsilon = 0$ its expression transfers directly across to the resulting form
\begin{equation}
	(\check\Delta \omega)_{m_1 \dots m_k} 
		\propto (\LC^m \LC_m - \tfrac14 R) \omega_{m_1 \dots m_k}
\end{equation}
Next, we expand the de Rham Laplacian acting on a generic form $\omega$ as
\begin{equation}
\begin{aligned}
	(\Delta \omega)_{m_1 \dots m_k} 
		& = -\LC^2 \omega_{m_1 \dots m_k}  
			+ k [\LC^p, \LC_{[m_1}] \omega_{|p|m_2\dots m_k]} \\
		&= -\LC^2 \omega_{m_1 \dots m_k} 
			+ k R^p{}_{[m_1}\omega_{|p|m_2 \dots m_k]} 
		 - k(k-1) R^p{}_{[m_1}{}^q{}_{m_2} \omega_{|pq|m_3 \dots m_k]} \\
		&= -\LC^2 \omega_{m_1 \dots m_k} 
			+ k R^p{}_{[m_1}\omega_{|p|m_2 \dots m_k]}
		 - \tfrac12 k(k-1) R_{[m_1m_2}{}^{pq} \omega_{|pq|m_3 \dots m_k]} 
\end{aligned}
\end{equation}
where we have used the algebraic Bianchi identity. However, as $R \in \Gamma(\mf{g} \otimes \mf{g})$ and $\cx$ was defined to be orthogonal to $\cy$, we have
\begin{equation}
	R_{mn}{}^{pq} \omega_{pq n_1 \dots} = 0
\end{equation}
for forms $\omega$ in $\cx$, so that for $\omega \in \cx$ the Riemann tensor term above vanishes. 
If the metric is Ricci flat, which for some structure groups $G$ is implied by the vanishing of the intrinsic torsion, we have $\check\Delta \propto \Delta \propto \LC^2$ for both operators acting on $\cx$.

This result implies that the harmonic forms, which are representatives of the cohomology, match those for the usual de Rham Laplacian in the relevant representation of the structure group. 
Note, however, that the Ricci flat condition is sufficient but not necessary here. For example, it is well-known that on all K\"ahler manifolds the Dolbeault Laplacian $\Delta_{\bar\der}$ is half the de Rham Laplacian and the cohomology groups also decompose by $(p,q)$ type.

\subsection{Bundle valued forms and instanton moduli}
\label{sec:instanton}

One can extend the \complex as defined above to one based on 
\begin{equation}
\Omega^\bullet(\End(V))=\Gamma(\Lambda^\bullet T^*\otimes \on{End}(V))
\end{equation} for $V$ a vector bundle with gauge group $K$, so long as one has a connection $A$ on it whose curvature $F \in \Omega^2(\End(V))$ is an instanton, by which we mean
\begin{equation}
\label{eq:gen-instanton}
	F \in \Gamma(\mf{g} \otimes \End(V)) \subset \Omega^2(\End(V))
\end{equation}
The point is that this means quantities of the form 
\begin{equation}
\label{eq:ddA-squared}
	\dd_A^2 \omega = (\dd + A \wedge)^2 \omega = F \wedge \omega \in \Gamma(\mf{g} \wedge \Lambda^\bullet T^*\otimes \End(V))
\end{equation}
are projected out when one takes the quotient by $\Gamma(\mf{g} \wedge \Lambda^\bullet T^*\otimes \End(V))$. Defining $\check\dd_A$ to be the composition of $\dd_A$ and projection onto the quotient $\cx^\bullet(\End V)$, the above arguments establishing that $(\cx^\bullet(\End V), \check\dd_A)$ is a complex go through exactly as before aside from the small modification that now~\eqref{eq:check-dd-squared} becomes
\begin{equation}
\label{eq:check-ddA-squared}
\begin{aligned}
	F \wedge \alpha = \dd_A^2 \alpha = (\check{\dd}_A^2 \alpha + \eta') + \dd_A \eta
\end{aligned}
\end{equation}
Due to~\eqref{eq:ddA-squared}, this still implies that $\check\dd_A^2 \alpha = 0$.

The physical motivation for defining these \complexes is that their cohomologies naturally capture the infinitesimal moduli spaces of generalised instanton configurations. These are defined to be field configurations (i.e.\ connections on $V$ up to global gauge transformations given infinitesimally by $\delta A = \dd_A \lambda$ for $\lambda \in \Omega^0(\End V)$) where~\eqref{eq:gen-instanton} holds. Given such a field configuration, one can consider an infinitesimal deformation of it $\delta A = \alpha \in \Omega^1(\End V)$. The induced infinitesimal change in the curvature is then $\delta F = \dd_A \alpha$, so the infinitesimal moduli space is~\cite{ReyesCarrion:1998si}
\begin{equation}
\label{eq:instanton-moduli}
\begin{aligned}
	\mc{M}_{\text{Instanton}} 
		= \frac{\{ \alpha : \check\dd_A \alpha = 0\}}{\{\alpha = \dd_A \lambda \}}
		= \frac{\ker (\check\dd : \cx^1 \ra \cx^2)}{\im (\check\dd : \cx^0 \ra \cx^1)} = H^1(\cx)
\end{aligned}
\end{equation}
One could also wonder why one defines a generalised instanton to be a configuration satisfying~\eqref{eq:gen-instanton}. Firstly, one notes for the four-dimensional case of~\eqref{eq:4d-Carrion} one recovers the usual condition of (anti-)self-dual curvature. As is now well-known, studying the moduli spaces of such objects led to the discovery of Donaldson invariants of four-manifolds~\cite{donaldson1997geometry}. More generally, a natural place where such configurations appear is in supersymmetric solutions to supersymmetric gauge theories on curved manifolds. Indeed, this setup exactly appears as the gauge sector part of the equations defining supersymmetric solutions of heterotic supergravity. 
In such theories, one can have a supersymmetry generated by a Killing spinor $\epsilon$ of the underlying supergravity background. Part of the supersymmetry conditions then become
\begin{equation}
\label{eq:instanton-susy}
\begin{aligned}
	\delta \chi \sim \gamma^{mn} F_{mn} \epsilon = 0
\end{aligned}
\end{equation}
where the spinor field $\chi$ is the fermionic superpartner of the gauge field. The spinor $\epsilon$ defines a $G$-structure on the manifold where $G$ is the stabiliser of $\epsilon$ and in the simplest case $\epsilon$ is parallel with respect to the Levi-Civita connection such that the $G$-structure is torsion-free. Equation~\eqref{eq:instanton-susy} then states that the two-form curvature $F$ lies in the Lie algebra $\mf{g}$ at each point, which is exactly our condition~\eqref{eq:gen-instanton} above.\footnote{Note again that for $G=\SU(N) \subset \SO(2N)$ the supersymmetry condition~\eqref{eq:instanton-susy} imposes both the holomorphic bundle condition $F^{0,2} = 0$ (F-term) and the hermitian Yang-Mills condition $\omega \inn F = 0$ (D-term). One expects that the infinitesimal moduli space of such configurations will be that of only the holomorphic bundle condition (F-term) moduli complexified gauge transformations. This will correspond to the cohomology of~\eqref{eq:Dolbeault-complex-anti} (with values in $\End(V)$), so we use this complex in these cases.} Note that, in the usual cases, such configurations satisfy a BPS bound, see~\cite{Acharya:1996tw} for more details.

\subsection{Product manifolds}\label{sec:product_manifolds}

Another complex one can construct using the above prescription, and one that will preempt some of the double complexes we will produce in later sections, is if the manifold has a product structure with a product metric. Suppose, for example, we have $M=X\times Y$ with $\dim X = m$, $\dim Y=n$, and metric $g_M = g_X + g_Y$. In this case, the structure group reduces
\begin{equation}\label{eq:product_structure}
    SO(d) \rightarrow SO(m)\times SO(n)
\end{equation}
and the de Rham complex decomposes into a double complex
\begin{equation}
    \ext^{k}T^{*}M \simeq \bigoplus_{k=p+q} \ext^{p}T^{*}X\otimes \ext^{q}T^{*}Y\ , 
\end{equation}
with
\begin{equation}
    \dd_{M} = \dd_{X} + \dd_{Y} \ , \quad \dd_{X}^{2} = \dd_{Y}^{2} = \dd_{X}\dd_{Y}+\dd_{Y} \dd_{X} = 0 \ .
\end{equation}

Suppose further that the manifolds $X,Y$ have a reduced structure group $G_{X} \subset SO(m)$, $G_{Y}\subset SO(n)$. Then it is easy to show that the \complex associated to $G$ is the tensor product of the \complexes associated to $G_{X}$ and $G_{Y}$. In the notation of section \ref{sec:basic_carrion}, we have
\begin{equation}\label{eq:product_manifold_complex}
    \begin{aligned}
        \bx^{k}_{M} &= \frac{\ext^{k}T^{*}M}{\mathfrak{g}\wedge\ext^{k-2}T^{*}M} \\
        & = \bigoplus_{p+q=k} \left( \frac{\ext^{p}T^{*}X}{\mathfrak{g}_{X}\wedge \ext^{p-2}T^{*}X } \right)\otimes \left( \frac{\ext^{q}T^{*}Y}{\mathfrak{g}_{Y}\wedge \ext^{q-2}T^{*}Y } \right) \\
        &= \bigoplus_{p+q=k} \bx^{p}_{X} \otimes \bx^{q}_{Y}
    \end{aligned}
\end{equation}
Furthermore, it easy to see that
\begin{equation}
    \check{\dd}_{M} = \check{\dd}_{X} + \check{\dd}_{Y} \ , \quad \check{\dd}_{X}^{2} = \check{\dd}_{Y}^{2} = \check{\dd}_{X}\check{\dd}_{Y}+\check{\dd}_{Y} \check{\dd}_{X} = 0 \ ,
\end{equation}
where $\check{\dd}_{X}$ is the differential associated to the \complex $\cx_{X}$ on $X$, and similarly for $\check{\dd}_{Y}$. Hence, the \complex decomposes into a double complex in this case.

It also follows from the representation theory of spinors that if $\cx_{X}$ and $\cx_{Y}$ are spinor complexes, then the total complex $\cx_{M}$ is also a spinor complex. Indeed, the tensor product of two pinor representations (over $\bbR$) gives an object on which one can act with the higher-dimensional Clifford algebra, which can be seen explicitly via similar gamma matrix decompositions to those found in appendix~\ref{app:Gamma}. 

\section{BV Chern--Simons theories}
\label{sec:BV}

In certain cases, the \complexes of section~\ref{sec:Carrion} occur with a natural graded symplectic pairing $\left< \cdot,\cdot \right>$ (or can be completed to give a natural pairing, see section \ref{sec:6d-superpotential}) which is compatible with the differential structure. When this happens, the \complex can provide the BV complex associated to some QFT. At the quadratic order, these are generally associated to generalisations of Chern--Simons theories. The general construction runs as follows. 

Suppose $p$ is such that $\langle \cx^p,\cx^{p+1}\rangle$ is non-zero.
In particular, this requires that the pairing is of degree $-(2p+1)$. Provided the pairing is compatible with the differential, there is a gauge symmetry $f_{p}\sim \check{\dd}f_{p-1}$, and gauge for gauge $f_{p-1}\sim\check{\dd}f_{p-2}$, etc. In the BV quantisation we would need to introduce ghosts, and ghosts for ghosts $f_{n}\in \cx^{n}$ for $n<p$. For each $f_{n}$, we introduce an anti-field $f_{2p+1-n} \in  \cx^{2p+1-n}$, and we can write the total BV action as
\begin{equation}
    S_{\text{BV}} = \frac{1}{2}\left< f,\check{\dd}f\right> \ , \qquad f\in \cx^{\bullet}\ .
\end{equation}
Above, we have written $f = f_{0} + f_{1} + ... $ for a generic element of the \complexNoSpace.

Since this action is quadratic, it is straightforward to quantise it and find the 1-loop partition function. A detailed discussion of how this is done in the de Rham case is given in e.g.\ \cite{Pestun:2005rp}, and the procedure works similarly for any case. Due to the alternating statistics of the ghosts and anti-ghosts, the modulus of the 1-loop partition function reduces to an alternating product of determinants of Laplacians $\check{\Delta}_{n}$ on $\cx^{n}$. The final result is, in general~\cite{Schwarz:1978cn,Schwarz:1979ae}
\begin{equation}
    |Z_{\text{1-loop}}| = \left(  \prod_{n=0}^{2p} (\det{}'\check{\Delta}_{n})^{(-1)^{n}} \right)^{\frac{1}{4}(-1)^{p+1}} = \mathrm{Tor}(\cx,\check{\dd})^{\frac{1}{2}(-1)^{p}}
\end{equation}
where, on the right hand side, we have introduced the definition of the analytic torsion of a complex $(\cx,\check{\dd})$.

We will see how this works for various cases below.

\subsection{3d abelian Chern--Simons}

The 3d abelian Chern--Simons theory is a quadratic theory with the following action.\begin{equation}
    S = \frac{1}{2}\int_{M} A_{1}\wedge \dd A_{1} \ , \qquad A_{1}\in \cx^{1}=\Omega^{1}(M)
\end{equation}
This has a 0-form gauge symmetry and so we need to introduce a ghost field $A_{0}\in \Omega^{0}$, and antifields $A_{n}\in \Omega^{n}(M)$ where $n=2,3$. The Grassman statistics of the fields satisfy $\epsilon(A_{n}) = (-1)^{n+1}$. The quadratic BV action for this theory is then given by
\begin{equation}
    S_{\text{BV}} = \frac{1}{2}\left< A,\check{\dd} A \right> = \frac{1}{2}\int_{M} A\wedge \dd A \ , \qquad A\in \cx^{\bullet} = \Omega^{\bullet}(M)
\end{equation}
where integration over the top-form component is implied. In this case, it is clear that the BV complex is just the de Rham complex, and the symplectic pairing is $\langle A,A'\rangle=\int_M A\wedge A'$. This is the \complex associated with the trivial structure group $G = \mathbb{1}$.\footnote{Note that by Stiefel's theorem, every orientable compact three-dimensional manifold is parallelisable~\cite{stiefel35}.} In this case, the 1-loop partition function gives the Ray-Singer torsion of the 3-manifold $M$ \cite{Witten:1990}.

\subsection{Holomorphic Chern--Simons}

Holomorphic Chern--Simons on a Calabi-Yau 3-fold works analogously. The quadratic action is given by
\begin{equation}
    S = \frac{1}{2}\int_{M} A_{0,1}\wedge \delb A_{0,1} \wedge \Omega \ , \qquad A_{0,1}\in \cx^{1} = \Omega^{0,1}(M)
\end{equation}
We need to introduce ghosts and anti-fields $A_{0,n} \in \Omega^{0,n}(M)$ with statistics $\epsilon(A_{0,n})=(-1)^{n+1}$. The quadratic BV action then becomes
\begin{equation}
    S_{\text{BV}} = \frac{1}{2}\left< A,\check{\dd} A \right>_{\Omega} = \frac{1}{2}\int_{M}A\wedge \delb A\wedge \Omega\ , \qquad A\in \cx^{\bullet} = \Omega^{0,\bullet}(M) 
\end{equation}
In this case, the BV complex is the \complex associated to the K\"ahler structure of the Calabi-Yau with $G= U(3)$, as in~\eqref{eq:Dolbeault-complex-anti} and the surrounding discussion. The symplectic pairing is the wedge-product followed by integration against the holomorphic top-form. Note that compatibility of the symplectic pairing with the differential $\delb$ requires that $\Omega$ is a holomorphic section of $\Omega^{3,0}(M)$, and hence we require a fully integrable $SU(3)$ structure, not just $U(3)$ structure, to recover the holomorphic Chern--Simons theory. Here, the 1-loop partition function gives the holomorphic Ray-Singer torsion of the Calabi-Yau 3-fold~\cite{thomas1997gauge}.

\subsection{\texorpdfstring{$G_{2}$}{G2} Chern--Simons}

There is a similar theory of $G_2$ instantons~\cite{Donaldson:1996kp,Earp:thesis} (further discussion can be found in e.g.~\cite{Walpuski:thesis,SaEarp:2014lbh,delaOssa:2016ivz}). The quadratic action is given by
\begin{equation}
    S = \frac{1}{2}\int_{M} A \wedge \dd A \wedge *\varphi \ , \qquad A \in \cx^{1} = \Omega^1(M)
\end{equation}
where $\varphi \in \Omega^3(M)$ is the $G_2$ three-form. 
We need to introduce ghosts and anti-fields $A_n \in \cx^n$ with statistics $\epsilon(A_{n})=(-1)^{n+1}$. The quadratic BV action then becomes
\begin{equation}
    S_{\text{BV}} = \frac{1}{2}\left< A,\check{\dd} A \right>_{\varphi} = \frac{1}{2}\int_{M}A\wedge \check{\dd}A\wedge *\varphi \ , \qquad A\in \cx^{\bullet}
\end{equation}
Here, the BV complex is simply the \complex for a $G_2$ manifold, and the symplectic pairing is the wedge product followed by integration with $*\varphi$. In this action, one could in fact use the usual exterior derivative instead of $\check\dd$ and the presence of $*\varphi$ will project $\dd A$ onto the relevant representation in $\cx$.

\subsection{4d Chern--Simons}

In \cite{Costello:2013zra}, a 4-dimensional version of Chern--Simons was put forward via a holomorphic twist of a 4-dimensional theory. It is defined on a manifold of the form $M=C\times \Sigma$, where $C = \bbR^{2}$ or $S^{1}\times \bbR$, and $\Sigma$ is a Riemann surface. To define the action, we require a meromorphic 1-form $\omega = \omega(z)\dd z$ on $\Sigma$, where $z$ is a holomorphic coordinate. The quadratic action is then given by
\begin{equation}\label{eq:4d_CS}
    S = \frac{1}{2}\int_{M}\omega\wedge A_{1}\wedge \dd A_{1} \ , \qquad A_{1} \in \Omega^{1}(M)
\end{equation}
Note that, because of the non-dynamical 1-form $\omega$, the $A_{z}\dd z$ component of $A_{1}$ drops out, and we can view $A\in \Gamma(\Lambda^{1}_\mb C(C)\oplus \Lambda^{0,1}(\Sigma))$.

To find the BV action, and relate this to the \complexes we have been discussing previously, we note that the background $C\times \Sigma$ has a natural $\mathbb{1}\times U(1) \subset SO(2)\times SO(2)\subset SO(4)$ structure. The associated \complex consists of sections of the bundles
\begin{equation}
\begin{aligned}
    \bx^{0} &= \Lambda^0_\mb C \ ,& \bx^{1} &= \Lambda^{1}_\mb C(C)\oplus \Lambda^{0,1}(\Sigma) \ , \\ 
    \bx^{2} &= \Lambda^{2}_\mb C(C) \oplus \big(\Lambda^{1}_\mb C(C) \otimes \Lambda^{0,1}(\Sigma)\big)  \ , & \bx^{3} &= \Lambda^{2}_\mb C(C)\otimes \Lambda^{0,1}(\Sigma) \ ,
\end{aligned}
\end{equation}
and the differential is
\begin{equation}
    \check{\dd} = \dd_{C} + \delb_{\Sigma}
\end{equation}
We can think of this as a double complex, as in section~\ref{sec:product_manifolds}, with the left and right differentials $\dd_{C}$ and $\delb_{\Sigma}$ respectively. To write down the BV action, we require a graded symplectic pairing which is compatible with the differential. Such a pairing is not unique, but is instead defined by an element $\omega \in \Omega^{1,0}(\Sigma)$, which must be $\dd_{C}+\delb_{\Sigma}$ closed, up to possible $\delta$-function contributions at the poles of $\omega$.\footnote{If this is the case, then one must carefully define boundary conditions for the field $A$. See e.g.\ \cite{Lacroix:2021iit}.} The full BV action can then easily be written as
\begin{equation}
    S_{\text{BV}} = \frac{1}{2}\left<A,\check{\dd}A\right>_{\omega} = \frac{1}{2}\int_{M}\omega\wedge A\wedge \check{\dd}A \ , \qquad A\in\cx^{\bullet}
\end{equation}
Restricted to degree 1, this reproduces the 4d Chern--Simons action \eqref{eq:4d_CS}, as required.

\section{Gravitational \complexNoSpace}
\label{sec:grav}

To build the gravitational \complex associated to a torsion-free $G$ structure, we start from the graded vector space $\Omega^\bullet (TM)$. This does not carry a natural differential. One then views $\mf{gl}(TM) \simeq \Lambda^1(T^*) \otimes T$ in the same way as one identified $\mathfrak{so}(TM) \simeq \Lambda^2 T^*$ in section~\ref{sec:Carrion}. Given a $G$-structure, one can then view $\mf{g} \subset \Lambda^1T^*\otimes T$ and form $\bx$ as the quotient of $\Lambda^\bullet T^*\otimes T$ by $\mf{g} \wedge \Lambda^\bullet T^*$, and again set $\cx:=\Gamma(\bx)$. The first terms of this complex are then\footnote{A similar complex to this also appears in~\cite{goto2004moduli}.}
\begin{equation}
\label{eq:GLd-moduli-complex}
\begin{aligned}
	0 \ra \Gamma(TM)
	\ra \Gamma(\mathfrak{gl}(TM) / \mathfrak{g})
	\ra \Gamma(T^{(\text{int})})  \ra \dots
\end{aligned}
\end{equation}
The third term here requires a little explanation. Consider a tensor $\Sigma \in \Omega^1(\mf{gl}(TM))$ which can be thought of as the difference of two connections on the tangent bundle. We define a map $\tau\colon T^* \otimes (T \otimes T^*) \ra T \otimes \Lambda^2 T^*$ to give the difference of the torsions of the two connections, i.e.\ with respect to any frame $\{ \he_a \}$ for the tangent bundle
\begin{equation}
	\tau(\Sigma)^a{}_{bc} = - 2 \Sigma_{[b}{}^a{}_{c]}
\end{equation}
This map restricts to a map $\tau|$ on $T^* \otimes \mf{g}$. We then have an exact sequence of bundles 
\begin{equation}
	0 \ra \ker (\tau|) \stackrel{\iota}{\ra} T^* \otimes \mf{g}
		\stackrel{\tau|}{\ra} T \otimes \Lambda^2 T^*
		\stackrel{\pi}{\ra} T^{(\text{int})} (\tau|) \ra 0
\end{equation}
where we have defined the bundle
\begin{equation}
	T^{(\text{int})} = {\rm coker} (\tau|) = (T \otimes \Lambda^2 T^*) / \im (\tau|)
\end{equation}
Given a $G$-compatible connection, the projection of the torsion of this connection onto $T^{(\text{int})}$ does not change if one shifts to a different $G$-compatible connection by adding to it a tensor $\Sigma \in \Omega^1(\mf{g})$. Therefore, it is independent of the choice of $G$ compatible connection and represents a property of the $G$-structure itself. It is called the intrinsic torsion of the structure, and can be thought of as a part of the torsion which is common to all connections compatible with the $G$-structure.

Given a torsion-free $G$-compatible connection $\hat\LC$, we can define on $\alpha \in \Omega^k(TM)$
\begin{equation}
	(\hat\dd \alpha)^m{}_{n_1 \dots n_{k+1}} 
		= (k+1) \hat\LC_{[n_1} \alpha^m{}_{n_2 \dots n_{k+1}]}
\end{equation}
As the connection is compatible, and thus preserves $G$ representations, this can then be projected onto the quotient complex to define an operator $\check\dd$ on the quotient complex $\mc{A}$ as in section~\ref{sec:Carrion}. One can see that this is independent of the choice of torsion-free compatible connection and thus $\check\dd$ is a natural operator on the complex. First, if one shifts the torsion-free compatible connection by a torsion-free tensor $\Sigma \in \Omega^1(\mf{g})$ then the shift of the operator $\hat\dd$ acting on $\alpha\in\Omega^k(TM)$ representing an element of $\mc{A}^k$ is
\begin{equation}
\begin{aligned}
	\delta ( \hat\dd \alpha)^m{}_{n_1 \dots n_{k+1}} 
		&= (k+1) (\Sigma_{[n_1} \cdot \alpha)^m{}_{n_2 \dots n_{k+1}]} \\
		&= (k+1) \Sigma_{[n_1}{}^m{}_{|p|} \alpha^p{}_{n_2 \dots n_{k+1}]}
			- k(k+1) \Sigma_{[n_1}{}^p{}_{n_2} \alpha^m{}_{|p| \dots n_{k+1}]} \\
		& = (k+1) \Sigma_{p}{}^m{}_{[n_1} \alpha^p{}_{n_2 \dots n_{k+1}]}
\end{aligned}
\end{equation}
which lies in $\Gamma(\mf{g} \wedge \Lambda^{k}T^*)$ and is thus annihilated on projection to $\mc{A}$.

Further one can check that $\check\dd^2 = 0$ via the same proof as in section~\ref{sec:Carrion}. For $\alpha\in\Omega^k(TM)$ representing a class in $\mc{A}^k$, one has
\begin{equation}
\begin{aligned}
	(\hat\dd^2 \alpha)^m{}_{n_1 \dots n_{k+2}} 
		&= (k+1)(k+2) \hat\LC_{[n_1} \hat{\LC}_{n_2} \alpha^m{}_{n_3 \dots n_{k+2}]} \\
		&= (k+1)(k+2) R_{[n_1n_2}{}^m{}_{|p|} \alpha^p{}_{n_3 \dots n_{k+2}]} \\
			&\qquad\qquad
				- k(k+1)(k+2) R_{[n_1n_2}{}^p{}_{n_3} \alpha^m{}_{|p|n_4 \dots n_{k+2}]}  \\
		& = -2(k+1)(k+2) R_{p[n_1}{}^m{}_{n_2} \alpha^p{}_{n_3 \dots n_{k+2}]} \\
\end{aligned}
\end{equation}
where we have used the Bianchi identity $R_{[mn}{}^p{}_{q]} = 0$. Since the Riemann tensor is a section of $\Lambda^2 T^* \otimes \mf{g}$, we have $\check\dd^2\alpha\in\Gamma(\mf{g} \wedge \Lambda^{k+1}T^*)$. One can then write
\begin{equation}
	\hat\dd \alpha = \check\dd \alpha + \eta
\end{equation}
for $\eta \in \Gamma(\mf{g} \wedge \Lambda^{k}T^*)$ and then 
\begin{equation}
	\hat\dd^2 \alpha = \check\dd^2 \alpha + \eta' + \hat\dd\eta
\end{equation}
for $\eta' \in \Gamma(\mf{g} \wedge \Lambda^{k+1}T^*)$. Projecting this equation onto the quotient $\mc{A}^{k+2}$ one thus arrives at $\check\dd^2 = 0$ on $\mc{A}$.

One can interpret this \complex as follows. Consider a $G$-frame $\{ \hat{e}_a \}$ for the tangent bundle, which is a local section of the principal sub-bundle of the frame bundle corresponding to the $G$-structure. (In these frames, the components of all invariant tensors of $G$ take a specific set of constant values so that the matrix representation of the structure group $G$ is also the same at all points.)

For simplicity of exposition\footnote{This is not necessary: one can instead work via projections onto the quotient $\mf{gl}(d,\bbR) / \mf{g}$.} we assume reducibility of $\mf{gl}(d,\bbR)$ under the subalgebra $\mf{g}$ so that

\begin{equation}
	\mf{gl}(d,\bbR) = \mf{g} \oplus K
\end{equation}
for $K$ some representation of $\mf{g}$. An infinitesimal variation of the $G$-structure then corresponds to defining a new frame $\{ \hat{e}'_a \}$ for the deformed $G$-structure, which we can write as
\begin{equation}
	\hat{e}'_a = \hat{e}_a + X^b{}_a \hat{e}_b
\end{equation}
for components $X^a{}_b$ in $K$ at each point. The tensor $X$ thus defines an element of $\cx^1$.

Suppose we have a compatible connection $\LC$. For a vector field $v = v^c \he_c$ we then have
\begin{equation}
	\LC_v \he_a = v^c \LC_{\he_c} \he_a = v^c \omega_c{}^b{}_a \he_b 
		=: \omega_v{}^b{}_a \he_b
\end{equation}
where the components of $\omega_v$ lie in $\mf{g}$. We then look at deforming the connection $\LC$ to $\LC' = \LC + \Sigma$ for $\Sigma \in \Omega^1 (\mf{gl}(TM))$. Then 
\begin{equation}
\begin{aligned}
	\omega'_v{}^b{}_a \he'_b = \LC'_v \he'_a = \LC_v \he'_a + \Sigma_v{}^b{}_a \he'_b
\end{aligned}
\end{equation}
and we have
\begin{equation}
\begin{aligned}
	\LC_v \he'_a &= \LC_v (\hat{e}_a + X^b{}_a \hat{e}_b) \\
	&= \omega_v{}^b{}_a \he_b + (\der_v X^b{}_a) \he_b + X^b{}_a \omega_v{}^c{}_b \he_c \\
	&= \omega_v{}^b{}_a (\he'_b - X^c{}_b \he'_c) 
		+ (\der_v X^b{}_a + \omega_v{}^b{}_c X^c{}_a) \he'_b + O(X^2) \\
	&= (\omega_v{}^b{}_a  + \LC_v X^b{}_a) \he'_b + O(X^2)
\end{aligned}
\end{equation}
Note that in these equations, the components of $\omega_v$ are naturally the components of the connection $\LC$ with respect to the original frame $\he_a$, while the components of $\omega'_v$ of $\LC'$ and the tensor $\Sigma$ are taken with respect to the $\he'_a$ frame. We thus have that to first order in $X$ the variation of the connection is given by:
\begin{equation}
	\omega'_v{}^a{}_b  - \omega_v{}^a{}_b = \Sigma_v{}^{a}{}_b + (\LC_v X)^a{}_b
\end{equation}
We require that $\omega_v$ and $\omega'_v$ lie in $\mf{g}$ and therefore so must $\Sigma_v + \LC_v X$. If we write
\begin{equation}
	\Sigma = \Sigma^{(\mf{g})} + \Sigma^{(K)}
\end{equation}
then we have that
\begin{equation}
	\Sigma^{(K)}_v = - \LC_v X
\end{equation}
We then define the map $\tau$ as above and $\tau^{(\text{int})} = \pi \circ \tau$ and have
\begin{equation}
\begin{aligned}
	\tau^{(\text{int})}(\Sigma) &= \tau^{(\text{int})}(\Sigma^{(\mf{g})} + \Sigma^{(K)})
		= \tau^{(\text{int})}(\Sigma^{(K)})
		= \pi (\tau(\Sigma^{(K)})) \\
		&= \mc{P}_{A^2} \Big( \tfrac12 (2 \LC_{[b} X^a{}_{c]}) \; 
			\he_a \otimes (e^b \wedge e^c) \Big) \\
		&= \check\dd X
\end{aligned}
\end{equation}
so that $\check\dd X$ is the intrinsic torsion of the new $G$-structure to first order in $X$.

Next, we consider what it means for $X$ to be exact. If we act with an infinitesimal diffeomorphism generated by a vector field $v$ on our frame, we have 
\begin{equation}
\label{eq:diff-variation}
	\delta \he_a = \mc{L}_v \he_a = \mc{L}^{\LC}_v \he_a
	 	= \LC_v \he_a - (\LC \times v) \cdot \he_a \\
		= (\omega_v{}^b{}_a - \LC_a v^b) \he_b
\end{equation}
where here we used the notation $\LC \times v$ for $\LC v$ viewed as an endomorphism of $TM$. Now, as $\omega_v$ lies in $\mf{g}$ an infinitesimal rotation of the frame by $\omega_v$ merely results in a new frame compatible with the original $G$-structure and thus does not change the $G$-structure itself. Therefore, for our purposes here we can discard this part of the variation. The last term in~\eqref{eq:diff-variation}, however, does contain a part which appears to change the $G$-structure. This is the part of $\LC \times v$ with components in $K$, which is precisely $\check\dd v$. Therefore, $G$-structures which are related by an infinitesimal diffeomorphism have 
\begin{equation}
	X = \check\dd v
\end{equation}
for some vector field $v$.

We thus see that the infinitesimal moduli space of torsion-free $G$-structures is given by
\begin{equation}
\label{eq:G-str-moduli}
\begin{aligned}
	\mc{M}_{G\text{-str}} 
		= \frac{\{ X : \check\dd X = 0\}}{\{ X = \check\dd v \}}
		= H^1(\cx)
\end{aligned}
\end{equation}
exactly as for~\eqref{eq:instanton-moduli}. 

For example, one can consider the case of a $\GL(n,\bbC)$ structure on a $2n$-dimensional manifold, corresponding to a complex structure. In this case one has that $\mf{g} \sim \Big[ \big ((T^*)^{1,0}\otimes T^{1,0} \big) \oplus \big((T^*)^{0,1}\otimes T^{0,1} \big)\Big]_\bbR$, so that projecting out $\mf{g} \wedge \Lambda^\bullet T^*$ we are left with the \complex $\cx^\bullet$ with
\begin{equation}
	\cx^k = \Big[ \Omega^{0,k}(T^{1,0}) \oplus \Omega^{k,0}(T^{0,1}) \Big]_\bbR \ ,
\end{equation}
with the differential given by the usual $\bar\der$ on $\Omega^{0,k}(T^{1,0})$ and $\der$ on $\Omega^{k,0}(T^{0,1})$. As for the discussion of the usual Dolbeault complexes of forms in section~\ref{sec:Carrion} above, we tend to use the complex parameterisation of this complex and simply write it as
\begin{equation}
	(\cx^\bullet, \check\dd)  = (\Omega^{0,k}(T^{1,0}), \bar\der) 
\end{equation}

This \complex is of course well-known in the study of complex structures (see e.g.~\cite{Huybrechts}). Indeed, $\mu \in \Omega^{0,1}(T^{1,0})$ provides a Beltrami differential, deforming the complex structure via
\begin{equation}
	\frac{\der}{\der \bar{z}'^{\ba}} = \frac{\der}{\der \bar{z}^{\ba}} 
		+ \mu^b{}_{\ba} \frac{\der}{\der z^{b}}
\end{equation}
and in the infinitesimal case, such a deformation is induced by an infinitesimal diffeomorphism if
\begin{equation}
	\mu = \bar\der \epsilon \hs{30pt} \epsilon \in \Gamma(T^{1,0})
\end{equation}
The condition that $[T^{0,1}, T^{0,1}] \subset T^{0,1}$ then comes out to be 
\begin{equation}
	\bar\der \mu = 0
\end{equation}
so that indeed the infinitesimal moduli space of complex structures is given by the first cohomology as above.

In the above example, the structure group $\GL(n,\bbC)$ does not preserve a metric. In general, given a torsion-free $G$-structure on $TM$, the Riemann tensor $R \in \Omega^2(\mf{g})$. If $G$ preserves a metric, i.e.\ $G\subset \SO(d)$ as we assumed in section~\ref{sec:Carrion}, then $R_{mnpq} = R_{pqmn}$ and we have $R \in \Gamma(\mf{g} \otimes \mf{g}) \subset \Omega^2(\mf{g})$. Thus for metric $G$-structures, one has a close analogue of~\eqref{eq:gen-instanton}. 

\section{The \complex for Courant algebroids and supergravity}
\label{sec:gen-Dolbeault}

We now proceed to mimic the construction of section~\ref{sec:grav} in generalised geometry. We will see that this produces a natural \complex associated to torsion-free generalised $G$-structures, one of whose cohomology groups will later be seen to give the infinitesimal moduli space of such structures in section~\ref{sec:moduli}. For particular types of structure which preserve a generalised metric, we will see that the \complex becomes a double complex, which is a tensor product of \complexes of the type considered in section~\ref{sec:Carrion}. We show that these satisfy K\"ahler type identities in general. There is also a notion of spinor-type complexes in generalised geometry, which have the further interesting properties that the left and right Laplacian operators are equal, and thus they satisfy $\der\bar\der$-type lemmas.

We start with the tensor hierarchy (excluding dilaton terms) for $O(d,d)$ generalised geometry. This is simply the graded vector space $\Gamma(\Lambda^\bullet E)$ where $E \simeq T \oplus T^*$ is the generalised tangent bundle. In fact, one can also consider the case of $O(d,d+n)$ generalised geometry for heterotic supergravity where $E \simeq T \oplus \End(V) \oplus T^*$ for a gauge bundle $\End(V)$, or more generally any Courant algebroid. Similarly to the gravitational construction of section~\ref{sec:grav}, the graded vector space $\Gamma(\Lambda^\bullet E)$ does not carry a natural differential. 
However, we will consider the case where $E$ has a torsion-free generalised $G$-structure~\cite{CSW4} and take a quotient of $\Lambda^\bullet E$ by  $\mathfrak{g} \wedge \Lambda^\bullet E$  viewing $\mathfrak{g} \subset \Lambda^2 E \sim \mf{so}(d,d+n)$. 
We will then find that the sections of this quotient have a natural differential as in the previous section.

Demonstrating the existence of the differential and deriving its properties is slightly harder than in the constructions of the previous sections, in large part due to the complications of defining a Riemann tensor in generalised geometry. 
In this section, we explain how the complex is constructed for any generalised $G$-structure in $O(d,d+n)$ generalised geometry, and go on to show that for a large class of structure groups which preserve a generalised metric one in fact obtains a double complex satisfying K\"ahler type identities. 
Further, as $\SO(d,d+n)$ is an orthogonal group, it has spin representations and there is an analogue of the spinor type complexes of section~\ref{sec:Carrion}. The generalised-metric-compatible cases of these turn out to have the property that the two natural Laplacians on the double complex are equal, which is the analogue of the statement that $\Delta_\der = \Delta_{\bar\der}$ in ordinary K\"ahler geometry.

A quick note on notation and conventions in this section. We shall use $e$ to denote vectors in a local frame of the Courant algebroid $E$. Since $E \cong E^{*}$, we shall not distinguish between their frames. Uppercase Roman indices $e_{A}$ will denote an arbitrary frame and indices will be raised and lowered with the canonical inner product $\left<\cdot,\cdot \right>$. Lowercase Roman indices $e_{a},e_{\bar{a}}$ will denote frames of $C_{\pm}$ respectively, which are the positive and negative eigenbundles of the generalised metric. 

  \subsection{Preliminaries}
  \subsubsection{Courant algebroids}
    Let us start setting the stage by introducing Courant algebroids \cite{liu1997manin} and discussing some of their basic properties.

    A Courant algebroid is a vector bundle $E\to M$ equipped with 
    \begin{itemize}
        \item a bracket $[\slot,\slot]\colon \Gamma(E)\times\Gamma(E)\to\Gamma(E)$
        \item a fiberwise non-degenerate symmetric bilinear form $\langle \slot,\slot\rangle$
        \item a vector bundle map $\rho\colon E\to TM$
    \end{itemize}
    satisfying the following conditions for all $u,v,w\in\Gamma(E)$ and $f\in C^\infty(M)$:
        \begin{equation}[u,[v,w]]=[[u,v],w]+[v,[u,w]],\qquad [u,fv]=f[u,v]+(\rho(u)f)v,\end{equation}
        \begin{equation}\rho(u)\langle v,w\rangle=\langle [u,v],w\rangle+\langle v,[u,w]\rangle,\qquad [u,v]+[v,u]=\rho^*d\langle u,v\rangle,\end{equation}
        where $\rho^*\colon T^*M\to E^*$ is the transpose of $\rho$, and we used the identification $E\cong E^*$ provided by the pairing $\langle\slot,\slot\rangle$.
    Note that, in parallel with ordinary geometry we will also use the notation $L_uv:=[u,v]$ and call $L$ the Dorfman derivative.
 
    Various other properties can be derived from these axioms, such as
  \begin{equation}\rho([u,v])=[\rho(u),\rho(v)]\end{equation}
  or $\rho\circ \rho^*=0$. The latter property can be rephrased as the statement that
  \begin{equation}0\to T^*M\xrightarrow{\rho^*} E\xrightarrow{\rho}TM\to 0\end{equation}
  is a chain complex. When this is in fact an exact sequence, we say that the Courant algebroid is \emph{exact}. More generally, if the complex is exact in the last (or equivalently in the first) point, i.e.\ when $\rho$ is surjective, we say that the algebroid is \emph{transitive}.

  An example of an exact Courant algebroid is given by
  \begin{equation}\label{eq:exactCA}
      E=TM\oplus T^*M,
  \end{equation}
  equipped with
  \begin{equation}\langle X+\alpha,Y+\beta\rangle=\alpha(Y)+\beta(X),\qquad \rho(X+\alpha)=X,\end{equation}
  \begin{equation}[X+\alpha,Y+\beta]=\mc L_XY+(\mc L_X\beta-i_Yd\alpha+H(X,Y,\slot),\end{equation}
  where $H$ is closed 3-form. In fact, one can show \cite{Severa,Severa:2015hta} that every exact Courant algebroid is of this form, for some $H$. 

    More generally, starting from an arbitrary principal $G$-bundle $P \to M$ with vanishing first Pontryagin class w.r.t.\ an invariant pairing on $\mf g$, one can construct a transitive Courant algebroid structure on
    \begin{equation}E=TM\oplus\on{ad}(P)\oplus T^*M.\end{equation}
    Every transitive Courant algebroid is locally of this form \cite{Severa,Severa:2015hta}.

    Returning now to the general case, let $e_A$ be a local frame for which $\langle e_A,e_B\rangle$ are constant functions. It is then easy to see that the structure coefficients \begin{equation}c_{ABC}:=\langle [e_A,e_B],e_C\rangle\end{equation}
    are completely antisymmetric. The Jacobi identity can then be written as
    \begin{equation}\label{eq:jacobi}
        \rho(e_{[A})c_{BC]E}+c_{D[AB}c\indices{_{C]}^D_E}-\tfrac13\rho(e_E)c_{ABC}=0.
    \end{equation}
    Note that in this section all the indices will always be raised/lowered with $\langle \slot,\slot\rangle$ and we will freely assume the identification $E\cong E^*$ provided by this pairing.

    \subsubsection{Generalised metrics, connections, and \texorpdfstring{$G$}{G}-structures}
        The analogue of an ordinary metric is provided by a \emph{generalised metric} $\mc G$, which is a symmetric endomorphism (vector bundle map) of $E$ satisfying $\mc G^2=1$. We will denote the $\pm 1$-eigenbundles of $\mc G$ by $C_\pm$. Note that 
        \begin{equation}\label{eq:ogdec}
            E=C_+\oplus C_-    
        \end{equation}
        is an orthogonal decomposition. Conversely, any orthogonal decomposition \eqref{eq:ogdec} corresponds to a generalised metric $\mc G$. Generalised metric is thus equivalent to the choice of a subbundle $C_+
        \subset E$ for which $\langle \slot,\slot\rangle|_{C_+}$ is non-degenerate.

        As an example, any choice of metric $g$ and 2-form $B$ on $M$ defines a generalised metric on the exact Courant algebroid \eqref{eq:exactCA} by
        \begin{equation}C_+:=\on{graph}(g+B)=\{X+(g(X,\slot)+B(X,\slot))\mid X\in TM\}.\end{equation}
        Similarly, one can encode the data of $g$, $B$, and a connection $A$ using a generalised metric on a transitive Courant algebroid.
        
        If the signature of the pairing $\langle \slot,\slot\rangle$ is $(p,q)$, the Courant algebroid naturally has a reduced structure group $O(p,q)\subset GL(p+q,\mb R)$. A choice of a generalised metric breaks this group down further to the product of the orthogonal group corresponding to the decomposition \eqref{eq:ogdec}. In particular, if the induced pairing on $C_+$ and $C_-$ is positive and negative definite, respectively, then the group reduces as
        \begin{equation}O(p,q) \to O(p)\times O(q).\end{equation}

        For any subgroup $G\subset O(p,q)$, a \emph{generalised $G$-structure} is a reduction of the structure group from $O(p,q)$ to $G$. In particular, if $G\subset O(p)\times O(q)$ then a generalised $G$-structure induces a generalised metric of the above type. Note that since any local $G$-frame is also an $O(p,q)$-frame, the functions $\langle e_A,e_B\rangle$ are automatically constant.

        Quite analogously to the standard case, we define \emph{(Courant algebroid) connections} \cite{AXu} as operators $D$ satisfying
        \begin{equation}D_{fu}v=fD_uv,\quad D_u(fv)=fD_uv+(\rho(u)f)v,\quad \rho(u)\langle v,w\rangle=\langle D_uv,w\rangle+\langle v,D_uw\rangle.\end{equation}
        Defining the action of $D$ on functions by $D_uf:=\rho(u)f$ and using the Leibniz rule, we can act with $D$ on an arbitrary section of the tensor products of $E$.
        
        A \emph{torsion} of $D$ is the tensor \cite{AXu}
        \begin{equation}T_D(u,v)=D_uv-D_vu-[u,v]+\langle Du,v\rangle.\end{equation}
        It is easy to see that $T\in\Gamma(\Lambda^3 E)$. A connection $D$ is called \emph{Levi--Civita} if it has vanishing torsion and $D\mc G=0$.
        
        Finally, a generalised $G$-structure is called \emph{torsion-free} if it admits a torsion-free connection preserving the generalised $G$-structure. For instance, any generalised metric is torsion-free \cite{Garcia-Fernandez:2016ofz}.

  \subsubsection{Riemann tensor}
    For any torsion-free connection $D$ we define the \emph{Riemann tensor} \cite{Siegel:1993th}
    \begin{equation}\Riem(w,z,x,y):=\tfrac12w^D y^B(x^A[D_A,D_B]z_D+z^A[D_A,D_D]x_B-(D_A x_B)(D^A z_D)).\end{equation}
    It is not difficult to see that this indeed defines a tensor with the following symmetries:
    \begin{equation}\Riem_{ABCD}=\Riem_{[AB]CD}=\Riem_{AB[CD]}=\Riem_{CDAB}.\end{equation}
    
    The algebraic Bianchi identity \cite{Hohm:2011si,Jurco:2016emw}
    \begin{align}\label{eq:sym}
        \Riem_{[ABC]D}=0.
    \end{align}
    is more involved. Too see this, note that due to the other symmetries this statement is equivalent to $\Riem_{[ABCD]}=0$. First, picking a local frame with $\langle e_A,e_B\rangle$ constant, we calculate
      \begin{equation*}\begin{aligned}
        D_A D_B (e_C)_D
        &=(D_{e_A}D_{e_B}e_C)_D
        -(D_{D_{e_A}e_B} e_C)_D
        =(D_{e_A}(\Gamma_B{}^E{}_C e_E))_D
        -\Gamma_A{}^E{}_B(D_{e_E}e_C)_D \\
        &=\rho(e_A)\Gamma_{BDC}+\Gamma_{ADE}\Gamma\indices{_B^E_C}-\Gamma\indices{_A^E_B}\Gamma\indices{_E_D_C},
      \end{aligned}\end{equation*}
      Using the torsion-free condition $2\Gamma_{[AB]C}+\Gamma_{CAB}=-c_{ABC}$ twice in a row we get
      \begin{equation*}\begin{aligned}
        \Riem_{[DCAB]}&=2(e_{[A|})^E D_E D_{|B} (e_C)_{D]}-\tfrac12(D_E (e_{[A})_B)(D^E(e_C)_{D]})\\
        &=2D_{[A} D_{B} (e_C)_{D]}-\tfrac12\Gamma_{E[BA}\Gamma\indices{^E_{DC]}}\\
        &=-2\rho(e_{[A})\Gamma_{BCD]}-2\Gamma\indices{_{[A}_B^E}\Gamma_{CD]E}+2\Gamma\indices{_E_{[A}_B}\Gamma\indices{_C^E_{D]}}-\tfrac12\Gamma_{E[AB}\Gamma\indices{^E_{CD]}}\\
        &=\tfrac23\rho(e_{[A})c_{BCD]}+(\Gamma\indices{^E_{[A}_B}\Gamma_{CD]E}+c\indices{_{[A}_B^E}\Gamma_{CD]E})\\
        &\quad+2\Gamma\indices{_E_{[A}_B}\Gamma\indices{_C^E_{D]}}+\tfrac12(2\Gamma_{[AB|E|}\Gamma\indices{^E_{CD]}}+c_{[AB|E|}\Gamma\indices{^E_{CD]}})\\
        &=\tfrac23\rho(e_{[A})c_{BCD]}+c\indices{_{[A}_B^E}\Gamma_{CD]E}+\tfrac12c_{[AB|E|}\Gamma\indices{^E_{CD]}}\\
        &=\tfrac23\rho(e_{[A})c_{BCD]}-\tfrac12c\indices{_{[A}_B^E}c_{CD]E},
      \end{aligned}\end{equation*}
      which vanishes by the Jacobi identity \eqref{eq:jacobi}.

  \subsubsection{Curvature operator}
    Let $D$ be a Levi-Civita connection for a generalised metric $\mc G$. Define the \emph{curvature operator}~\cite{CSW1}
\begin{equation}
\label{eq:gen-curvature}
    \curv\colon \Gamma(C_+)\times\Gamma(C_-)\to\on{Der}(E), \qquad \curv(x_+,y_-):=x_+^a y_-^{\bar a}[D_a,D_{\bar a}],
\end{equation}
    where $\on{Der}(E)$ stands for degree 0 derivations of the algebra $\Gamma(\Lambda^\bullet E)$. We can also rewrite this as follows:
    \begin{equation}\curv(x_+,y_-)=x_+^a y_-^{\bar a}([D_{e_a},D_{e_{\bar a}}]-D_{D_{e_a}e_{\bar a}-D_{e_{\bar a}}e_a})=x_+^a y_-^{\bar a}([D_{e_a},D_{e_{\bar a}}]-D_{[e_a,e_{\bar a}]}),\end{equation}
    where we used $x_+^a y_-^{\bar a}\langle De_a,e_{\bar a}\rangle=0$. From this it follows that $\curv(x_+,y_-)$ vanishes on $C^\infty(M)\cong \Gamma(\Lambda^0E)$ and hence is a purely algebraic operator,
    \begin{equation}\curv(x_+,y_-)\in\Gamma(\on{End}(E)).\end{equation}
        Explicitly, we have
    \begin{equation}
        [D_a,D_{\bar a}]^A{}_B=2\Riem_{a\bar a}{}^A{}_B.
    \end{equation}

    Suppose now that we have a torsion-free generalised $G$-structure, where $G\subset O(p)\times O(q)$. Let $D$ be a torsion-free compatible connection and $e_A$ a local $G$-frame. Then an easy calculation shows that
    \begin{equation}
        [D_a,D_{\bar a}]=\rho(e_a)\Gamma_{\bar a}-\rho(e_{\bar a})\Gamma_a+[\Gamma_a,\Gamma_{\bar a}]-c_{a\bar a}{}^A\Gamma_A,\qquad (\Gamma_A)^B{}_C:=\Gamma_A{}^B{}_C.
    \end{equation}
    Since in a $G$-frame we have $\Gamma_A\in\mf g$, we in particular obtain
    \begin{equation}\label{eq:hol}
        R(x_+,y_-)\in\Gamma(\mf g)\subset\Gamma(\on{End}(E)).
    \end{equation}

\subsection{The \complex for torsion-free generalised \texorpdfstring{$G$}{G}-structures}
    Suppose now that we have a torsion-free generalised $G$-structure. Set $\gb:=\Lambda^\bullet E$, $\gb^p:=\Lambda^p E$, and with the corresponding spaces of sections denoted $\gc$ and $\gc^p$, respectively. Just as before, we define $\bx:=\gb/(\mf g\wedge \gb)$ and $\cx:=\Gamma(\bx)$.
  
  Let $D$ be a torsion-free compatible generalised connection. This gives a derivation \begin{equation}\hat D\colon \gc\to \gc,\qquad (\hat D\omega)_{AB\dots C}:=(p+1)D_{[A}\omega_{B\dots C]}\quad \text{for }\omega\in\gc^p.\end{equation}
  In particular we have $\hat D f=\rho^*(df)$. Taking $e_A$ any $G$-frame, the compatibility of $D$ means $\Gamma_{ABC}e^B \otimes e^C\in \mf g\subset \gb$, while torsion-freeness corresponds to $c_{ABC}=-3\Gamma_{[ABC]}$. Note that the antisymmetry in the last two indices gives $\Gamma_{[ABC]}=\tfrac23\Gamma_{[AB]C}+\tfrac13\Gamma_{CAB}$.

We claim that $\hat{D}$ preserves $\mf{g} \wedge \gb$ so that $\hat{D}$ descends to a well-defined operator $\check D$ on $\cx$. This follows exactly as for the ordinary case of section~\ref{sec:Carrion}. The connection $D$ preserves $G$-representations, so that if $\alpha \in \Gamma(\mf{g})$ then $D_{e_A} \alpha \in \Gamma(\mf{g})$ and so $\hat{D} \alpha = e^A \wedge D_{e_A} \alpha \in \Gamma(E \wedge \mf{g})$. The result then follows from the fact that $\hat{D}$ is a derivation of the wedge product on $\gc$.

Denoting equality up to an element of $\Gamma(\mf g\wedge \gb)$ by $\equiv$, we have
  \begin{equation}\hat De_A=\Gamma_{[CB]A}e^C\wedge e^B=(\tfrac32\Gamma_{[CBA]}-\tfrac12\Gamma_{ACB}) e^C\wedge e^B\equiv -\tfrac12 c_{CBA}e^C\wedge e^B,\end{equation}
  which shows that the induced operator $\check D$ is actually independent of the choice of $D$.

    Finally, let us show that $\check D$ is a differential, \begin{equation}\check D^2=0.\end{equation} Since $\hat D^2=\tfrac12[\hat D,\hat D]$ is a derivation, it suffices to check that $\hat D^2f$ and $\hat D^2e_A$ both lie in $\Gamma(\mf g\wedge\gb)$. For the former we have
  \begin{align*}
    \hat D^2f&=\hat D(e^A \rho(e_A)f)\equiv -\tfrac12c\indices{_C_B^A}e^C\wedge e^B \rho(e_A)f-e^A\wedge e^B \rho(e_B)\rho(e_A)f,
  \end{align*}
  which vanishes due to $[\rho(e_B),\rho(e_A)]=\rho([e_B,e_A])=c\indices{_B_A^C}\rho(e_C)$. For the latter,
  \begin{align*}
    \hat D^2e_A&\equiv -\tfrac12 \hat D(c_{CBA}e^C\wedge e^B)=(-\tfrac12\rho(e_D)c_{CBA}+\tfrac12c_{EBA}c\indices{_D_C^E})e^D\wedge e^C\wedge e^B\\
    &=-\tfrac16\rho(e_A)c_{DCB}e^D\wedge e^C\wedge e^B\equiv 0,
  \end{align*}
  where we have used the Jacobi identity \eqref{eq:jacobi}.

Thus we have a natural complex $(\cx, \check{D})$ associated to our torsion-free generalised $G$-structure, which is our \complex in this setting.  

\subsection{The \texorpdfstring{$G_+ \times G_-$}{G+ x G-} double complex and K\"ahler identities}
\label{sec:G+G-}

Suppose we have a Courant algebroid $E$ of signature $(n_+,n_-)$, and a torsion-free $G$-structure, where $G=G_+\times G_-$, with $G_\pm\subset O(n_\pm)$. This induces a positive-definite generalised metric $\mc G\in\on{End}(E)$ and the associated bundle decomposition $E=C_+\oplus C_-$. Let now $D$ be a torsion-free compatible generalised connection. Assume that the trace of $D$ coincides with the divergence w.r.t.\ some volume form $\Phi$ on $M$, i.e.\ for every $u\in\Gamma(E)$ we have
\begin{equation}\label{eq:compatible_with_density}
      D_A u^A=\Phi^{-1}\mc L_{\rho(u)}\Phi.
\end{equation}
  
Note that the pair consisting of a $G$-structure and a volume form $\Phi$ is equivalent to a $G$-structure in the sense of $O(p,q)\times\mb R^+$-geometry of \cite{CSW1}. The torsion-free connection $D$ automatically extends to a corresponding connection in the $O(p,q)\times\mb R^+$-geometry by taking $\Phi$ to be covariantly constant. The condition \eqref{eq:compatible_with_density} is then equivalent to the statement that this new connection is also torsion-free (in the $O(p,q)\times\mb R^+$-sense). Since it is the latter geometric description in which supersymmetry takes the natural form, in the physically relevant examples studied below the condition \eqref{eq:compatible_with_density} will be satisfied.
  
Assuming \eqref{eq:compatible_with_density}, we can now define a positive-definite inner product on each $\gc^p$ by
\begin{equation}
    (\omega,\tau)_{\gc}:=\tfrac1{p!}\int_M \Phi\,\mc G^{AB}\!\dots \mc G^{CD}\omega_{A\dots C}\tau_{B\dots D},
\end{equation}
and we have
  \[(\hat D^\dagger \omega)_{A\dots B}=-\mc G^{CD}D_C\omega_{DA\dots B}.\]
  
Now extend $\mc G$ to a derivation $\hat {\mc G}$ of $\gb$ (with trivial action on $\gb^0$). Defining the bicomplex
\begin{equation}\gc^{p,q}:=\Gamma(\gb^{p,q}),\qquad \gb^{p,q}:=\Lambda^{p}C_+\otimes\Lambda^{q}C_-\end{equation}
we have $\hat{\mc G}\omega=(p-q)\omega$ for any $\omega\in\gc^{p,q}$. Since $D$ is Levi-Civita, we have the decomposition
\begin{equation}\hat D=\hat D_++\hat D_-,\qquad \hat D_+ \gc^{p,q}\subset \gc^{p+1,q},\quad \hat D_- \gc^{p,q}\subset \gc^{p,q+1}.
\end{equation}
  
Since $\mc G$ induces a positive-definite inner product on the whole $\gb$, we can again decompose $\gb$ into 
\begin{equation}
\label{eq:ortho-decomp}
\gb=\bx\oplus\by,\qquad \by:=\mf g\wedge \gb,\qquad \bx:=\by^\perp.
\end{equation}
Both $\bx$ and $\by$ inherit the bigrading from $\gb$. Explicitly, we have
\begin{equation}\bx=\gb^0\oplus\gb^1\oplus\{\omega\in \gb^{\ge2}\mid x^{AB}\omega_{AB\dots C}=0,\;\forall x\in\mf g\}.\end{equation}
We will use $p\colon \gb\to\bx$ and $i\colon \bx\to\gb$ for the orthogonal projection and inclusion, respectively.
  
Define also $\gb_\pm:=\Lambda^\bullet C_\pm$ and $\bx_\pm:=(\mf g_\pm\wedge\gb_\pm)^\perp\subset \gb_\pm$, with $p_\pm\colon\gb_\pm\to\bx_\pm$ denoting the orthogonal projection. Crucially,
  \begin{equation}\bx=\bx_+\otimes\bx_-,\end{equation}
    and so in particular $p=p_+\otimes p_-$.
To see this, we directly calculate
  \begin{align*}
      \bx&=(\mf g\wedge\gb)^\perp=(\mf g_+\wedge\gb_+\wedge\gb_-+\mf g_-\wedge\gb_+\wedge\gb_-)^\perp\\
      &=(\mf g_+\wedge\gb_+\wedge\gb_-)^\perp\cap (\gb_+\wedge \mf g_-\wedge \gb_-)^\perp\\
      &\cong ((\mf g_+\wedge\gb_+)\otimes\gb_-)^\perp\cap (\gb_+\otimes( \mf g_-\wedge \gb_-))^\perp\\
      &=(\bx_+\otimes \gb_-)\cap (\gb_+\otimes \bx_-)=\bx_+\otimes\bx_-.
  \end{align*}

Note that $\hat D$ preserves $\cy$, since the adjoint action of $\mf g$ preserves $\mf g\subset \gb^2$. Consequently, $\hat D^\dagger$ preserves $\cx$. 
Using the orthogonal decomposition~\eqref{eq:ortho-decomp}, we can write the differential $\check{D}$ on $\cx$ now via
\begin{equation}\check D:=p\circ \hat D\circ i\colon \cx\to\cx.\end{equation}
  
We now define the Laplacian \begin{equation}\hat\Delta:=\hat D\hat D^\dagger+\hat D^\dagger \hat D\colon \gc\to\gc.\end{equation}
A quick calculation reveals that for $\omega\in\gc^p$ we have
\begin{equation}\hat\Delta \omega_{AB\dots C}=-\mc G^{DE}D_D D_E \omega_{AB\dots C}+p\,\mc G^{DE}[D_D,D_{[A}]\omega_{|E|B\dots C]}.\end{equation}
Let us now show that if $\mc G$ is Ricci flat then the operator
\begin{equation}\label{eq:laplacehat_vs_bigrading}
        p\circ \hat\Delta\circ i\colon \cx\to\cx
\end{equation} preserves the bigrading. This is equivalent to showing that in the Ricci flat case
\begin{equation}\label{eq:laplacehat_vs_bigrading_different}
    [\hat{\mc G},\hat\Delta]\,\cx\subset\cy.
\end{equation}
We start by calculating
\begin{align*}
    [\hat {\mc G},\hat\Delta]\omega&=\tfrac p{p!}(\mc G^{DE}[D_D,D_{F}]\mc G\indices{^F_{A}}\omega_{EB\dots C}-\mc G^{DE}[D_D,D_{A}]\omega_{FB\dots C}\mc G\indices{^F_E})e^A\wedge e^B\wedge \dots\wedge e^C\\
    &=\tfrac p{p!}(\mc G\indices{^D_E}[D^E,D_{F}]\mc G\indices{^F_{A}}\omega_{DB\dots C}-[D^D,D_{A}]\omega_{DB\dots C})e^A\wedge e^B\wedge \dots\wedge e^C\\
    &=-\tfrac{2p}{p!} \left[([D^c,D_{\bar a}]\omega_{cB\dots C})e^{\bar a}\wedge e^B\wedge \dots\wedge e^C+([D^{\bar c},D_{a}]\omega_{\bar cB\dots C})e^a\wedge e^B\wedge \dots\wedge e^C\right]\\
      &=\tfrac{4p(p-1)}{p!} (R\indices{^c_{\bar a}^D_B}\omega_{cD\dots C}e^{\bar a}\wedge e^B\wedge \dots\wedge e^C+R\indices{^{\bar c}_a^D_B}\omega_{\bar cD\dots C}e^a\wedge e^B\wedge \dots\wedge e^C),
      \end{align*}
      where in the last line we have used $R\indices{^c_{\bar a}^D_c}=0=R\indices{^{\bar c}_a^D_{\bar c}}$ (note that $R_{a\bar bc\bar d}$ automatically vanishes). Continuing, and using $R_{A[BC]D}=-\tfrac12R_{AD BC}$ by \eqref{eq:sym}, we have
      \begin{align*}
      R\indices{^c_{\bar a}^D_B}&\omega_{cD\dots C}e^{\bar a}\wedge e^B
      +R\indices{^{\bar c}_a^D_B}
      	\omega_{\bar cD\dots C}e^a\wedge e^B\\
      &=R\indices{^c_{\bar a}^d_b}\omega_{cd\dots C}e^{\bar a}\wedge e^b
      +R\indices{^c_{\bar a}^{\bar d}_{\bar b}}\omega_{c\bar d\dots C}e^{\bar a}\wedge e^{\bar b} \\
      &\qquad\qquad
      +R\indices{^{\bar c}_a^d_b}\omega_{\bar cd\dots C}e^a\wedge e^b
      +R\indices{^{\bar c}_a^{\bar d}_{\bar b}}\omega_{\bar c\bar d\dots C}e^a\wedge e^{\bar b}\\
      &=-R\indices{_{\bar a}^{[cd]}_b}\omega_{cd\dots C }e^{\bar a}\wedge e^b
      -R\indices{^c_{[\bar a\bar b]}^{\bar d}}\omega_{c\bar d\dots C}e^{\bar a}\wedge e^{\bar b} \\
      &\qquad\qquad
      -R\indices{^{\bar c}_{[ab]}^d}\omega_{\bar cd\dots C}e^a\wedge e^b
      -R\indices{_a^{[\bar c\bar d]}_{\bar b}}\omega_{\bar c\bar d\dots C}e^a\wedge e^{\bar b}\\
      &=\tfrac12( R\indices{_{\bar a}_b^{cd}}\omega_{cd\dots C}e^{\bar a}\wedge e^b
      +R\indices{^c^{\bar d}_{\bar a\bar b}}\omega_{c\bar d\dots C}e^{\bar a}\wedge e^{\bar b} \\
      &\qquad\qquad
      +R\indices{^{\bar cd}_{ab}}\omega_{\bar cd\dots C}e^a\wedge e^b
      +R\indices{_a_{\bar b}^{\bar c\bar d}}\omega_{\bar c\bar d\dots C}e^a\wedge e^{\bar b})\\
      &=\tfrac12(R\indices{_{\bar a}_b^{AB}}\omega_{AB\dots C}e^{\bar a}\wedge e^b+R\indices{^c^{\bar d}_{AB}}\omega_{c\bar d\dots C}e^{A}\wedge e^{B})=\tfrac12R\indices{^c^{\bar d}_{AB}}\omega_{c\bar d\dots C}e^{A}\wedge e^{B}\equiv0,
      \end{align*}
      where we have used the fact that $R_{a\bar aAB}e^A\wedge e^B\sim (\curv(e_a,e_{\bar a})_{AB})e^A\wedge e^B\in\Gamma(\mf g)$ due to \eqref{eq:hol}.
This proves \eqref{eq:laplacehat_vs_bigrading} and \eqref{eq:laplacehat_vs_bigrading_different}.

Let us also define another Laplacian
\begin{equation}\check \Delta:=\check D\check D^\dagger+\check D^\dagger \check D\colon \cx\to\cx.\end{equation}
Using the previous calculation, it is not difficult to show that $\check\Delta$ also preserves the bigrading of $\cx$:
    
We start by noting that
\begin{equation}p\circ \hat\Delta \circ i-\check\Delta=p\circ \hat D \circ (1-i\circ p)\circ \hat D^\dagger \circ i+p\circ \hat D^\dagger \circ (1-i\circ p)\circ \hat D\circ i=p\circ \hat D^\dagger \circ (1-i\circ p)\circ \hat D\circ i,\end{equation}
since $\hat D^\dagger$ preserves $\cx$. Thus it suffices to show that $[\hat{\mc G},p\hat{D}^\dagger(1-ip)\hat{D}i]=0$. Since 
\begin{equation}\hat D= e^A\wedge D_{e_A},\quad \hat D^\dagger =- \mc G^{AB}i_{e_A}D_{e_B},\quad [\hat{\mc G},\hat D]= \mc G^{AB}e_A\wedge D_{e_B},\quad [\hat{\mc G},\hat D^\dagger]=i_{e^A}D_{e_A},\end{equation}
we can write for any $\omega\in\cx$ (dropping $i$'s)
\begin{align*}
      [\hat{\mc G},p\hat{D}^\dagger(1-p)\hat{D}]\omega&=p[\hat{\mc G},\hat{D}^\dagger](1-p)\hat{D}\omega+p\hat{D}^\dagger(1-p)[\hat{\mc G},\hat{D}]\omega\\
      &=pi_{e^A}(1-p)D_{e_A}(e^B\wedge D_{e_B}\omega)\\
      &\quad-p\mc G^{AB}i_{e_A}(1-p)\mc G^{CD}D_{e_B}(e_C\wedge D_{e_D}\omega)\\
      &= 2pi_{e^a}(1-p)D_{e_a}(e^{\bar a}\wedge D_{e_{\bar a}}\omega)+2pi_{e^{\bar a}}(1-p)D_{e_{\bar a}}(e^a\wedge D_{e_a}\omega).
\end{align*}
This vanishes since for any $u_+\in C_+$, $v_-\in C_-$, $\tau=\tau_+\wedge \tau_-\in \bx=\bx_+\wedge\bx_-$ we have
\begin{align*}
      pi_{u_+}(1-p)(v_-\wedge \tau)&\sim pi_{u_+}(1-p_+p_-)(\tau_+\wedge (v_-\wedge \tau_-))\\
      &=pi_{u_+}(\tau_+\wedge (1-p_-)(v_-\wedge \tau_-))\\
      &=(p_+i_{u_+}\tau_+)\wedge p_-(1-p_-)(v_-\wedge\tau_-)=0,
\end{align*}
and similarly when we exchange $+$ and $-$. Thus $\check\Delta$ also preserves the bigrading.

As $D$ preserves the generalised metric, we can write $\check{D} = \dd_+ + \dd_-$ with 
\begin{equation}
	\dd_+\colon \cx^{p,q} \longrightarrow \cx^{p+1,q}
	\hs{30pt}
	\dd_-\colon \cx^{p,q} \longrightarrow \cx^{p,q+1}
\end{equation}
Similarly $\check{D}^\dagger = \dd_+^\dagger + \dd_-^\dagger$ with
\begin{equation}
	\dd_+^\dagger\colon \cx^{p,q} \longrightarrow \cx^{p-1,q}
	\hs{30pt}
	\dd_-^\dagger\colon \cx^{p,q} \longrightarrow \cx^{p,q-1}
\end{equation}
and we have
\begin{equation}
	\check{\Delta} = \{ \dd_+ , \dd_+^\dagger \} + \{ \dd_- , \dd_-^\dagger \} 
		+ \{ \dd_+ , \dd_-^\dagger \} + \{ \dd_- , \dd_+^\dagger \}
\end{equation}
The statement that $\check{\Delta}$ preserves the bigrading of $\cx$ corresponds precisely to the K\"ahler identities
\begin{equation}
	\{ \dd_+ , \dd_-^\dagger \} = \{ \dd_- , \dd_+^\dagger \} = 0
\end{equation}

\subsection{Examples}

\subsubsection{$G \times G$ structures in type II}\label{sec:GxG_in_type_II}

For $U(3) \times U(3)$, $G_2 \times G_2$ and $\Spin(7) \times \Spin(7)$ these double complexes appeared in~\cite{Ashmore:2021pdm} in the context of the type II topological string. We shall review the structure of the complexes here before reviewing the application to topological strings in section \ref{sec:applications}.

These generalised $G$-structures were originally used to describe supersymmetric backgrounds of type II \cite{Grana:2005sn,Claus_Jeschek_2005}. Indeed, a background preserving $\mathcal{N}=2$ supersymmetry in an NSNS background requires two internal $O(d)$ spinors $\epsilon^{\pm}$ which are parallel with respect to the Bismut connections $\nabla^{\pm}= \nabla \pm \tfrac{1}{2}H$, respectively. These two spinors define two different $SU(3)$, $G_{2}$, or $Spin(7)$ structures in 6, 7, and 8 dimensions respectively. Lifting this to generalised geometry, the two spinors transform with respect to the left and right groups in the generalised metric structure \begin{equation}O(d,d)\to O(d)_{+}\times O(d)_{-}.\end{equation} They define, respectively, $SU(3)\times SU(3)$, $G_{2}\times G_{2}$, or $Spin(7)\times Spin(7)$ structures, and the Killing spinor equations guarantee that one can choose some generalised Levi-Civita connection $D$ such that $D\epsilon^{\pm}=0$. That is, supersymmetry guarantees that the reduced $G$-structure is integrable.

For these cases, the \complex is a doubled version of Carri\'on's gauge theory \complex from section~\ref{sec:Carrion}. 
As discussed in section~\ref{sec:Carrion}, for on a torsion-free $\SU(3)$ structure it is natural to take the complex~\eqref{eq:Dolbeault-complex-anti} which is half of that for $U(3)$ with the scalar part complexified. Correspondingly, for $\SU(3)\times\SU(3)$ structures in generalised geometry we take the complex for $U(3)\times U(3)$, also with complexified scalar. 
For the $G_{2}$ and $Spin(7)$ cases, it will be convenient to enhance the notation and explicitly write the representations that appear. Namely, we write
\begin{equation}
    \mathcal{A}^{p,q}_{\mathbf{r},\mathbf{s}} = \Gamma(\Lambda^{p}_{\mathbf{r}}C_{+} \otimes \Lambda^{q}_{\mathbf{s}}C_{-})
\end{equation}
where $\mathbf{r},\mathbf{s}$ are the representations appearing in the gauge theory \complexesNoSpace. 

We can then write the $G_{2}$ complex as
\begin{equation}\label{eq:double_complex}
    \begin{tikzcd}[column sep = tiny, row sep = tiny]
    & & \arrow{dl}[swap]{\dd_{+}} & \cA^{0,0}_{\rep{1},\rep{1}} \arrow[dl] \arrow[dr] & \arrow{dr}{\dd_{-}} & & \\
    &\left.\right. & \cA^{1,0}_{\rep{7},\rep{1}} \arrow[dl] \arrow[dr] & &  
    \cA^{0,1}_{\rep{1},\rep{7}} \arrow[dl] \arrow[dr] & \left.\right. & \\
    & \cA^{2,0}_{\rep{7},\rep{1}} \arrow[dl] \arrow[dr] & & 
    \cA^{1,1}_{\rep{7},\rep{7}} \arrow[dl] \arrow[dr] & &
    \cA^{0,2}_{\rep{1},\rep{7}} \arrow[dl] \arrow[dr] & \\
    \cA^{3,0}_{\rep{1},\rep{1}} \arrow[dr] & &
    \cA^{2,1}_{\rep{7},\rep{7}} \arrow[dl] \arrow[dr] & & 
    \cA^{1,2}_{\rep{7},\rep{7}} \arrow[dl] \arrow[dr] & &
    \cA^{0,3}_{\rep{1},\rep{1}} \arrow[dl] \\
    & \cA^{3,1}_{\rep{1},\rep{7}} \arrow[dr] & &
    \cA^{2,2}_{\rep{7},\rep{7}} \arrow[dl] \arrow[dr] & &
    \cA^{1,3}_{\rep{7},\rep{1}} \arrow[dl] & \\
    & & \cA^{3,2}_{\rep{1},\rep{7}} \arrow[dr] & & 
    \cA^{2,3}_{\rep{7},\rep{1}} \arrow[dl]  & & \\
    & & & \cA^{3,3}_{\rep{1},\rep{1}} & & &
    \end{tikzcd}
\end{equation}
and the $Spin(7)$ complex as
\begin{equation}\label{eq:Spin(7)_double_complex}
    \begin{tikzcd}[column sep = small, row sep = small]
     & \arrow{dl}[swap]{\dd_{+}} & \cA^{0,0}_{\rep{1},\rep{1}} \arrow[dl] \arrow[dr] & \arrow{dr}{\dd_{-}} &  \\
    \left. \right. & \cA^{1,0}_{\rep{8},\rep{1}} \arrow[dl] \arrow[dr] & &  
    \cA^{0,1}_{\rep{1},\rep{8}} \arrow[dl] \arrow[dr] & \left. \right.  \\
    \cA^{2,0}_{\rep{7},\rep{1}} \arrow[dr] & & 
    \cA^{1,1}_{\rep{8},\rep{8}} \arrow[dl] \arrow[dr] & &
    \cA^{0,2}_{\rep{1},\rep{7}} \arrow[dl]  \\
    & \cA^{2,1}_{\rep{7},\rep{8}}  \arrow[dr] & &
    \cA^{1,2}_{\rep{8},\rep{7}} \arrow[dl]  &  \\
 & & \cA^{2,2}_{\rep{7},\rep{7}} & & 
\end{tikzcd}
\end{equation}

In the case of $SU(3)$ structures, there is a subtlety. The gauge theory \complexNoSpace~\eqref{eq:Dolbeault-complex-anti} in this case is isomorphic to the Dolbeault complex, and we may choose whether this is $(\Omega^{p,0},\del)$ or $(\Omega^{0,p},\delb)$. In section \ref{sec:Carrion}, we made an arbitrary choice to associate it with the antiholomorphic complex. When one combines two copies of the complex, however, there are two inequivalent choices which correspond to whether we take two antiholomorphic copies of the Dolbeault complex, or one holomorphic and one antiholomorphic. We therefore have two non-isomorphic \complexes for $SU(3)\times SU(3)$ structures given by
\begin{align}
    \cA^{p,q}_{A} &= \Gamma( \Lambda^{p,0}C_{+} \otimes \Lambda^{0,q}C_{-}) \\
    \cA^{p,q}_{B} &= \Gamma(\Lambda^{0,p}C_{+}\otimes \Lambda^{0,q}C_{-})
\end{align}
In \cite{Ashmore:2021pdm}, they named these two choices the $A$- and $B$-complexes, and indeed they are relevant for the topological A- and B-models respectively, as we will review later.

In \cite{Kapustin:2004gv}, this structure was related to the generalised K\"ahler geometry associated to the $SU(3)\times SU(3)$ structure. In that case, one has two generalised complex structures $\mathcal{J}_{i}$ which commute and define a generalised metric $\mc G=-\mathcal{J}_{1}\mathcal{J}_{2}$. Since they commute, they define a decomposition of the complexification of $E=T\oplus T^{*}$ into simultaneous $(\pm i, \pm i) $ eigenspaces. Moreover, since they commute with the generalised metric, this decomposition respects the generalised metric decomposition into $C_{\pm}$. In particular, the simultaneous eigenspaces are precisely the spaces $C^{1,0}_{\pm}, C^{0,1}_{\pm}$, and one can show that the $+i$ eigenspaces $L_{i}^{+}$ of $\mathcal{J}_{i}$ satisfy
\begin{equation}
    L_{1}^{+} = C^{1,0}_{+} \oplus C^{1,0}_{-} \ , \qquad L_{2}^{+} = C^{1,0}_{+} \oplus C^{0,1}_{-}\ .
\end{equation}
Therefore, the $A$- and $B$-complexes have the following total spaces
\begin{equation}
    \bigoplus_{p,q}\cA^{p,q}_{A} = \bigoplus_{n} \Gamma(\Lambda^{n} \overline{L_{2}^{+}}) \ , \qquad \bigoplus_{p,q}\cA^{p,q}_{B} = \bigoplus_{n} \Gamma(\Lambda^{n} \overline{L_{1}^{+}})
\end{equation}
and correspond to a refinement of the Dolbeault complex associated to the generalised complex structure into a double complex \cite{Gualtieri-thesis,Gualtieri:2010fd}. The two inequivalent choices correspond to the two different generalised complex structures in a generalised K\"ahler structure.

\subsubsection{$G \times \SO(d+n)$ structures for heterotic}

For heterotic structures, we have the groups $G \times \SO(d+n) \subset \SO(d,d+n)$, which are of the type discussed in section~\ref{sec:G+G-}, but are not spinor type. 
In these cases the generalised Dolbeault complex takes the form of the tensor product of the gauge theory \complex for the group $G$ on the left (which is of spinor type in many interesting cases) with the trivial one~\eqref{eq:trivial-ord-complex} on the right:
\begin{equation}\label{eq:het-double-complex}
    \begin{tikzcd}[column sep = tiny, row sep = tiny, cells={text width={width("$\Omega^{+,2}_{\rep{r_2}}(C_-) $")}, align=center}]   
    & & & & \arrow[dr, "\dd_{-}", xshift=1ex, yshift=-1ex]  \\
    & & & \arrow{dl}[swap]{\dd_{+}} 
    & \Omega^{+,0} \arrow[dl] \arrow[dr] 
    & \left.\right. 
    & \left.\right.& \\
    & &\left.\right. 
    & \Omega^{+,1} \arrow[dl] \arrow[dr] 
    & &  \Omega^{+,0}(C_-) \arrow[dl] & \left.\right. 
    & \\
    & & \Omega^{+,2}_{\rep{r_2}}  \arrow[dl] \arrow[dr] 
    & & \Omega^{+,1}(C_-) \arrow[dl] 
    \\
    & \Omega^{+,3}_{\rep{r_3}} \arrow[dr] \arrow[dl]& &
    \Omega^{+,2}_{\rep{r_2}}(C_-) \arrow[dl] 
    \\
    \phantom{\Omega^{+,3}_{\rep{r_3}}} & & \Omega^{+,3}_{\rep{r_3}}(C_-) \arrow[dl]
    \\
    & \phantom{\Omega^{+,3}_{\rep{r_3}}(C_-) }
    \end{tikzcd}
\end{equation}
where $\Omega^{+,\bullet}=\Gamma(\Lambda^\bullet C_+)$ and the diagram continues to the lower left as far as needed. We identify the bundles called $Q$ in refs~\cite{delaOssa:2014cia,delaOssa:2017pqy} with $C_-$ (or its complexification in the $\SU(N)$ case). The lower row of~\eqref{eq:het-double-complex} then forms the various differential complexes involving $Q$ written down in those references. Note that $C_- \simeq TM \oplus \End(V)$ and in the case $G = \SU(N)$ we have $(C_-)_\bbC \simeq TM_\bbC \oplus \End(V)_\bbC \simeq T^{(1,0)} \oplus T^{*(1,0)} \oplus \End(V)_\bbC$, which is the bundle $Q$ discussed in~\cite{delaOssa:2014cia}. The upper row forms part of the gauge structure needed to write BV actions for such theories.

\subsection{Spinor type complexes in \texorpdfstring{$O(d,d)\times\bbR^+$}{O(d,d)xR+} generalised geometry}
\label{sec:Odd-spinor}
We now specialise to the case of $O(d,d)$ generalised geometry relevant to the NS-NS fields of type II theories as in~\cite{CSW1}, where the generalised tangent space is an exact Courant algebroid. We will see that these geometries admit analogues of the spinor type complexes of section~\ref{sec:spinor-type}, where the \complex becomes a sum of spinors of $O(d,d)$ decomposed under the structure group. 
In $O(d,d)$ generalised geometry, one can have spinor type complexes for structure groups which do not preserve a generalised metric. For example, the structure groups $U(N,N) \subset \SO(2N,2N)$ relevant to generalised complex structures have this property. 
However, here we will focus on the case where the structure group has the form $G_+ \times G_-$ as in section~\ref{sec:G+G-} and the corresponding generalised metric is positive definite and generalised Ricci flat. 
In these cases, the complexes can be viewed as tensor products of two gauge theory spinor type complexes. Similarly to what was found in section~\ref{sec:spinor-type}, the spinor type property leads to relations between the various Laplacians appearing in the construction. In particular, we show that $\Delta_+ = \Delta_-$ for these cases, which when combined with the K\"ahler identities is sufficient to prove a $\der\bar\der$-type lemma. We also construct a different Laplacian operator which is shown to agree with the $H$-twisted de Rham Laplacian on the polyform representation of the spinors. We start by reviewing some features of spinors in $O(d,d)\times\bbR^+$ generalised geometry, and how the additional $\bbR^+$ factor in the structure group enables one to represent the spinors as polyforms.

First, let us recall some features of the spinor bundle $S(E)$ in $O(d,d)\times\bbR^+$ generalised geometry (see~\cite{CSW1} for full details) and the associated Clifford algebra $\Cliff(d,d;\bbR)$. The spinor bundle (with zero $\bbR^+$ weight) is isomorphic to the bundle $(\det T^*)^{-1/2} \otimes \Lambda^\bullet T^*$ whose sections are weighted polyforms. As such, we will here assume that we have a section $\Phi \in \Gamma(\det T^*)$ which gives us an isomorphism to weighted spinors $S(E)_{1/2}$ which can be represented directly as polyforms under the relevant $\GL(d,\bbR)$ subgroup of $O(d,d)\times\bbR^+$
\begin{equation}
	S(E) \stackrel{\Phi}{\simeq} S(E)_{1/2} \simeq \Lambda^\bullet T^*M
\end{equation}
In applications to physics, we will take $\Phi = \sqrt{g} \ee^{-2\phi}$ to be the natural string frame integration measure. We will also assume that our $O(d,d)\times\bbR^+$ connections are compatible with the density $\Phi$ such that we need not distinguish carefully between weighted and un-weighted spinors  when acting with our differential operators built from these connections. Recall that we also noted this as our assumption at the start of section~\ref{sec:G+G-}. 

We denote the generators of the Clifford algebra $\Cliff(d,d;\bbR)$ by $\Gamma^A$. In the cases we examine, the generalised structure group will always be contained in the maximal compact subgroup, so that it will define a generalised metric. We can thus consider the decomposition of the spinors and Clifford algebra under $\Spin(d) \times \Spin(d)$. As such, we decompose the index $A \ra (a,\ba)$, the vector indices for $C_{\pm}$ and let the matrices $\gamma^a$ and $\gamma^{\ba}$ be two sets of generators for $\Cliff(d,\bbR)$, again one for each of $C_\pm$. There are then several possible cases for the decomposition of the spinors and the matrices $\Gamma^A$, depending on the dimension $d$, as we review in appendix~\ref{app:Gamma}. However, these differ only slightly in form. The decomposition of the spinor always has the form
\begin{equation}
	\Psi = \sum \, \zeta_+ \otimes \zeta_- \otimes t
\end{equation}
where $\zeta_\pm \in S(C_\pm)$ and $t=1$ for $d$ even or a constant vector in a two dimensional auxiliary space for $d$ odd. There are similar tensor product decompositions of the matrices $\Gamma^A$ acting on the decomposed spinor, as we describe in more detail in appendix~\ref{app:Gamma}. For our purposes here, the important point to extract is that in all of these cases we have the decompositions
\begin{equation}
\label{eq:Gamma-2-decomp}
\begin{aligned}
	\Gamma^a \Gamma^b &= \eta^{ab} \id + \Gamma^{ab} 
		= (g^{ab} \id + \gamma^{ab})\otimes \id \otimes \id
		= (\gamma^a \gamma^b)\otimes \id\otimes \id \\
	\Gamma^{\ba} \Gamma^{\bb} &= \eta^{\ba\bb} \id + \Gamma^{\ba\bb} 
		= \id \otimes (-g^{\ba\bb} \id - \gamma^{\ba\bb})\otimes \id
		= -\id \otimes  (\gamma^{\ba} \gamma^{\bb})\otimes \id		
\end{aligned}
\end{equation}
This means that if we have two spinors $\zeta_\pm \in S(C_\pm)$ and we embed the tensor product $\zeta_+ \otimes \zeta_-$ into $S(E)$ by tensoring it with some constant auxiliary vector $t$ as above, then acting with $\Gamma^a \Gamma^b$ on it corresponds to acting on $\zeta_+ \otimes \zeta_-$ with $ (\gamma^a \gamma^b)\otimes \id$ and similarly acting with $\Gamma^{\ba} \Gamma^{\bb}$ on it corresponds to acting on $\zeta_+ \otimes \zeta_-$ with $ -\id \otimes  (\gamma^{\ba} \gamma^{\bb})$.

Finally, we recall from~\cite{CSW1} that for a torsion-free generalised connection the $O(d,d)$ Dirac operator $\slashed{D} = \Gamma^A D_A$ acts on the polyform presentation of a weighted spinor $\Psi$ via the exterior derivative. In fact, this is true in the ``twisted" picture" of generalised geometry in which a section of $S(E)_{1/2}$ is in fact a collection of local polyforms on patches of the space related by gauge transformations. One can also work in the ``untwisted picture" in which we use the $B$-field specified by the generalised metric to define a global polyform corresponding to the spinor. In this picture, the Dirac operator becomes $\dd_{H} = \dd + (H \wedge)$. In any presentation, one sees that the Dirac operator squares to zero 
\begin{equation}
\label{eq:Dirac-sqaured-zero}
	\slashed{D}^2 = (\Gamma^A D_A)^2 = 0 \ .
\end{equation}

We now use these facts to derive properties of the Laplacians on spinor type complexes as follows. First, we examine~\eqref{eq:Dirac-sqaured-zero} in terms of $\Cliff(d,d,\bbR)$ $\Gamma$-matrices with indices split under $O(d)\times O(d)$ acting on a generalised spinor
\begin{equation}
\label{eq:d-squared}
	(\Gamma^A D_A)^2 = (\Gamma^a D_a + \Gamma^{\ba} D_{\ba})^2 
		=  (\Gamma^a D_a)^2 + (\Gamma^{\ba} D_{\ba})^2 
			+ \Gamma^a \Gamma^{\ba} [D_a, D_{\ba}]
\end{equation}
The last term here can be expressed as the $\Spin(d) \times \Spin(d)$ action of the generalised curvature operator~\eqref{eq:gen-curvature} on the generalised spinor $\Psi$ which is given by
\begin{equation}
	[D_a, D_{\ba}] \Psi
	= \frac14 R_{a\ba bc} \Gamma^{bc} \Psi + \frac14 R_{a\ba \bb\bc} \Gamma^{\bb\bc} \Psi
\end{equation}
since in terms of $\Cliff(d,d,\bbR)$ $\Gamma$-matrices, $\Spin(d) \times \Spin(d)$ is generated by $\Gamma^{ab}$ and $\Gamma^{\ba\bb}$. Thus contracting with the additional $\Gamma$-matrices in~\eqref{eq:d-squared} we have 
\begin{equation}
\begin{aligned}
	\Gamma^a \Gamma^{\ba} [D_a, D_{\ba}] \Psi
	&= \frac14 \Gamma^{\ba}  ( R_{\ba a bc} \Gamma^a \Gamma^{bc}) \Psi 
		+ \frac14 \Gamma^a (R_{a\ba \bb\bc} \Gamma^{\ba} \Gamma^{\bb\bc}) \Psi \\
	&= \frac14 \Gamma^{\ba}  (R_{\ba [abc]} \Gamma^{abc} 
		+ 2 \eta^{ab} R_{\ba abc} \Gamma^{c}) \Psi \\
	&\qquad \qquad + \frac14 \Gamma^a  (R_{a[\ba \bb\bc]} \Gamma^{\ba\bb\bc} 
		+ 2\eta^{\ba\bb}R_{a\ba \bb\bc}\Gamma^{\bc}) \Psi \\
	&= 0
\end{aligned}
\end{equation}
where in the last step we have used the algebraic Bianchi identity~\eqref{eq:sym} and that $\eta^{ab} R_{\ba abc} \sim R_{\ba c} = 0$ and $\eta^{\ba\bb}R_{a\ba \bb\bc} \sim R_{a\bc}=0$ on a generalised Ricci flat manifold. So we have
\begin{equation}
\label{eq:d-squared-Laplacians}
	(\Gamma^a D_a)^2 + (\Gamma^{\ba} D_{\ba})^2 = 0
\end{equation}

Next we show that the terms $(\Gamma^a D_a)^2 \Psi$ and $(\Gamma^{\ba} D_{\ba})^2 \Psi$ become $\Delta_+$ and $-\Delta_-$ when seen as actions on an element of the double complex. 

From~\eqref{eq:Gamma-2-decomp}, on a spinor $\Psi = \zeta_+ \otimes \zeta_- \otimes t$, we have that 
\begin{equation}
	(\Gamma^a D_a)^2 \Psi 
		= \Big[ (\gamma^a D_a \otimes \id)^2 (\zeta_+ \otimes \zeta_-)\Big]  \otimes t
\end{equation}
The operator $(\gamma^a D_a)$ acting on the spinor $\zeta_+$, viewed as an element of the spinor type complex $\cx_+$, is the Dirac operator $\mc{D}_+$, the analogue of~\eqref{eq:Dirac-complex} for $\dd_+$, on $\cx_+$. When tensored with the $C_-$ spinor $\zeta_-$, the action of the generalised connections on $\zeta_-$ will be precisely such that $(\gamma^a D_a \otimes \id)$ acts on the tensor product as $\mc{D}_+$ on $\cx$. Thus 
\begin{equation}
	(\Gamma^a D_a)^2 \Psi = \Delta_+ \Psi
\end{equation}
Similarly, 
\begin{equation}
	(\Gamma^{\ba} D_{\ba})^2 \Psi 
		= - \Big[ (\id \otimes \gamma^{\ba} D_{\ba})^2 (\zeta_+ \otimes \zeta_-)\Big]  \otimes t
		= - \Delta_- \Psi
\end{equation}
so that~\eqref{eq:d-squared-Laplacians} becomes
\begin{equation}
\label{eq:Laplacians-equal}
	\Delta_+  = \Delta_-
\end{equation}
for the double complex $\cx$.

The equality \eqref{eq:Laplacians-equal} and the K\"ahler identities in particular imply the Hodge decomposition on compact spaces for each of $\check D$, $\dd_+$, $\dd_-$. This can be stated as 
\begin{equation}
\label{eq:filtered-Hodge}
	\cx^{p,q} = \dd_+ \dd_- \cx^{p-1,q-1} \oplus \dd_+^\dagger \dd_- \cx^{p+1,q-1} 
		\oplus \dd_+ \dd_-^\dagger \cx^{p-1,q+1} \oplus \dd_+^\dagger \dd_-^\dagger \cx^{p+1,q+1} \oplus \mc{H}^{p,q}
\end{equation}
From this it follows that we also have the $\der\bar\der$-lemma. For completeness, we repeat the proof here as it appears in K\"ahler geometry.

\begin{lemma}[\unboldmath{$\del\delb$}-lemma]
   Let $\alpha\in\cA^{p,q}$ be $\check{D}$-closed, where $\check{D}=\dd_++\dd_-$. Then the following are equivalent
   \begin{enumerate}[(a)]
       \item $\alpha$ is $\check{D}$-exact.
       \item $\alpha$ is $\dd_+$-exact.
       \item $\alpha$ is $\dd_-$-exact.
       \item $\alpha$ is $\dd_+\dd_-$-exact.
       \item $\alpha$ is orthogonal to the harmonic forms ${\cal H}^{(p,q)}$.
   \end{enumerate}
\end{lemma}

\begin{proof}
    Clearly $(d)$ implies $(a)$, $(b)$ and $(c)$. The usual Hodge decomposition also gives that any of $(a)$, $(b)$ and $(c)$ imply $(e)$. To show that $(e)$ implies $(d)$, note that as $\alpha$ is $\check{D}$-closed of pure type, $\alpha$ is also $\dd_+$-closed and $\dd_-$-closed. Assuming $(e)$, and using the $\dd_+$-Hodge decomposition, it follows that $\alpha$ is $\dd_+$-exact,
    \begin{equation}
        \alpha=\dd_+\eta\:,
    \end{equation}
for $\eta\in\cA^{p-1,q}$. But $\eta$ can be decomposed under the $\dd_-$-Hodge decomposition as 
\begin{equation}
\eta =  \dd_-\gamma + \dd_-^\dagger \gamma' + \theta \ .
\end{equation} 
Using the K\"ahler identities, we see that only the $\dd_-$-exact part of $\eta$ contributes to $\alpha$:
\begin{equation}
0 = \dd_- \alpha = \dd_- \dd_+ \eta 
= - \dd_- \dd_-^\dagger \dd_+ \gamma'
\quad \Ra \quad 
\dd_-^\dagger \dd_+ \gamma' = 0 \ .
\end{equation} 
We therefore have
    \begin{equation}
        \alpha=\dd_+\dd_-\gamma\:,
    \end{equation}
    for $\gamma\in\cA^{p-1,q-1}$.
\end{proof}

Next, consider the Pin-cover $s$ of the generalised metric, viewed as an endomorphism. This has
\begin{equation}
	\mc G^A{}_B (s \Gamma^B s^{-1}) = \Gamma^A
	\hs{10pt} \Ra \hs{10pt} s \Gamma^A s^{-1} = \mc G^A{}_B \Gamma^B
\end{equation}
and thus
\begin{equation}
	s \Gamma^A D_A s^{-1} \Psi = G^{AB} \Gamma_A D_B \Psi
		= (\Gamma^a D_a - \Gamma^{\ba} D_{\ba}) \Psi
		:= \slashed{D}^{(\mc G)} \Psi
\end{equation}

Viewing the $O(d,d)$ spinor as a polyform, it is known~\cite{CSW1} that 
\begin{equation}
	s = \left\{ \begin{matrix} 
	\Gamma^{(+)} = \tfrac{1}{d!} \epsilon^{a_1 \dots a_d} 
		\Gamma_{a_1 \dots a_d} & \hs{10pt} d\text{ odd} \\
	\Gamma^{(-)} = \tfrac{1}{d!} \epsilon^{\ba_1 \dots \ba_d} 
		\Gamma_{\ba_1 \dots \ba_d} &  \hs{10pt}d\text{ even} 		
	\end{matrix} \right. 
\end{equation}
and that on a form $\psi_k$ of degree $k$
\begin{equation}
	\Gamma^{(+)} \psi_k = (-1)^{\lfloor n/2 \rfloor} * \psi_k
	\hs{30pt}
	\Gamma^{(-)} \psi_k = (-1)^d (-1)^{\lfloor (n+1)/2 \rfloor} * \psi_k
\end{equation}
Putting all these together, one finds that 
\begin{equation}
\label{eq:mass-shell}
	(\Delta_+ + \Delta_-) \Psi= (\Gamma^a D_a)^2 \Psi - (\Gamma^{\ba} D_{\ba})^2 \Psi
	= \tfrac12 \{ \slashed{D}, \slashed{D}^{(\mc G)} \} \Psi = \tfrac{1}{2} (-1)^d \Delta_{H} \Psi
\end{equation}
is essentially the $H$-twisted de Rham Laplacian $\Delta_{H} = \{\dd_H, \dd_H^\dagger \}$ (in the untwisted picture where $\slashed{D} = \dd_H$).

\section{Applications}
\label{sec:applications}

In this section, we will look at a few applications of the \complex to supergravity and string theory. To match conventions in supergravity, it is often easier to work in a different frame to the one used in the previous section. We will follow \cite{CSW1} and define a frame $\{ \hE_A \}$ for the generalised tangent bundle a frame $\{ E^{A} \}$ for the dual frame of the dual bundle induced by the generalised metric. 
That is, the natural inner product is $\langle E^A , \hE_{B} \rangle = \delta^A{}_B$. 
Note this differs to the conventions in the previous section in a relative minus sign when raising/lowering indices on sections of $C_{-}$.

\subsection{Moduli spaces of flux backgrounds}
\label{sec:moduli}

It is straightforward to see that the \complex $\cx$ we have constructed calculates the infinitesimal moduli of the underlying torsion-free generalised $G$-structure as its second cohomology group, by employing the same argument as in section~\ref{sec:grav}.
Indeed we can interpret the first few terms as (similar statements have appeared in~\cite{goto2005deformations,Ashmore:2015joa}) 
\begin{equation}
\label{eq:Odd-moduli-complex}
\begin{aligned}
	0 \ra \Gamma(\Lambda^{0} E) \stackrel{}{\ra} \Gamma(\Lambda^1 E) \stackrel{}{\ra}  \Gamma(\mathfrak{o}(d,d+n)/\mathfrak{g}) \stackrel{}{\ra} \Gamma(T^{(\text{int})}) \ra \dots
\end{aligned}
\end{equation}
Here, there is a reducible gauge symmetry generated by sections of $\Lambda^{0} E$ (i.e.\ scalars), which in cases relevant to supergravity corresponds to the gauge-of-gauge transformations for the $B$-field, so that the cohomology relevant to the moduli of the structure becomes the second cohomology here. 
Again the last term $T^{(\text{int})}$ written here corresponds to the bundle of which the intrinsic torsion of the generalised $G$-structure is a section. This can be understood as follows.

Consider a tensor $\Sigma \in E^* \otimes \mf{g}$ which can be thought of as the difference of two compatible generalised connections $D\colon E \ra E^* \otimes E$. Recall that the torsion of a generalised connection can be defined using the Dorfman derivative to give a tensor $T(D) \in \Gamma(\Lambda^3 E)$. 
We define a map $\tau\colon E^* \otimes \Lambda^2 E^* \ra \Lambda^3 E$ to give the difference of the torsions of the two connections. With respect to such the frame $\{\hE_{A}\}$, the map $\tau$ takes the form
\begin{equation}\label{eq:torsion}
	\tau(\Sigma)_{ABC} = 3 \Sigma_{[ABC]}
\end{equation}
This map restricts to a map $\tau|$ on $E^* \otimes \mf{g}$. We then have an exact sequence of bundles 
\begin{equation}
	0 \ra \ker (\tau|) \stackrel{\iota}{\ra} E^* \otimes \mf{g}
		\stackrel{\tau|}{\ra} \Lambda^3 E
		\stackrel{\pi}{\ra} T^{(\text{int})} \ra 0
\end{equation}
where we have defined the bundle of which the intrinsic torsion of a $G$-structure is a section by
\begin{equation}
	T^{(\text{int})} = \on{coker} (\tau|) = (\Lambda^3 E) / \im (\tau|)
\end{equation}
Given a $G$-compatible connection, the projection of the torsion of this connection onto $T^{(\text{int})}$ does not change if one shifts to a different $G$-compatible connection by adding to it a tensor $\Sigma \in E^* \otimes \mf{g}$. Therefore, it is independent of the choice of $G$-compatible connection and represents a property of the $G$-structure itself. It is called the intrinsic torsion of the structure, and can be thought of as a part of the torsion which is common to all connections compatible with the $G$-structure.

As in section~\ref{sec:grav}, we can interpret the maps $\check{D}$ by considering a $G$-frame $\hE_A$ for $E$ and a $G$-compatible connection $D$ with
\begin{equation}
	D_V \hE_A = \Omega_V{}^B{}_A \hE_B
\end{equation}
where again we use the notation $\Omega_V{}^B{}_A = V^C \Omega_C{}^B{}_A$ and $\Omega_V$ is a section of $\mf{g}\subset E^*\otimes E$.

As in section~\ref{sec:grav}, for simplicity we assume reducibility so that we can decompose\footnote{Again, this is not necessary: one can instead work via projections onto the quotient $\mf{so}(d,d+n) / \mf{g}$.}
\begin{equation}
	\mf{so}(d,d+n) = \mf{g} \oplus \mf{k}
\end{equation}
and write a frame for a nearby $G$-structure as
\begin{equation}
	\hE'_A = \hE_A + X^B{}_A \hE_B
\end{equation}
where the generalised tensor $X$ lies in $\mf{k}$.

We can then compute the intrinsic torsion induced by the deformation $X$ of the structure, working to first order in $X$. We consider a deformed compatible connection $D' = D + \Sigma$, for $\Sigma \in \Gamma(E^* \otimes \mf{g})$, so that
\begin{equation}
\begin{aligned}
	\Omega'_V{}^B{}_A \hE'_B = D'_V \hE'_A = D_V \hE'_A + \Sigma_V{}^B{}_A \hE'_B
\end{aligned}
\end{equation}
and we have
\begin{equation}
\begin{aligned}
	D_V \hE'_A &= D_V (\hE_A + X^B{}_A \hE_B) \\
	&= \Omega_V{}^B{}_A \hE_B 
		+ (\der_V X^B{}_A) \hE_A + X^B{}_A \Omega_V{}^C{}_B \hE_C \\
	&= \Omega_V{}^B{}_A (\hE'_B - X^C{}_B \hE'_C) 
		+ (\der_V X^B{}_A + \Omega_V{}^B{}_C X^C{}_A) \hE'_B + O(X^2) \\
	&= (\Omega_V{}^B{}_A  + D_V X^B{}_A) \hE'_B + O(X^2)
\end{aligned}
\end{equation}
The shift in the components of the connection (each with respect to their corresponding frames) is thus given by 
\begin{equation}
\begin{aligned}
	\Omega'_V{}^A{}_B - \Omega_V{}^A{}_B = \Sigma_V{}^A{}_B + D_V X^A{}_B
\end{aligned}
\end{equation}
and the quantities on both sides must lie in $\mf{g}$. Let us again split the tensor $\Sigma$ via 
\begin{equation}
	\Sigma = \Sigma^{(\mf{g})} + \Sigma^{(\mf{k})}
\end{equation}
and the previous equation tells us that
\begin{equation}
	\Sigma^{(\mf{k})}_V = - D_V X
\end{equation}
Using the map $\tau$ as above and defining $\tau^{(\text{int})} = \pi \circ \tau$, we have
\begin{equation}
\begin{aligned}
	\tau^{(\text{int})}(\Sigma) &= \tau^{(\text{int})}(\Sigma^{(\mf{g})} + \Sigma^{(\mf{k})})
		= \tau^{(\text{int})}(\Sigma^{(\mf{k})})
		= \pi (\tau(\Sigma^{(\mf{k})})) \\
		&= \mc{P}_{\mc{A}^2} \Big( -\tfrac1{3!} (3 D_{[A} X_{BC]}) \; 
			E^A \wedge E^B \wedge E^C \Big) \\
		&= -\mc{P}_{\mc{A}^2} \Big( \hat{D} X \Big)\\
		&= - \check{D} X
\end{aligned}
\end{equation}
so that for $X \in \mc{A}^1$ we have that $\check{D} X$ is proportional to the intrinsic torsion of the new $G$-structure to first order in $X$.

Next we examine which tensors $X$ are induced by the action of an infinitesimal generalised diffeomorphism generated by a generalised vector $V \in \Gamma(E)$. The action on generalised tensors is via the Dorfman derivative so we have
\begin{equation}
\label{eq:Gdiff-variation}
	\delta \hE_A = L_V \hE_A = L^{D}_V \hE_A
	 	= (D_V - D \times_{\mf{so}(d,d)} V) \cdot \hE_A \\
		= \Big( \Omega_V{}^B{}_A - 2 \eta^{BC}D_{[A} V_{C]} \Big) \hE_B
\end{equation}
The parts of the last expression in parentheses which lies in $\mf{g}$ merely rotate $\hE_A$ to a new $G$-frame for the original $G$-structure and thus gives no deformation of the structure itself. We are thus interested only in the part which lies in $\mf{k}$. As $\Omega_V$ lies in $\mf{g}$, this is given by
\begin{equation}
	X = - \mc{P}_{\mf{k}} \Big(\tfrac{1}{2} (2 D_{[A} V_{B]}) E^A \wedge E^B \Big)
		= - \mc{P}_{\mc{A}^1} (\hat{D} V)
		= - \check{D} V
\end{equation}
We thus see that the infinitesimal moduli space of torsion free generalised $G$-structures is given by
\begin{equation}
\label{eq:gen-G-str-moduli}
\begin{aligned}
	\mc{M}_{G\text{-str}} 
		= \frac{\{ X : \check{D} X = 0\}}{\{ X = \check{D} V \}}
		= H^1(\mc{A})
\end{aligned}
\end{equation}
exactly as for~\eqref{eq:instanton-moduli} and~\eqref{eq:G-str-moduli}. 

Note that for the heterotic cases \begin{equation}G \times O(d+n) \subset O(d)\times O(d+n)\subset O(d,d+n),\end{equation} the \complex takes the form~\eqref{eq:het-double-complex}, in which the lower row is the tensor product of a gauge theory \complex for $C_+$ tensored with $C_-$ 
and the differential $\dd_+$ is defined using a generalised connection. Denote this complex by $\mc{A}_+(C_-)$. Again, the bundle $C_-$ (or its complexification) is the one denoted $Q$ in~\cite{delaOssa:2014cia,delaOssa:2017pqy}, and $\mc{A}_+(C_-)$ is the complex used to compute the moduli in those references. 
The upper row is the gauge theory \complexNoSpace, denoted here $\mc{A}_+$, again with the differential $\dd_+$ defined using a generalised connection . Denoting by $\mc{A}_T$ the total complex, one has a short exact sequence of complexes:
\begin{equation}
	0 \ra \mc{A}_+(C_-)[1] \ra \mc{A}_T \ra \mc{A}_+ \ra 0
\end{equation}
where the maps are inclusion and projection and $[1]$ denotes a degree shift by one. 
If one has $H^{k-1}(\mc{A}_+) = H^{k}(\mc{A}_+) = 0$ for some $k$, then the associated long exact sequence in cohomology gives us
\begin{equation}
	H^k (\mc{A}_T) \simeq H^k(\mc{A}_+(C_-))
\end{equation}
For $k=2$, this would give us the infinitesimal moduli space of the $G \times O(d+n) \subset O(d,d+n)$ structure in terms of the cohomology of $\mc{A}_+(C_-)$, which was the moduli space appearing  in~\cite{delaOssa:2014cia,delaOssa:2017pqy}. 

Note that naively this is not quite the full moduli space of the supersymmetric background, which would be the moduli space of a $G \times O(d+n) \subset O(d,d+n) \times \bbR^+$ structure, with the additional $\bbR^+$ factor corresponding to the dilaton field. Thus in general, $H^2 (\mc{A}_T)$ will give the moduli space of the supergravity background at fixed dilaton. 

However, in the case of $\SU(3)$ structures, as considered in~\cite{delaOssa:2014cia}, it turns out very non-trivially that it matches the physical moduli space. This follows from similar reasoning to that used in the discussion of gauge theory instantons on torsion-free $\SU(N)$ structures in section~\ref{sec:Carrion}. 
In particular, we take the analogue of the complex~\eqref{eq:Dolbeault-complex-anti} with the complexified scalar, so that overall we have
\begin{equation}
\label{eq:het-su3-double-complex-moduli}
    \begin{tikzcd}[cramped, column sep = tiny, row sep = tiny, cells={text width={width("$\Omega^{(0,0)}_+(C_-)$")}, align=center}]
    & & & \arrow[dr, "\dd_{-}", xshift=2ex, yshift=-2ex]  \\
    & & \arrow{dl}[swap]{\dd_{+}} 
    & \Omega^{(0,0)}_+ \arrow[dl] \arrow[dr] 
    & \left.\right. 
    & & \\
    &\left.\right. 
    & \Omega^{(0,1)}_+ \arrow[dl] \arrow[dr] 
    & &  \Omega^{(0,0)}_+(C_-) \arrow[dl] & \left.\right. 
    & \\
    & \Omega^{(0,2)}_+ \arrow[dl] \arrow[dr] 
    & & \Omega^{(0,1)}_+ (C_-) \arrow[dl] 
    \\
    \Omega^{(0,3)}_+ \arrow[dr] & &
   \Omega^{(0,2)}_+(C_-) \arrow[dl] 
    \\
    &\Omega^{(0,3)}_+(C_-)
    \\
    \end{tikzcd}
\end{equation}
Taking the tensor product with the (complex) forms in $\Omega_+^{0,\bullet}$ effectively complexifies the $C_-$ bundle. 
Consequently, taking the cohomology in $\dd_+$ effectively is a quotient by the complexification of the infinitesimal gauge transformations of the supergravity theory. 
Further, as we actually construct the \complex using the Lie algebra $\mf{u}(3) \otimes \mf{so}(6+n)$ we are really computing the moduli of the structure as a $U(3) \times \SO(6+n) \subset \SO(6,6+n)$ structure, which corresponds to the $J$-structure of~\cite{Ashmore:2019rkx}. 
To get the physical moduli space, one needs to include the $\psi$-structure, which includes an additional ${C}^\infty(\bbC^*)$ of degrees of freedom compared with the $J$-structure, as we will discuss below in section~\ref{sec:6d-superpotential}.  
The physical moduli space (plus an additional $\bbC^*$ factor corresponding to constant rescaling of $\psi$) is the K\"ahler quotient of the moduli space of $\psi$-structures with integrable $J$-structure by the gauge transformations. But one expects that this can be computed as a regular quotient by the complexified gauge transformations, at least if one restricts to so-called polystable points. (For a more detailed discussion of this construction, see~\cite{Ashmore:2019rkx}.) 
As a $\bbC^*$ family of integrable $\psi$-structures give the same $J$-structure, the moduli space of $J$-structures should match the physical moduli space. 
Thus, the second cohomology of~\eqref{eq:het-su3-double-complex-moduli} in fact should match the physical moduli space. 

%%%%%%%%%

\subsection{Target spaces for topological strings}

As mentioned in section \ref{sec:GxG_in_type_II}, the first appearance of the type II \complex is in the study of topological strings of~\cite{Ashmore:2021pdm}. It was noted that the double complex precisely captures the target space realisation of the worldsheet BRST complex of the topologically twisted sigma model. Moreover, the 1-loop partition function of the theory calculates a quantity associated to the complex called the analytic torsion.

To see this, note that in an $\mathcal{N}=(1,1)$ 2-dimensional sigma model, the left- and right-moving fermion fields $\psi_{\pm}$ can be seen as sections of $S(\Sigma)\otimes \phi^{*}(C_{\pm})$ respectively, where $S(\Sigma)$ is the worldsheet spinor bundle and $\phi\colon\Sigma \rightarrow M$ is the embedding function. When the target space has a refined $G$-structure, the worldsheet theory has enhanced symmetry. The most famous example is when the target space is K\"ahler and the worldsheet theory has enhanced $\mathcal{N}=(2,2)$ supersymmetry. For the $G_{2}$ and $Spin(7)$ strings, the enhanced symmetries are particular $\mathcal{W}$-algebras. Schematically, these algebras are generated by the worldsheet operators $T, G$, which generate the left/right-moving $\mathcal{N}=1$ Virasoro algebra, as well as operators $X,M$ which take the approximate form\footnote{More details can be found in \cite{Ashmore:2021pdm}.}
\begin{equation}
    X= \Phi_{a_{1}...a_{n}}\psi^{a_{1}}...\psi^{a_{n}}+... \ , \qquad M = \Phi_{a_{1}...a_{n}}\psi^{a_{1}}...\del\phi^{a_{n}}+... \ .
\end{equation}
Here $\Phi_{a_{1}...a_{n}}$ are the components of differential forms in singlet representations with respect to the $G$-structure.\footnote{For $G_{2}$ structures and $SU(3)$ structures, we get one set of tensors for \emph{each} singlet differential form.}

In each of the cases above, one can define a twist of the theory by an operator related to the spectral flow operator of the SCFT. In particular, one looks for a bosonic operator $\rho$ through which one can define a twisted energy momentum tensor
\begin{equation}
    T_{\text{twist}}\sim T + \del^{2}\rho
\end{equation}
which has vanishing central charge. In the $SU(3)$ case, the operator $\ee^{\ii\rho}$ defines spectral flow of the worldsheet SCFT. In the $G_{2}$ and $Spin(7)$ cases, one can bosonise the theory and find that
\begin{equation}
    X\sim (\del\rho)^{2} + \del^{2}\rho
\end{equation}
In any case, the new twisted energy momentum tensor provides new Lorentz charges to all the fields. One finds that the dimensionless operators can be identified with sections of $\bx^{p,q}$, where $\bx^{p,q}$ are the vector bundles in the double complex defined in section \ref{sec:GxG_in_type_II}, i.e.\ the vector bundles appearing in the \complexNoSpace. Furthermore, one can find a nilpotent operator $Q$ acting on the Hilbert space which decomposes into left- and right-moving pieces and acts as
\begin{equation}
    Q=Q_{L} + Q_{R}
\end{equation}
\begin{equation}\label{eq:pullback}
    Q_{L}\colon\Gamma(\bx^{p,q})\to \Gamma(\bx^{p+1,q}) \ , \quad Q_{R}\colon\Gamma(\bx^{p,q})\to \Gamma(\bx^{p,q+1})
\end{equation}
The left- and right-moving pieces are the operators $\dd_{\pm}$ from the double complexes of section~\ref{sec:GxG_in_type_II} \begin{equation}
    Q_{L} =  \dd_{+} \ , \quad Q_{R} =  \dd_{-}\ .
\end{equation}

From this, the authors in \cite{Ashmore:2021pdm} were able to identify the 1-loop partition function by using the formula for free energy at 1-loop
\begin{equation} \label{eq:topl-string-one-loop}
    \mathcal{F}_{1} = \delta(H_{L}-H_{R}) \frac{1}{2}\log\left[ \prod_{F_{L},F_{R}} (\det{}'(H_{L}+H_{R}))^{(-1)^F F_{L}F_{R} } \right]\ .
\end{equation}
Here $F_{L/R}$ is the left/right fermion number and $H_{L}=\{Q_{L},Q_{L}^{\dagger}\}$ is the left-moving Hamiltonian, and similarly for $H_{R}$. By the identification of $Q_{L/R}$ with $\dd_{\pm}$, and using the fact that $\Delta_{+} = \Delta_{-}$ for these cases,\footnote{The Laplacians $\Delta_\pm$ correspond to the worldsheet Hamiltonians for the left- and right-moving modes. The condition~\eqref{eq:Laplacians-equal} then becomes the level matching condition on the worldsheet. It is then natural that this is related to the vanishing of the square of the Dirac operator on the generalised spinor, as in the coordinate basis one has:
\begin{equation}
	\dd^2 = (\Gamma^A \der_A)^2 = \der^A \der_A + \Gamma^{AB} \der_{[A} \der_{B]} 
	= \der^A \der_A
\end{equation}
and the ``weak constraint" $\der^A \der_A = 0$ is well-known to be related to the level matching condition in the literature on double field theory~\cite{Hull:2009mi}. The other natural Laplacian operator~\eqref{eq:mass-shell} has leading term $\Delta \Phi \sim 2 G^{AB} D_A D_B + \dots$, and being harmonic with respect to this operator is analogous to the mass shell condition.} we identify the term in the $[...]$ as the analytic torsion of the \complexNoSpace.

In the case of topological strings on K\"ahler manifolds, and topological strings on $\G_{2}$ manifolds, where the \complex admits a pairing, we can also provide a Chern--Simons-like description of the target space theory. In these cases, the pairing is given by 
\begin{equation}
	\langle f, f' \rangle = \int_M \Phi f \wedge f' \ , \qquad f\in \cx^{\bullet}\ .
\end{equation}
where $\Phi = \sqrt{g} \ee^{-2\phi}$ is the generalised density for the dilaton field $\phi$ and we assume that we work with a generalised Levi-Civita connection which preserves this density. 

One can then construct a diagonal complex $\mc{D}$ from the double complex via $\mc{D}^k = \cx^{k,k}$, corresponding to the level-matched worldsheet operators, with second order differential $\dd_{\mc{D}} = \dd_+ \dd_-$. The ghost number zero fields in the theory then form a general element $f \in \mc{D}^\bullet$, with Chern--Simons type action
\begin{equation}
\label{eq:topl-string-CS}
	S = \langle f_0, \dd_{\mc{D}} f \rangle = \int_M \Phi f \wedge \dd_+ \dd_- f \ , \qquad f\in \mc{D}^{\bullet}\ .
\end{equation}
However, unlike in our examples above, this is only the ghost number zero fields and to write a BV action, one must introduce ghosts and anti-fields, which can also be described in terms of elements of the \complex $\cx$, and the BRST operator which can be described in terms of $\dd_+$ and $\dd_-$. The gauge symmetry for the field strength $F = \dd_{\mc{D}} f = \dd_+ \dd_- f$ is
\begin{equation}
	\delta f^{k,k} = \dd_+ \lambda^{k-1,k} + \dd_- \kappa^{k,k-1}
\end{equation}
so that the cohomology of the BRST complex at ghost number zero is given by the Aeppli cohomologies of the double complex $\cx$ for the degree $(k,k)$ elements. One could thus think of this theory as a Chern--Simons-type theory for the Aeppli cohomology. One can represent the full set of fields and anti-fields diagrammatically using a series of copies of the double complex $\cx$, one for each element $f$ appearing in~\eqref{eq:topl-string-CS}, and selecting only a particular subset of the elements in each, similarly to the complexes appearing in~\cite{Tseng:2009gr,tseng2014generalized}. Further details of this construction will be presented elsewhere. 
Note, however, that the $\der\bar\der$-lemma means that the Aeppli cohomologies are isomorphic to the $\dd_\pm$-cohomologies for these complexes, which are in turn isomorphic to the spaces of harmonic elements. Thus, despite the apparently different form of the action, the classical states for these theories are still related to on-shell deformations of the classical backgrounds.

The BV quantisations of the actions~\eqref{eq:topl-string-CS} were previously shown~\cite{Pestun:2005rp,Ashmore:2021pdm} to reproduce the analytic torsion appearing in the one-loop partition functions~\eqref{eq:topl-string-one-loop}.

%%%%%%%%%%%%%%%%%%%%%%%%%%%%%%%%%%%%%%%%%%%%%%%%%

\subsection{\texorpdfstring{$\SU(3)$}{SU(3)} heterotic superpotential}
\label{sec:6d-superpotential}

In this section, we will see that in six-dimensions the heterotic $U(3)\times\SO(6+n)$ complex (on a supersymmetric background with a torsion-free $\SU(3)\times\SO(6+n) \subset \SO(6,6+n)\times\bbR^+$ structure) is closely related to the theory of~\cite{Ashmore:2018ybe,Ashmore:2023vji} based on the heterotic superpotential~\cite{Gurrieri:2004dt,delaOssa:2015maa,McOrist:2016cfl}. In fact, one has to perform a mild extension of this \complex to give it a BV symplectic pairing. One can easily see that an enlargement of some kind must be necessary as the superpotential depends on the dilaton, which does not appear in the degrees of freedom of the $U(3)\times\SO(6+n) \subset \SO(6,6+n)$ structure (one must enlarge the generalised structure group to $\SO(6,6+n)\times\bbR^+$ to see this degree of freedom~\cite{CSW1}). Correspondingly, the equation of motion of the superpotential theory is more constraining than simply requiring that the $U(3)\times\SO(6+n)$ structure is torsion-free~\cite{Ashmore:2019rkx}.  
There is an additional part of the generalised intrinsic torsion of the enclosing $\SU(3)\times\SO(6+n) \subset \SO(6,6+n)\times\bbR^+$ structure transforming in the $\repp{\bar{3}}{1}$ representation that is constrained to vanish by the vanishing of the superpotential and its variation.
However, these points are cured by using a simple extended version of the $U(3)\times\SO(6+n)$ complex described below, and we will see that the natural Chern--Simons-type theory for this complex precisely reproduces the heterotic superpotential at quadratic order. 

The superpotential is a functional of a generalised $\SU(3)\times\SO(6+n) \subset \SO(6,6+n)\times \bbR^+$ structure. 
It was shown in~\cite{Ashmore:2019rkx} how a generalised $U(3)\times\SO(6+n) \subset \SO(6,6+n)$ structure (a ``$J$ structure") is defined by a generalised tensor $J \in \Gamma(\Lambda^2 E)$. As $U(3)\times\SO(6+n) \subset \SO(6)\times\SO(6+n)$, this defines a generalised metric on $E$, and $J$ annihilates $C_-$ and squares to minus one acting on $C_+$. 
The (complexified) $C_+$ bundle is thus split into the $+\ii$-eigenbundle $C_+^{1,0}$ and the $-\ii$-eigenbundle $C_+^{0,1}$. 
If the line-bundle $\Lambda^3 C_+^{0,1}$ is trivial, then a further reduction of the structure group to $\SU(3)\times\SO(6+n)$ is possible, and this has an associated invariant tensor $\chi \in \Gamma(\Lambda^3 C_+^{0,1}) \subset \Gamma(\Lambda^3 E)$. 
To define a $\SU(3)\times\SO(6+n) \subset \SO(6,6+n) \times\bbR^+$ structure, one requires a generalised tensor $\psi \in \Gamma(\Lambda^3 E \otimes \det T^*M)$, which is the product of a tensor of the same type as $\chi$ above with a density $\Phi = \sqrt{g} \ee^{-2\phi}$ that trivialises the $\bbR^+$ factor of the enlarged generalised structure group $\SO(6,6+n) \times\bbR^+$ which includes the dilaton. Such a tensor defines the $\SU(3)\times\SO(6+n)$ subgroup entirely and thus also defines a generalised metric, a dilaton and an associated $J$ structure.

For a $\psi$-structure, with associated $J$-structure, the value of the heterotic superpotential is 
\begin{equation}
\label{eq:gen-superpotential}
    W \sim \int \tr (J \cdot D \proj{\adj} \psi) \sim \int \psi^{ABC} D_A J_{BC}
\end{equation}
In~\cite{Ashmore:2019rkx}, it was stated that $D$ in this equation is a generalised Levi-Civita connection, but in fact the equation holds for any torsion-free generalised connection. 

Let us write $\psi^{ABC} = \Phi \chi^{ABC}$ where $\Phi = \sqrt{g} \ee^{-2\phi}$ is the generalised volume density. 
We define an $\SU(3)\times\SO(6+n)$ frame $\{ \hE_A \} = \{ \hE^+_a, \hE^+_{\ba}, \hE^-_{\tm} \}$ (nb.\ $a,b = 1,2,3$ are holomorphic indices for $C_+$ and $\tm, \tn = 1,\dots, 6+n$ are real indices for $C_-$) to be one for which we have:
\begin{equation}
\label{eq:SU3xSO6-frame}
	J^a{}_b = \ii \delta^a{}_b
	\hs{30pt}
	J^{\ba}{}_{\bb} = -\ii \delta^{\ba}{}_{\bb}
	\hs{30pt}
	\chi^{\ba\bb\bc} = \epsilon^{\ba\bb\bc}
	\hs{30pt}
	\Phi = 1	
\end{equation}
where the last condition also fixes the $\bbR^+$ frame. 

If we then vary the connection in~\eqref{eq:gen-superpotential} by an arbitrary tensor $\Sigma_A{}^B{}_C$, we find
\begin{equation}
\begin{aligned}
	\delta( \psi^{ABC} D_A J_{BC}) &= \epsilon^{\ba\bb\bc} (\Sigma_{\ba} \cdot J)_{\bb\bc} \\
		&= \epsilon^{\ba\bb\bc} (-2 \Sigma_{\ba}{}^E{}_{\bb} J_{E\bc}) \\
		&= \epsilon^{\ba\bb\bc} (-2 \Sigma_{\ba}{}^e{}_{\bb} J_{e\bc}) \\
		&= 2 \ii \epsilon^{\ba\bb\bc} ( \Sigma_{[\ba\bc\bb]}) \\
\end{aligned}
\end{equation}
which vanishes if $\Sigma$ is torsion-free (so that $\Sigma_{[ABC]} = 0$). This demonstrates the claim above that the connection appearing in the definition of the superpotential~\eqref{eq:gen-superpotential} can be any torsion-free generalised connection. 

We can then parameterise an infinitesimal variation of the structure via a variation of the frame and density:
\begin{equation}
\label{eq:gen-deform}
\begin{aligned}
	\delta \hE^+_a &= \ii \alpha \hE^+_a 
		+ \bar\beta_{a}{}^{\bb} \hE^+_{\bb} + \Lambda_{a}{}^{\tm}  \hE^-_{\tm} \\
	\delta \hE^+_{\ba} &= -\ii \alpha \hE^+_{\ba} 
		+ \beta_{\ba}{}^{b} \hE^+_{b} + \Lambda_{\ba}{}^{\tm}  \hE^-_{\tm}\\
	\delta \hE^-_{\tm} &= \Lambda^{a}{}_{\tm} \hE^+_a  + \Lambda^{\ba}{}_{\tm} \hE^+_{\ba} \\
	\delta \Phi &= \Lambda \Phi
\end{aligned}
\end{equation}
The parameters in this equation parameterise a (real) element of $\mf{so}(6,6+n) \times \bbR^+ / \mf{su}(3)\times \mf{so}(6+n)$ at each point, 
so that $\hE^+_{\ba} = (\hE^+_a)^*$, $\Lambda_{\ba}{}^\tm = (\Lambda_a{}^{\tm})^*$ and $ \bar\beta_{a}{}^{\bb} = (\beta_{\ba}{}^{b})^*$ . Note here we are using the convention of raising/lowering decomposed indices with the ordinary metric $g_{a\bb}$ as in~\cite{CSW1}. The ``canonical" index positions in which they are matched with the $O(d,d+n)$ indices are
\begin{equation}
	J^A{}_B \ , \hs{30pt} \psi^{ABC} \ , \hs{30pt} \Lambda^{A}{}_B \ , \hs{30pt} \hE_A \ ,
 \hs{30pt} E^A \ .
\end{equation}

To calculate the components of $\delta \psi^{ABC}$ and $\delta J^A{}_B$ we write
\begin{equation}
\begin{aligned}
	\psi' &= \Phi' \chi'
		= \tfrac{1}{3!} \Phi' \epsilon^{\ba\bb\bc} 
			\hE'^+_{\ba} \wedge \hE'^+_{\bb} \wedge \hE'^+_{\bc} \\
	J' &= \ii \delta^a{}_b \hE'^+_a\otimes E'^{+b} 
		- \ii \delta^{\ba}{}_{\bb} \hE'^+_{\ba} \otimes E'^{+\bb}
\end{aligned}
\end{equation}
and then expand out the primed objects to first order in the deformation parameters $\alpha, \beta_{\ba\bb}, \Lambda_{\ba}{}^\tm$ and $\Lambda$. This results in
\begin{equation}
\begin{aligned}
	\delta\psi^{\ba\bb\bc} &= (\Lambda - 3\ii\alpha)\epsilon^{\ba\bb\bc}
	\hs{30pt}
	\delta\psi^{\ba\bb c} = \epsilon^{\ba\bb\be} \beta_{\be}{}^{c}
	\hs{30pt}
	\delta\psi^{\ba\bb\tm} = \epsilon^{\ba\bb\be} \Lambda_{\be}{}^{\tm}\\
	\delta J^{\tm}{}_{a} &= \ii \Lambda_a{}^{\tm}
	\hs{30pt}
	\delta J^{a}{}_{\tm} = -\ii \Lambda^{a}{}_{\tm}\\
	\delta J^{\tm}{}_{\ba} &= -\ii\Lambda_{\ba}{}^{\tm}
	\hs{30pt}
	\delta J^{\ba}{}_{\tm} = \ii \Lambda^{\ba}{}_{\tm} \\
	\delta J^{b}{}_{\ba} &= -2\ii \beta_{\ba}{}^b
	\hs{30pt}
	\delta J^{\bb}{}_{a} = 2\ii \beta_a{}^{\bb} \\
\end{aligned}
\end{equation}

Using that one can take $D$ in~\eqref{eq:gen-superpotential} to be any fixed background torsion-free connection (i.e.\ it does not depend on the structure), it is then very straightforward to calculate the perturbative expansion of the superpotential. One has:
\begin{equation}
    \delta W \sim \int \Big[ \delta \psi^{ABC} D_A J_{BC} + \psi^{ABC} D_A \delta J_{BC} \Big]
\end{equation}
and 
\begin{equation}
    \delta^2 W \sim \int \Big[ \delta^2 \psi^{ABC} D_A J_{BC} + 2\delta \psi^{ABC} D_A \delta J_{BC} 
    	+ \psi^{ABC} D_A \delta^2 J_{BC} \Big]
\end{equation}
However, if we are expanding around a supersymmetric background, we can make the convenient choice that $D$ is a torsion-free compatible connection for that background solution. This means that in the background $D J = D\psi = 0$, so on integration-by-parts $\delta W = 0$ (as it should be around a supersymmetric solution) and 
\begin{equation}
\label{eq:W-quadratic}
\begin{aligned}
    \delta^2 W &\sim \int  2\delta \psi^{ABC} D_A \delta J_{BC} \\
    &= \int 2\ii \epsilon^{\ba\bb\bc} \Big[ 
    	-2 \Lambda_{\ba}{}^{\tm} D_{\bb} \Lambda_{\bc\tm}
	+2 \Lambda_{\ba}{}^{\tm} D_{\tm} \beta_{\bb\bc}
	+ 6 \varphi D_{\ba} \beta_{\bb\bc}
	- 2\beta_{\ba}{}^{e} D_{e} \beta_{\bb\bc} \Big]
\end{aligned}
\end{equation}
where we have defined the complex scalar $\varphi = \Lambda - 3\ii\alpha$. This is the quadratic action (for ghost number zero fields) which we seek to recover from our extended \complex below. 
Note that it depends only on the variables $(\Lambda_{\ba}{}^{\tm}, \beta_{\ba\bb}, \varphi)$ and not their complex conjugates, so that it is indeed holomorphic on the parameter space. 

%%%%%%%%%%%%%%%%%%%%%%%%%%%%%%%%%%%%%%%%%%%%%%%%%

We now examine the heterotic complex~\eqref{eq:het-double-complex} for the structure group $U(3) \times \SO(6+n)$ in six-dimensions, which takes the form:
\begin{equation}
\label{eq:het-su3-double-complex}
    \begin{tikzcd}[cramped, column sep = tiny, row sep = tiny, cells={text width={width("$\Omega^{(0,0)}_+(C_-)$")}, align=center}]
    & & & \arrow[dr, "\dd_{-}", xshift=2ex, yshift=-2ex]  \\
    & & \arrow{dl}[swap]{\dd_{+}} 
    & \Omega^{(0,0)}_+ \arrow[dl] \arrow[dr] 
    & \left.\right. 
    & & \\
    &\left.\right. 
    & \Omega^{(0,1)}_+ \arrow[dl] \arrow[dr] 
    & &  \Omega^{(0,0)}_+(C_-) \arrow[dl] & \left.\right. 
    & \\
    & \Omega^{(0,2)}_+ \arrow[dl] \arrow[dr] 
    & & \Omega^{(0,1)}_+ (C_-) \arrow[dl] 
    \\
    \Omega^{(0,3)}_+ \arrow[dr] & &
   \Omega^{(0,2)}_+(C_-) \arrow[dl] 
    \\
    &\Omega^{(0,3)}_+(C_-)
    \\
    \end{tikzcd}
\end{equation}
One might try to write a BV action associated to the total complex of~\eqref{eq:het-su3-double-complex}, i.e.\ a free-field theory whose equations of motion are the statement that the fields in $\Omega^{(0,2)}_+ \oplus \Omega^{(0,1)}_+(C_-)$ are closed in the total differential, with the gauge symmetry shifting them by exact pieces. Unfortunately,~\eqref{eq:het-su3-double-complex} does not have a cyclic structure (in the language of the $L_\infty$-algebra community) or graded symplectic pairing (in the BV language). The row with $C_-$ factors is self-dual, but there are no partners for the fields and anti-fields in the first row. This means that one cannot apply the standard procedure for constructing an action from the complex (see e.g.~\cite{Hohm:2017pnh}).

However, there is a simple procedure to construct a complex equipped with a cyclic structure from one that does not. Consider a general complex $(C, \dd_C)$. If one simply adds the dual complex shifted by one degree $(C^*[-1], \dd_C^*)$ to form
\begin{equation}
\label{eq:add-dual-complex}
	\dots \to C_{-2} \oplus C^*_{3} 
		\to C_{-1} \oplus C^*_{2} 
		\to C_{0} \oplus C^*_{1} 
		\to C_{1} \oplus C^*_{0} 
		\to C_{2} \oplus C^*_{-1} 
		\to C_{3} \oplus C^*_{-2} 
		\to \dots
\end{equation}
then this has a natural pairing of the spaces opposite each other with respect to the middle arrow. One could then write a BV action using this pairing.\footnote{Note that one could do this for any of the the complexes which lack a BV symplectic pairing. For example, for the Carri\'on complexes for structure groups $\SU(2)$ in four dimensions or $\Spin(7)$ or $\SU(4)$ in eight dimensions, this recovers part of the action for the theories of instantons in~\cite{Witten:1990,Baulieu:1997jx}. One could also do this for the heterotic versions of these complexes.}

In our case, the $\Omega^{(0,\bullet)}_+(C_-)$ row of~\eqref{eq:het-su3-double-complex} is self-dual and would have a pairing of the right type were it all we had. It is the other $ \Omega^{(0,\bullet)}_+$ row which lacks a pairing. We can thus add a copy of the dual of this row and use that $\Omega^{0,p}_+ \simeq (\Omega^{0,3-p}_+)^*$ via contraction with the anti-holomorphic top-form $\bar\Omega$ for $C_+$. Via this isomorphism we also have that $\dd_+^*$ becomes $\dd_+$ again. Overall, we simply add a degree-shifted copy of the top row to get
\begin{equation}\label{eq:het-su3-complex-BV}
    \begin{tikzcd}[column sep = tiny, row sep = tiny, cells={text width={width("$\Omega^{(0,0)}_+(C_-)$")}, align=center}]
    & & 
    & \Omega^{(0,0)}_+ \arrow[dl] \arrow[dr] 
    & 
    & & \\
    &\left.\right. 
    & \Omega^{(0,1)}_+\arrow[dl] \arrow[dr] \arrow[drrr, dotted]
    & &  \Omega^{(0,0)}_+(C_-) \arrow[dl] \arrow[dr]& \left.\right. 
    & \\
    & \Omega^{(0,2)}_+  \arrow[dl] \arrow[dr] \arrow[drrr, dotted]
    & & \Omega^{(0,1)}_+(C_-) \arrow[dl] \arrow[dr] 
    & & \Omega^{(0,0)}_+
    \arrow[dl]  
    \\
    \Omega^{(0,3)}_+ \arrow[dr] \arrow[drrr, dotted] & &
    \Omega^{(0,2)}_+(C_-) \arrow[dl] \arrow[dr] 
    & & \Omega^{(0,1)}_+ \arrow[dl] 
    \\
    & \Omega^{(0,3)}_+(C_-) \arrow[dr]
    & & \Omega^{(0,2)}_+ \arrow[dl]
    \\
    & & \Omega^{(0,3)}_+ 
    \end{tikzcd}
\end{equation}
The horizontal levels in this diagram then have ghost numbers $+2,+1,0,-1,-2,-3$ respectively and one has a symplectic pairing of degree $+1$. We also see that there are natural additional maps $\Omega^{(0,k)}_+(C_-) \ra \Omega^{(0,k)}_+$, given by the duals of our existing maps $\Omega^{(0,k)}_+ \ra \Omega^{(0,k)}_+(C_-)$ 
(i.e.\ contractions $D_{\tm} \omega_{\bc_1 \dots \bc_k}{}^{\tm}$ using the frame indices introduced above) and maps $\Omega^{(0,k+1)}_+ \ra \Omega^{(0,k)}_+$ proportional to $\dd_+^\dagger$, so we include these. The inclusion of these last maps means that naively~\eqref{eq:het-su3-complex-BV} will not have the structure of a double complex.

It is not obvious that the total differential $\dd_W$ on~\eqref{eq:het-su3-complex-BV} squares to zero. Acting on the middle and lower rows, one can deduce that this is so from the fact that the dual differential on total complex $\mc{A}_T$ from~\eqref{eq:het-su3-double-complex} squares to zero. However, acting on an element of the top row, there are several terms. The first is proportional to the square of the differential on $\mc{A}_T$ which we know vanishes from the results of section~\ref{sec:gen-Dolbeault}. However, in our extended complex there is also a term proportional to $\dd_-^\dagger \dd_- = \Delta_-$ (for the top row) and another proportional to $\{\dd_+, \dd_+^\dagger \} = \Delta_+$. It is not immediately clear that these terms cancel as we do not have a direct analogue of the results of section~\ref{sec:Odd-spinor} for the heterotic case.  
However, one can see that these Laplacians are in fact equal by using that the complex $\Omega^{0,\bullet}_+$ is spinor type and the supersymmetry algebra. 
We have that $\Omega^{0,\bullet}_+ \simeq \Gamma(S(C_+))$ and in terms of the spinor variables, the Dirac operator $\Dirac_+$ (i.e.\ the analogue of~\eqref{eq:Dirac-complex} for this single complex) corresponds to the Dirac operator appearing in the supersymmetry transformation of the dilatino field $\rho$, i.e.\ for $\omega \in \Omega^{0,\bullet}_+$ corresponding to $\theta \in \Gamma(S(C_+))$ we have
\begin{equation}
   \Dirac_+ \omega \quad \lra \quad \gamma^a D_a \theta = \delta_\theta \rho
\end{equation}
so that the Laplacian $\Delta_+ \propto \Dirac_+^2$ comes from
\begin{equation}
    \Dirac_+^2 \omega \quad \lra \quad (\gamma^a D_a)^2 \theta = \gamma^a D_a (\delta_\theta \rho)
\end{equation}

We also have that the map $\dd_-$ is given by the corresponding variation of the combined gravitino and gaugino fields $\psi_{\tm}$
\begin{equation}
    D_\tm \omega_{\ba_1 \dots \ba_k} \quad \lra \quad  D_{\tm} \theta = \delta_\theta \psi_{\tm}
\end{equation}
so that the corresponding Laplacian is proportional to
\begin{equation}
    D^{\tm} D_\tm \omega_{\ba_1 \dots \ba_k} \quad \lra \quad  
    	D^{\tm} D_{\tm} \theta = D^{\tm} (\delta_\theta \psi_{\tm})
\end{equation}
The proportionality of the Laplacians $\Delta_\pm$ then follows from the closure of the supersymmetry algebra on the $\rho$ equation of motion $\gamma^a D_a \rho  - D^{\tm} \psi_{\tm} = 0$ (or equivalently the construction of the generalised Ricci scalar as in~\cite{CSW1,Coimbra:2014qaa}):
\begin{equation}
\label{eq:het-Laplacians-equal}
\begin{aligned}
		\delta_\theta (\gamma^a D_a \rho  - D^{\tm} \psi_{\tm}) 
		= (\gamma^a D_a)^2 \theta - D^{\tm} D_{\tm} \theta
		= -\tfrac14 R \theta
		= 0
\end{aligned}
\end{equation}
where $R$ is the generalised Ricci scalar curvature which vanishes for a supersymmetric Minkowski vacuum.\footnote{Note that the closure of the supersymmetry algebra on the gravitino (and gaugino) equation of motion $\slashed{D} \psi_{\tm} - D_{\tm} \rho = 0$ also corresponds to properties of the complex: namely the double complex property $\{ \dd_+ , \dd_-\} = 0$ and the K\"ahler identity $\{ \dd_+^\dagger, \dd_-\} = 0$.}
This relation between the Laplacians leads to the result that the total differential $\dd_W$ on~\eqref{eq:het-su3-complex-BV} squares to zero. 

Having established that~\eqref{eq:het-su3-complex-BV} is a complex, which was constructed to have a BV symplectic pairing, one can then write the corresponding Chern--Simons-type BV action, as was done for gauge theories in section~\ref{sec:BV}. 
It is easy to see that taking an element $f_0$ of ghost number zero, with components $(\Lambda_{\ba}{}^{\tm}, \beta_{\ba\bb}, \varphi)$, 
one can apply the differential $\dd_W$ and use the natural integration against the $\psi$ structure (which has components $\epsilon^{\ba\bb\bc}$ in the $\SU(3)\times\SO(6+n)$ frames as in~\eqref{eq:SU3xSO6-frame}) to recover an expression of the type~\eqref{eq:W-quadratic} for the inner product
\begin{equation}
    S = \langle f_0, \check{\dd} f_0 \rangle \ .
\end{equation}
An element $f_{-1}$ of ghost number one has (complex) components $(V_{\ba}, V^{\tm})$ which are some of the components of a complex generalised vector $V$. Viewing $f_0$ as the parameters of a deformation of the generalised $\SU(3)\times\SO(6+n)$ structure, taking $f_0$ to be $\dd_W f_{-1}$ gives the action of an infinitesimal complexified generalised diffeomorphism on the structure. This is precisely the gauge symmetry of the superpotential theory and we see that~\eqref{eq:het-su3-complex-BV} is indeed the BRST complex of the superpotential theory, which has BV action
\begin{equation}
    S = \langle f, \check{\dd} f \rangle \ ,
\end{equation}
for $f$ a generic element.  

%%%%%%%%%%%%%%%%%%%%%%%%%%%%%%%%%%%%%%%%%%%%%%%%%

One could be troubled by the fact that~\eqref{eq:het-su3-complex-BV} looks different to the complex for the superpotential theory as described in~\cite{Ashmore:2023vji}. In particular, it does not have the structure of a double complex. This is due to the fact that we have expanded the superpotential functional in a different basis of fields. We will show in future work that a field redefinition, corresponding to a different parameterisation of the variation of the generalised structure, provides an equivalence to the complex of~\cite{Ashmore:2023vji}.

%%%%%%%%%%%%%%%%%%%%%%%%%%%%%%%%%%%%%%%%%%%%%%%%%%%%%%%%%%%%%%%%%%%%%%%%%%

\subsection{\texorpdfstring{$G_2$}{G2} heterotic superpotential}

We can perform the same analysis for the $G_{2}$ heterotic superpotential. Such a background corresponds to an $G_{2}\times SO(7+n)$ structure. Such a structure is defined by a generalised tensor $\psi\in \Gamma(\ext^{3}E\otimes \det T^{*}M)$. In a frame $\{\hat{E}_{A}\} = \{ \hat{E}^{+}_{m}, \hat{E}^{-}_{\tilde{m}} \}$ (with $m,n,...=1,2,...,7$ indices for $C_{+}$ and $\tilde{m},\tilde{n},... = 1,2,...,7+n$ indices for $C_{-}$), we can write
\begin{equation}
    \psi = \Phi\chi \ , \qquad \chi = \varphi^{mnp} \hat{E}_{m}^{+}\wedge \hat{E}^{+}_{n}\wedge \hat{E}^{+}_{p} \ , \qquad \Phi = \sqrt{g}\ee^{-2\dil} \ ,
\end{equation}
where $\varphi^{mnp}$ are the components of a stable 3-form in 7-dimensions (i.e.\ the components of a $\G_{2}$ 3-form). With this, the (real) superpotential takes the form
\begin{equation}
    W\sim \int \Phi  \chi^{ABE}\chi^{CD}{}_{E}\, D_{[A}\chi_{BCD]}\ ,
\end{equation}
where $D$ is any torsion free generalised connection (one can use the same proof as in the previous section to show that $W$ is independent of the choice of torsion free connection).

We would like to take the second variation of this around some on-shell background. This means we are free to choose $D$ to be torsion-free and compatible with the $G_{2}\times SO(7+n)$ structure, in which case the second variation takes the form
\begin{equation}\label{eq:G2_superpotential_2nd_var}
    \delta^{2}W \sim \int \left( \delta \Phi \chi^{ABE}\chi^{CD}{}_{E} + \Phi  \delta\chi^{ABE}\chi^{CD}{}_{E} + \Phi  \chi^{ABE}\delta\chi^{CD}{}_{E} \right) D_{[A} \delta \chi_{BCD]} \ .
\end{equation}
We can take the following form for the variations
\begin{equation}
    \begin{aligned}
        \delta \chi^{mnp} &= 3\beta^{[m}{}_{q} \varphi^{np]q} \ , \\
        \delta \chi^{\tilde{m}np} &= 3\Lambda^{\tilde{m}}{}_{q}\varphi^{npq}\ , \\
        \delta \Phi &= \Lambda \Phi\ ,
    \end{aligned}
\end{equation}
where $\beta_{mn} \in  \Omega^{+,2}_{\mathbf{7}}$, $\Lambda^{\tilde{m}}{}_{n} \in \Omega^{+,1}_{\mathbf{7}}(C_{-})$, and $\Lambda \in \Omega^{+,0}_{\mathbf{1}}$. With this parameterisation, the action \eqref{eq:G2_superpotential_2nd_var} takes the form
\begin{equation}
    \delta^{2}W \sim \int 6 \varphi^{mnp} \left( 2 \Lambda D_{m}\beta_{np} - 3 \beta_{mn} D^{r}\beta_{rp} + 2 \Lambda^{\tilde{r}}{}_{m}D_{\tilde{r}}\beta_{np} + 6 \Lambda_{\tilde{r}m} D_{n}\Lambda^{\tilde{r}}{}_{p} \right)
\end{equation}

How does this relate to the \complexes we have defined? The \complex for $G_{2}\times SO(7+n)$ structures takes the form
\begin{equation}\label{eq:het-g2-double-complex}
   \begin{tikzcd}[column sep = tiny, row sep = tiny, cells={text width={width("$\Omega^{+,0}_{\rep{1}}(C_-)$")}, align=center}]
   & & & \arrow[dr, "\dd_{-}", xshift=1.5ex, yshift=-1.5ex]  \\
   & & \arrow{dl}[swap]{\dd_{+}} 
   & \Omega^{+,0}_{\rep{1}} \arrow[dl] \arrow[dr] 
   & \left.\right. 
   & & \\
   &\left.\right. 
   & \Omega^{+,1}_{\rep{7}} \arrow[dl] \arrow[dr] 
   & &  \Omega^{+,0}_{\rep{1}}(C_-) \arrow[dl] & \left.\right. 
   & \\
   & \Omega^{+,2}_{\rep{7}}  \arrow[dl] \arrow[dr] 
   & & \Omega^{+,1}_{\rep{7}}(C_-) \arrow[dl]
   \\
   \Omega^{+,3}_{\rep{1}} \arrow[dr] & &
   \Omega^{+,2}_{\rep{7}}(C_-) \arrow[dl]
   \\
   & \Omega^{+,3}_{\rep{1}}(C_-)
   \\
   \end{tikzcd}
\end{equation}
As before, this does not have a cyclic structure but we can perform the same procedure as for the $SU(3)$ case to write down an extension which does. We find
\begin{equation}\label{eq:het-g2-complex-BV}
   \begin{tikzcd}[column sep = tiny, row sep = tiny, cells={text width={width("$\Omega^{+,0}_{\rep{1}}(C_-)$")}, align=center}]
   & & 
   & \Omega^{+,0}_{\rep{1}} \arrow[dl] \arrow[dr] 
   & 
   & & \\
   &\left.\right. 
   & \Omega^{+,1}_{\rep{7}} \arrow[dl] \arrow[dr] \arrow[drrr, dotted]
   & &  \Omega^{+,0}_{\rep{1}}(C_-) \arrow[dl] \arrow[dr]& \left.\right. 
   & \\
   & \Omega^{+,2}_{\rep{7}}  \arrow[dl] \arrow[dr] \arrow[drrr, dotted]
   & & \Omega^{+,1}_{\rep{7}}(C_-) \arrow[dl] \arrow[dr] 
   & & \Omega^{+,0}_{\rep{1}}
   \arrow[dl]
   \\
   \Omega^{+,3}_{\rep{1}} \arrow[dr] \arrow[drrr, dotted] & &
   \Omega^{+,2}_{\rep{7}}(C_-) \arrow[dl] \arrow[dr] 
   & & \Omega^{+,1}_{\rep{7}} \arrow[dl]
   \\
   & \Omega^{+,3}_{\rep{1}}(C_-) \arrow[dr]
   & & \Omega^{+,2}_{\rep{7}} \arrow[dl]
   \\
   & & \Omega^{+,3}_{\rep{1}}
   \end{tikzcd}
\end{equation}
The maps $\Omega^{+,k}(C_{-}) \to \Omega^{+,k}$ and $\Omega^{+,k+1} \to \Omega^{+,k}$ are precisely the duals of the maps appearing in \eqref{eq:het-g2-double-complex} as before. The horizontal levels are defined to have ghost number $+2,+1,0,-1,-2,-3$, and the symplectic pairing is of degree $+1$.

We need to determine whether the total complex of \eqref{eq:het-g2-complex-BV} really defines a complex, i.e.\ the total derivative squares to 0. The difficult part is the total derivative acting on the top row which will return a term proportional to ${\dd}^{2}\omega \propto (\Delta_{+} - \Delta_{-})\omega$. Once again, we can use the fact that this complex is a spinor complex and $\Omega^{+,\text{ev}} \simeq \Omega^{+,\text{odd}} \simeq \Gamma(S(C_{+}))$. The conditions coming from supersymmetry ensure that the left and right Laplacians are equal on the top row, as in \eqref{eq:het-Laplacians-equal} and hence the total space of \eqref{eq:het-g2-complex-BV} is a complex.

A generic element $f_{0}$ of degree 0 in this complex can be parameterised by $(\Lambda^{\tilde{m}}{}_{n}, \beta_{mn}, \Lambda)$. One then finds that the action \eqref{eq:G2_superpotential_2nd_var} takes the form
\begin{equation}
    S = \left<f_{0},\check{\dd} f_{0} \right>
\end{equation}
and that the associated BV action is given by a generic element $f$ of the complex
\begin{equation}
    S_{\text{BV}} = \left<f,\check{\dd} f\right>
\end{equation}
It remains an open and interesting question as to whether a field redefinition of the form used in \cite{Ashmore:2023vji} will allow us to rewrite \eqref{eq:het-g2-complex-BV} as a double complex. If it is possible then we can immediately write down the 1-loop partition function following the techniques in that paper.

\section{Discussion and outlook}

In this work, we have shown that Carri\'on's prescription to construct a complex associated to Donaldson-Thomas-type instantons of a gauge theory on a manifold with a torsion-free $G$-structure is actually part of a much more general picture. 
This includes gravitational instantons (i.e.\ torsion-free $G$-structures themselves) and supergravity instantons (i.e.\ supersymmetric backgrounds, or equivalently torsion-free generalised $G$-structures). We have shown how to construct these more general complexes, which we have labelled \complexesNoSpace, via information from only the (generalised) $G$-structure. 
Further, we have shown that for particular choices of the group $G$, these \complexes become equivalent to spinors and the corresponding differentials and their adjoints are packaged together into the Dirac operator acting on those spinors. 
This observation provided elegant general proofs of statements relating the the Laplacian on the \complex to the de Rham Laplacian, or its $H$-twisted version in the generalised geometry case. Thus far, these relations had been noted in specific cases and proved by direct calculations, whose shape appears to have little in common between the cases \cite{2003math......5124B,Ashmore:2021pdm}. 
The spinorial description thus provides a pleasing general structure to these results, as well as suggesting strong connections between these constructions and supersymmetry. 
It remains an interesting problem in algebra to find an elegant classification of the groups and subgroups for which the \complexes are spinor type in this sense, and to extract possible connections to the theory of pure spinors. 

We also explored how, in cases where the \complex has a BV symplectic pairing, one can write a quadratic BV actions associated to it. The classical on-shell states of these theories are those for which the gauge field is closed in the differential on the \complexNoSpace, and thus are instantons in the relevant sense. 
These actions reproduced the linearised versions of many Chern--Simons gauge theories that have been of interest over the years. 
One could thus view the construction as a way to construct interesting gauge theories of instantons associated to $G$-structures, using only their algebraic data, as the construction of the \complex is essentially a purely algebraic prescription. 
We also briefly noted that for cases which do not have a BV symplectic pairing of this type, there are other constructions of actions that one could perform. For example, for structure groups $\SU(2)$ in four dimensions or $\Spin(7)$ or $\SU(4)$ in eight dimensions one could proceed as in~\eqref{eq:add-dual-complex} to produce a new complex equipped with a symplectic pairing and then write an action.

For the generalised $G \subset O(d,d+n)$ structures appearing in heterotic and type II supersymmetric backgrounds, we have also found that the \complex possess significant additional structure. In particular, whenever the structure group 
has the form $G_+ \times G_-$ such that one can associate $G_\pm$ with the left- and right-moving modes of the string, the \complex becomes a double complex, which is the tensor product of the left- and right-moving gauge theory \complexesNoSpace. Further, we have seen that these satisfy K\"ahler identities.

In the cases of heterotic structures $G \times \SO(d+n)$, we were able to relate these double complexes fairly directly to prior works on heterotic moduli~\cite{delaOssa:2014cia,delaOssa:2017pqy}, and in spinor type cases, key properties of the complex such as the equality of the Laplacian $\Delta_+$ and $\Delta_-$ on the first row of the \complex could be expressed as the closure of the supersymmetry algebra on the fermion equations of motion. 
In the case of $\SU(3)\times\SO(6+n)$ in six-dimensions, ot three complex dimensions, we were able to construct a mild extension of the \complex to describe the BRST-BV complex of the superpotential theory, and this equality of Laplacians was crucial for consistency. Moreover, without prior knowledge of the superpotential theory, the extended \complex could be motivated by wishing for the existence of a BV symplectic pairing, such that one could write an associated BV action. In this way, one could have reconstructed the superpotential theory at quadratic order starting from the \complexNoSpace. This provides, admittedly with considerable hindsight, another example of the use of this formalism to construct interesting action functionals, also in complex dimensions other than three.  

Perhaps the most elegant of our examples, though, are the spinor type complexes in $O(d,d)$ geometries relevant to type II strings. These were shown to satisfy not only the K\"ahler identities but also the equality of the left- and right-moving Laplacians, which further gives rise to a $\der\bar\der$-lemma. In many ways, these \complexes thus behave much like the double complex of $(p,q)$-forms on a K\"ahler manifold. The equality of the Laplacians $\Delta_\pm$ and their relation to the ($H$-twisted) de Rham Laplacian can further be interpreted as target space artefacts of the level matching and mass-shell conditions on the worldsheet. Examples of these spinor type complexes were shown in previous work~\cite{Ashmore:2021pdm} to describe the physical operators on the worldsheets of (quasi-)topological strings. Our treatment here provides elegant general proofs of the properties of these double complexes that were previously derived by direct calculation in the specific cases considered. Further, we have provided a more general structure to the target space actions associated to these theories as a kind of Aeppli-Chern--Simons theory in which the the diagonal complex of the double complex provides the ghost number zero field content, with the second order differential $\dd_+ \dd_-$ providing the kinetic operators. The ghost number zero physical states thus become the Aeppli cohomologies of the double complex, while the gauge structure involves complexes of a similar nature to those appearing in~\cite{Tseng:2009gr,tseng2014generalized}. We aim to explore the further general properties of these theories in future work. 

There are many other directions in which one could hope to extend and apply this construction. For example, one could consider weakening the conditions on the torsion of the underlying $G$-structure. Throughout this article we have taken our $G$-structures to be torsion-free, which corresponds to supersymmetric Minkowski vacua (on the external space part of a compactification) in the gravitational cases. However, it is known that in Carri\'on's original construction, only some of the intrinsic torsion classes are required to vanish for the complex to exist. For example, for $U(N)$ structures in $2N$ dimensions, one only requires that the manifold be complex (rather than K\"ahler) for the usual Dolbeault complex to exist. 
One could explore if, in particular, one could allow a constant singlet torsion, with accompanying non-zero scalar curvature, to define similar structures on supersymmetric AdS vacua. 
It may also be possible to test for the existence of torsion-free structures by comparing Bott-Chern or Aeppli type cohomologies against the full cohomology of the complex as one does in usual K\"ahler geometry. It would also be interesting to see if one can connect the cohomology of these complexes to sheaf cohomologies via Poincar\'e lemmas.

One could also consider corresponding statements in  exceptional generalised geometries~\cite{PiresPacheco:2008qik,Hull:2007zu,CSW2} describing the internal sectors of eleven-dimensional supergravity and type II theories including RR fluxes. Here, there is a clear picture of how to proceed. 
In the $O(d,d+n)$ geometries, the graded vector space $\Lambda^\bullet(E)$ which we start with has a physical interpretation as the tensor hierarchy~\cite{deWit:2005hv,deWit:2008ta} of the supergravity theory, which can be seen as the analogue of the de Rham complex in generalised geometry~\cite{Hohm:2015xna,Wang:2015hca}. 
Notice that in $O(d,d)$ geometry, there is an alternative analogue, which is given by the weighted spinors that can be represented as polyforms on the manifold. Our spinor type complexes thus provide a concrete relation between these two objects in the case of $\cN=2$ backgrounds. 
The bundles corresponding to the tensor hierarchy are known for exceptional geometries~\cite{Palmkvist:2013vya}. While the algebraic product on this is more complicated, one can still use it to generate a subspace from the Lie algebra of the structure group $\mf{g} \subset E \otimes E^*$. The resulting quotient will then contain the full tower of ghosts for the generalised diffeomorphism symmetry, the infinitesimal deformations of a $G$-frame and the intrinsic torsion space so that it will give the infinitesimal moduli of the structure as in section~\ref{sec:moduli} of the present article. 
An analogue of the superpotential complex can also be constructed for the $J$ structure in the $E_{7(7)}$ case, as one would expect. 
We hope to provide details of this construction in future work. 
There are substantial complications in applying the methods that we use here in the exceptional context. In particular, the absence of a generalised Riemann tensor would seem to obstruct proofs of the types we have employed in section~\ref{sec:gen-Dolbeault} (though see~\cite{Hassler:2023axp} for some recent ideas). The construction of these complexes in exceptional geometry may therefore shed light on how to construct such Riemann tensors.  
It would also be curious to examine whether there is an analogue of spinor type complexes in exceptional geometry, producing an analogue of spinorial representations of the exceptional groups. Unlike in the cases we examine here, these representations may be infinite-dimensional.

One can also view the double complexes appearing in the $O(d,d)$ cases as double copy constructions, where a gravitational theory can be seen as the product of two gauge theories (see~\cite{Bern:2019prr,White:2024pve} and references therein).
In particular, these would be examples where the background geometry is non-trivial. The \complex is the tensor product of two gauge theory \complexesNoSpace, one for each of the left- and right-moving string sectors, with differentials constructed from the corresponding generalised connections. The double copy one finds is perturbative around a supersymmetric supergravity background equipped with two gauge theory instanton solutions. The cohomologies representing the on-shell deformations of these instantons have harmonic representatives in the Laplacians $\Delta_\pm$. As these Laplacians are equal, one can then simply take the tensor product of two such deformations to get a harmonic deformation of the supergravity background. This is precisely a classical double copy relation. The double copy has been applied to topological gauge theories previously in~\cite{Ben-Shahar:2021zww,Bonezzi:2024dlv} and to four-dimensional instantons in~\cite{Alawadhi:2021uie,Armstrong-Williams:2022apo}. Our construction provides a linearised version for a wide class of theories of instantons.

While in this work we have considered only the linearised theories, corresponding to the infinitesimal deformations of the underlying instanton solutions, one could also wish to study the non-linear deformation theory, or correspondingly the interacting field theories. 
For the Chern--Simons gauge theories and the superpotential theories~\cite{Ashmore:2018ybe}, the interactions are known. It would be curious to see whether there is an elegant systematic construction of them extending our work here. As a first step in this direction, it would be interesting to compute the index of the corresponding deformation complexes, counting the expected, or ``virtual", dimension of the moduli space. This is of particular interest for $\cN=1$ backgrounds, where one might expect a vanishing index, and thus a zero-dimensional virtual moduli space. It is tempting to speculate whether analogs of Donaldson-Thomas invariants \cite{Donaldson:1996kp, thomas1997gauge} can be defined in these cases, and if the index has something to say about the true nature of the physical moduli space, and indeed the string theory moduli problem, when all higher order and non-perturbative corrections are included. 

We also expect that the BPS complex will be useful in the study quantum aspects of the corresponding (quasi-)topological theories. For example, anomalies are often phrased in terms of curvature polynomials on a ``universal geometry",\footnote{In the Universal geometry picture, one thinks of the manifold, together with the geometric structure of interest, as a fibration over the moduli space of said geometric structure.} as for example with the holomorphic anomaly of Kodaira-Spencer theory \cite{Bershadsky:1993cx, Bershadsky:1993ta}. The Universal geometry picture was first considered by Atiyah-Singer in \cite{Atiyah:1984tf}, and further developed in \cite{Donaldson:1990kn, donaldson1987orientation, Witten:1988ze, Harvey:1991hq} in the study of Donaldson theory and Donaldson invariants. A Universal geometry picture has also been developed for six-dimensional heterotic geometries \cite{Candelas:2016usb, Candelas:2018lib, McOrist:2019mxh}, where it was observed that the universal geometry has many features mimicking that of the underlying geometric structure. We hence expect that the technology presented here will be useful for in pursuing these ideas in more generality, e.g. in defining analogs of Donaldson invariants, and we hope to explore these ideas further in the future.

%%%%%%%%%%%%%%%%%%%%%%%%%%%%%%%%%%%%%%%%%%%%%%%%%%%%%%%%%%%%%%%%%%%%%%%%%%

\acknowledgments

We are grateful to Anthony Ashmore, Leron Borsten, Xenia de la Ossa, Magdalena Larfors, Matthew Magill, Pavol \v{S}evera and Dan Waldram for helpful discussions. C.S.-C.~and F.V.~are supported by an EPSRC New Investigator Award, grant number EP/X014959/1. No new data was collected or generated during the course of this research. D.T.~is supported by the NSF grant PHYS-2112859.
C.S.-C., D.T.~and F.V.\ acknowledge hospitality from Simons Center for Geometry and Physics during the programme ``Supergravity, Generalized Geometry and Ricci Flow" at which some of this research was performed. 
C.S.-C.~and E.S.~acknowledge hospitality from the MATRIX institute programme ``New Deformations of Quantum Field and Gravity Theories" at which some of this research was performed. 

%%%%%%%%%%%%%%%%%%%%%%%%%%%%%%%%%%%%%%%%%%%%%%%%%%%%%%%%%%%%%%%%%%%%%%%%%%

\appendix

\section{Gamma matrix decompositions}\label{app:Gamma}

Here we provide some details of the decomposition of the $\Cliff(d,d;\bbR)$ gamma matrices $\Gamma^A$ in terms of two sets of generators $\gamma^a$ and $\gamma^{\ba}$ for $\Cliff(d,\bbR)$, thought of as attached to the bundles $C_\pm$. As above, we have $A,B = 1,\dots,2d$ is an $O(d,d)$ vector index and $a,b = 1,\dots,d$ and $\ba,\bb = 1,\dots, d$ are vector indices for $C_\pm$. We also write that $\gamma^{(d)} = \gamma^1 \dots \gamma^d$.

For $d$ odd, we can decompose $\Gamma^A$ as:
\begin{equation}
\label{eq:gamma-odd-decomp}
	\Gamma^{a} = \gamma^a \otimes \id \otimes \sigma^1
	\hs{30pt}
	\Gamma^{\ba} = \id \otimes \gamma^{\ba} \otimes \ii\sigma^2
\end{equation}
where $\sigma^i$ are the Pauli matrices and we take the matrices $\gamma^a$ and $\gamma^{\ba}$ to generate irreducible representations of $\Cliff(d,\bbR)$ for odd $d$ so that $\gamma^{(d)}$ is proportional to the identity. The $\Spin(d,d)$ spinor thus decomposes as 
\begin{equation}
	S(E) = S(C_+) \otimes S(C_-) \otimes \bbR^2
	\hs{30pt}
	\Psi = \sum \, \zeta_+ \otimes \zeta_- \otimes t
\end{equation}
where $\zeta_\pm \in S(C_\pm)$ and $t$ is an auxiliary vector in $\bbR^2$. This auxiliary vector is necessary to account for the fact that the Clifford algebra representations are real spaces, and one is really taking the tensor products over $\bbR$, despite that they are often expressed in terms of complex components. This means that, for example, one must think of $\bbC$ as a subalgebra of $2\times 2$ real matrices. From~\eqref{eq:gamma-odd-decomp} One can then easily see that
\begin{equation}
	\Gamma^{a} \Gamma^b = (\gamma^a \gamma^b) \otimes \id \otimes \id
	\hs{30pt}
	\Gamma^{\ba} \Gamma^{\bb} = - \id \otimes (\gamma^{\ba}  \gamma^{\bb}) \otimes \id
\end{equation}

For $d=4n$, we have $(\gamma^{(d)}){}^2 = +\id$ so that we can decompose $\Gamma^A$ as:
\begin{equation}
	\Gamma^{a} = \gamma^a \otimes \id
	\hs{30pt}
	\Gamma^{\ba} = \gamma^{(d)} \otimes \gamma^{\ba}\gamma^{(d)}
\end{equation}
so that the $\Spin(d,d)$ spinor decomposes as 
\begin{equation}
	S(E) = S(C_+) \otimes S(C_-)
	\hs{30pt}
	\Psi = \sum \, \zeta_+ \otimes \zeta_- 
\end{equation}
where $\zeta_\pm \in S(C_\pm)$. Again one has
\begin{equation}
	\Gamma^{a} \Gamma^b = (\gamma^a \gamma^b) \otimes \id 
	\hs{30pt}
	\Gamma^{\ba} \Gamma^{\bb} = - \id \otimes (\gamma^{\ba}  \gamma^{\bb}) 
\end{equation}

For $d=4n+2$, we have $(\gamma^{(d)}){}^2 = -\id$ so that we can decompose $\Gamma^A$ as:
\begin{equation}
	\Gamma^{a} = \gamma^a \otimes \id
	\hs{30pt}
	\Gamma^{\ba} = \gamma^{(d)} \otimes \gamma^{\ba}
\end{equation}
so that the $\Spin(d,d)$ spinor decomposes as 
\begin{equation}
	S(E) = S(C_+) \otimes S(C_-) 
	\hs{30pt}
	\Psi = \sum \, \zeta_+ \otimes \zeta_-
\end{equation}
where $\zeta_\pm \in S(C_\pm)$ and again
\begin{equation}
	\Gamma^{a} \Gamma^b = (\gamma^a \gamma^b) \otimes \id 
	\hs{30pt}
	\Gamma^{\ba} \Gamma^{\bb} = - \id \otimes (\gamma^{\ba}  \gamma^{\bb}) 
\end{equation}

We have thus established equations~\eqref{eq:Gamma-2-decomp} in all cases.

\bibliographystyle{JHEP}
\bibliography{citations}

%%%%%%%%%%%%%%%%%%%%%%%%%%%%%%%%%%%%%%%%%%%%%%%%%%%%%%%%%%%%%%%%%%%%%%%%%%
\end{document}